\documentclass[11pt,letterpaper]{article}
\usepackage[utf8]{inputenc}
\usepackage[T1]{fontenc}
\usepackage{amsmath,amsthm,amssymb}
\usepackage[ruled]{algorithm2e}
\usepackage{booktabs}
\usepackage{hyperref}
\usepackage{cleveref}
\usepackage{thmtools,thm-restate}
\usepackage{fullpage}
\usepackage{amssymb}
\usepackage{mathtools}
\usepackage{xcolor}
\usepackage{mlmodern}
\usepackage[authoryear]{natbib}
\allowdisplaybreaks
\usepackage{braket}
\usepackage{tikz}
\usetikzlibrary{positioning,shapes,calc,shapes.geometric,intersections,patterns}

\usepackage{latex_abbreviations}
\allowdisplaybreaks
\DeclareMathOperator*{\argmax}{argmax}
\DeclareMathOperator*{\argmin}{argmin}
\newcommand{\E}{\mathop{\mathbb{E}}}

\newcommand{\bfone}{\mathbf{1}}
\newcommand{\tmax}{{t_{\mathrm{max}}}}
\newcommand{\PiProd}{\Pi_{\mathsf{Prod}}}
\newcommand{\PiBNE}{\Pi_{\mathsf{BNE}}}
\newcommand{\PiBS}{\Pi_{\mathsf{BS}}}
\newcommand{\PiSF}{\Pi_{\mathsf{SF}}}
\newcommand{\PiANF}{\Pi_{\mathsf{ANF}}}
\newcommand{\PiCom}{\Pi_{\mathsf{Com}}}
\newcommand{\SigmaSF}{\Sigma_{\mathsf{SF}}}
\newcommand{\SigmaANF}{\Sigma_{\mathsf{ANF}}}
\newcommand{\SigmaProd}{\Sigma_{\mathsf{Prod}}}
\newcommand{\SigmaBNE}{\Sigma_{\mathsf{BNE}}}
\newcommand{\PoA}{{\mathsf{PoA}}}
\newcommand{\vSW}{v_{\mathsf{SW}}}
\newcommand{\phiSF}{\phi_{\mathsf{SF}}}
\newcommand{\vEP}{v_{\mathrm{EP}}}

\newcommand{\RLS}{R_{\mathsf{LS}}}
\newcommand{\RUSi}{R_{\mathsf{US},i}}
\newcommand{\RTSi}{R_{\mathsf{TS},i}}
\newcommand{\RSSi}{R_{\mathsf{SS},i}}
\newcommand{\RUSj}{R_{\mathsf{US},j}}
\newcommand{\RTSj}{R_{\mathsf{TS},j}}
\newcommand{\RSSj}{R_{\mathsf{SS},j}}

\newtheorem{theorem}{Theorem}[section]
\theoremstyle{plain}
\newtheorem{lemma}[theorem]{Lemma}

\newtheorem{proposition}[theorem]{Proposition}
\theoremstyle{definition}
\newtheorem{definition}[theorem]{Definition}
\newtheorem{remark}[theorem]{Remark}
\newtheorem{example}[theorem]{Example}

\title{Bayes correlated equilibria, no-regret dynamics in Bayesian games, and the price of anarchy}
\author{Kaito Fujii\footnote{National Institute of Informatics. Email: \texttt{fujiik@nii.ac.jp}}}

\begin{document}

\maketitle

\begin{abstract}%
This paper investigates equilibrium computation and the price of anarchy for Bayesian games, which are the fundamental models of games with incomplete information. In normal-form games with complete information, it is known that efficiently computable no-regret dynamics converge to correlated equilibria, and the price of anarchy for correlated equilibria can be bounded for a broad class of games called smooth games. However, in Bayesian games, as surveyed by Forges (1993), several non-equivalent extensions of correlated equilibria exist, and it remains unclear whether they can be efficiently computed or whether their price of anarchy can be bounded.

In this paper, we identify a natural extension of correlated equilibria that can be computed efficiently and is guaranteed to have bounds on the price of anarchy in various games. First, we propose a variant of regret called untruthful swap regret. If each player minimizes it in repeated play of Bayesian games, the empirical distribution of these dynamics is guaranteed to converge to communication equilibria, which is one of the extensions of correlated equilibria proposed by Myerson (1982). We present an efficient algorithm for minimizing untruthful swap regret with a sublinear upper bound, which we prove to be tight in terms of the number of types. As a result, by simulating the dynamics with our algorithm, we can approximately compute a communication equilibrium in polynomial time. Furthermore, we extend existing lower bounds on the price of anarchy based on the smoothness arguments from Bayes--Nash equilibria to equilibria obtained by the proposed dynamics.
\end{abstract}
\clearpage

\tableofcontents
\clearpage

\section{Introduction}
In normal-form games, computing a Nash equilibrium is known to be PPAD-complete even if for two-player games \citep*{CDT09}, and it requires exponentially many queries in the number of players \citep*{Babichenko16}.
In contrast, \textit{correlated equilibria} \citep*{Aum74} can be efficiently computed even for multi-player non-zero-sum games.
A correlated equilibrium can be interpreted as an equilibrium concept realized by a \textit{mediator} who can recommend (possibly correlated) actions to players.
A standard approach to computing a correlated equilibrium is to simulate a variant of \textit{no-regret dynamics},
in which players repeatedly play the same game many times and independently decide their own action in each round.
If each player decides their action according to an algorithm that exhibits no internal regret \citep*{FV97,HM00} or swap regret \citep*{BM07}, the dynamics are known to converge to correlated equilibria.%
\footnote{Formally, the empirical distribution of action profiles in all past rounds converges to the set of correlated equilibria.
In this paper, we say ``dynamics converge'' in this sense.
}

Another paramount challenge in algorithmic game theory is to bound the \textit{price of anarchy} (PoA) \citep*{KP99,BHLR08} in a broad class of games.
As each individual seeks their own interests, players in non-cooperative games can fall into a bad equilibrium.
To measure how bad equilibria could be, \citet*{KP99} proposed the price of anarchy, which is defined as the ratio of the social welfare in the worst equilibrium to the maximum social welfare.
\citet*{Roughgarden15} proved that if a game satisfies a property called \textit{smoothness}, the PoA of (coarse) correlated equilibria can be bounded,
and various classes of games satisfy smoothness.
\citet*{ST13} later developed a variant of smoothness for mechanism design and provided lower bounds on the PoA of correlated equilibria for a large class of mechanisms.

As discussed above, in normal-form games with \textit{complete information}, correlated equilibria can be efficiently computed by simulating dynamics and are guaranteed to have PoA bounds in various games.
Thus, we can compute an equilibrium efficiently in a distributed fashion and then realize it by introducing a mediator, which enables players to settle into an equilibrium with high social welfare.
In this paper, we focus on \textit{Bayesian games},
in which each player $i \in [n]$ has private information represented by a random variable $\theta_i \in \Theta_i$ called a \textit{type} jointly generated from a commonly known prior distribution $\rho \in \Delta(\Theta_1 \times \dots \times \Theta_n)$.
Each player $i \in [n]$ chooses their action $a_i \in A_i$ based on their own type $\theta_i$, and then each player $i \in [n]$ obtains the payoff $v_i(\theta; a)$ depending not only on the actions $a = (a_1,\dots,a_n)$ but also on the realized types $\theta = (\theta_1,\dots,\theta_n)$.
Since their introduction by \citet*{Harsanyi67,Harsanyi68a,Harsanyi68b}, Bayesian games have been studied as one of the most significant models in game theory.
The central question of this paper is as follows:
\begin{center}
\textit{Is there any equilibrium concept for Bayesian games that can be realized by a mediator, efficiently computable by simulating dynamics, and guaranteed to have PoA bounds for various games?}
\end{center}

Since correlated equilibria satisfy these properties for games with complete information, a natural approach to this question is to extend correlated equilibria to Bayesian games.
However, as surveyed by \citet*{Forges93,Forges14}, there are various non-equivalent extensions of correlated equilibria, which are collectively called \textit{Bayes correlated equilibria} (\Cref{fig:relation}).
For example, \textit{communication equilibria} are naturally realized by introducing a mediator who can bidirectionally communicate with players \citep*{Myerson82,Forges86}.
Other extensions include \textit{strategic-form correlated equilibria} (SFCEs) and \textit{agent-normal-form correlated equilibria} (ANFCEs), which are defined as correlated equilibria of complete-information interpretations of Bayesian games called the strategic form and agent normal form, respectively.

\begin{figure}[t]
\centering
\begin{tikzpicture}[line width=0.003\paperwidth, yscale=0.5]
\definecolor{darkred}{HTML}{990000}
\definecolor{darkblue}{HTML}{0D47A1}
\newlength{\ul}
\setlength{\ul}{0.01\paperwidth}
\draw [pattern=north east lines] (0\ul,-12.2\ul) arc (-77.3:+77.3:{27.5\ul} and {12.5\ul})
-- (0\ul,+12.2\ul) arc (+102.7:+257.3:{27.5\ul} and {12.5\ul}) -- cycle;
\node (ParCE) [draw=black, ellipse, inner sep=0, minimum width=70\ul, minimum height=18\ul] {};
\node (SFCE) [draw=black, ellipse, inner sep=0, minimum width=40\ul, minimum height=9\ul] at (ParCE) {};
\node (BNE) [draw=black, ellipse, inner sep=0, anchor=south, minimum width=30\ul, minimum height=5\ul] at ([yshift=1\ul]SFCE.south) {};
\node (ANFCE) [draw=black, fill=darkblue!50, fill opacity=0.3, ellipse, inner sep=0, minimum width=55\ul, minimum height=12.5\ul] at ([xshift=-6\ul]ParCE) {};
\node (ComCE) [draw=black, fill=darkred!50, fill opacity=0.3, ellipse, inner sep=0, minimum width=55\ul, minimum height=12.5\ul] at ([xshift=+6\ul]ParCE) {};

\node (BNEc) [align=center] at (BNE) {Bayes--Nash\\equilibria};
\node (SFCEc) [align=center, anchor=north, inner sep=2pt] at ([yshift=-2pt]SFCE.north) {Strategic-form CE};
\node (ComCEc) [align=right, anchor=east, text=darkred] at ([xshift=-20pt]ComCE.east) {Communi-\\[-2pt]cation\\equilibria};
\node (ANFCEc) [align=center, anchor=west, text=darkblue] at ([xshift=2pt]ANFCE.west) {Agent-normal-\\form CE};
\node (ParCEc) [align=left, anchor=north, inner sep=2pt] at ([yshift=-0.8\ul]ParCE.north) {Bayesian solution};

\end{tikzpicture}
\caption{%
Relations among classes of Bayes correlated equilibria.
The shaded region represents the intersection of communication equilibria and ANFCEs (equivalently, communication equilibria with strategy representability), which this paper focuses on.
}\label{fig:relation}
\end{figure}

While their definitions have been well studied, little is known about computational tractability of Bayes correlated equilibria.
By simulating a variant of no-regret dynamics proposed by \citet*{HST15}, we can compute a class of Bayes \textit{coarse} correlated equilibria, which is broader than Bayes correlated equilibria.%
\footnote{%
Formally, as Bayes correlated equilibria, Bayes coarse correlated equilibria have many classes.
The dynamics proposed by \citet*{HST15} converge to agent-normal-form coarse correlated equilibria, which are a superset of ANFCEs but not a superset of communication equilibria.
See \Cref{sec:bcce} for formal definitions of Bayes coarse correlated equilibria.
}
On the hardness side, since Bayes--Nash equilibria generalize Nash equilibria in complete-information games, hardness of Bayes--Nash equilibria immediately follows that of Nash equilibria, which is extended even for two players with a constant number of actions \citep*{Rubinstein18}.
Whether it is possible to efficiently compute Bayes correlated equilibria is a problem lying between them.

Although much effort has been devoted to PoA bounds for Bayesian games, most studies focused on Bayes--Nash equilibria \citep*{CKS16,PT10,Roughgarden15incomplete,Syrgkanis12,ST13,HHT14,FFGL20,JL22}.
\citet*{Roughgarden15incomplete} and \citet*{Syrgkanis12} derived bounds on the PoA of Bayes--Nash equilibria by using the smoothness arguments, but it remains to be seen whether these bounds can be extended to a broader class of equilibria.

In this paper, we present \textit{the intersection of communication equilibria and ANFCEs}%
\footnote{
This concept is equivalently defined as communication equilibria with \textit{strategy representability} (see \Cref{sec:sr} and \Cref{prop:com-sr} for details), which can be naturally interpreted as an equilibrium where communication with each player can be conducted separately in any order.
}
as an equilibrium concept that is efficiently computable and has PoA bounds for various games.
Our contributions are summarized as follows.
\begin{itemize}
\item This paper shows that dynamics minimizing a variant of swap regret, which we call \textit{untruthful swap regret}, converge to this class (\Cref{sec:concept}).
\item We propose an efficient algorithm for minimizing the untruthful swap regret with a sublinear upper bound (\Cref{sec:ub}).
\item We show that this upper bound is tight in terms of the number of types by providing a problem instance for which no algorithm can achieve better untruthful swap regret (\Cref{sec:lb}).
\item We show that most existing PoA bounds for Bayes--Nash equilibria based on the smoothness arguments can be extended to this intersection class (\Cref{sec:poa}).
\end{itemize}

\paragraph{Other related work}
Extensive-form games with imperfect information can represent Bayesian games by treating types as actions of a chance node.
\citet*{FCMG22} proposed dynamics converging to extensive-form correlated equilibria (EFCEs) in this class of games by using the $\Phi$-regret framework.
However, while EFCEs generalize ANFCEs, they do not generalize communication equilibria, and their regret definition extends internal regret rather than swap regret.
To study the Stackelberg value, \citet*{MMSS22} considered \textit{linear swap regret} (originally called \textit{linear regret} by \citet*{Gor08}) in repeated play of Bayesian games with two players, but they did not provide algorithms with regret guarantees.
We prove the equivalence of linear swap regret in Bayesian games and untruthful swap regret in \Cref{sec:linear-swap}.
In Appendix D.2.2 of their paper, \citet*{MMSS22} also mentioned an equilibrium concept similar to communication equilibria for two-player Bayesian games, in which \textit{a single player} cannot benefit from any type reporting \textit{or} action deviations but not their combination.
Since the initial appearance of this paper, several related studies have been published.
\citet*{Farina23} and \citet*{Zhang24} proposed efficient algorithms for minimizing linear swap regret in extensive-form games, which generalize untruthful swap regret in Bayesian games, but the order of their upper bounds is suboptimal for untruthful swap regret.
\citet*{PR24} and \citet*{DDFG24} proposed an algorithm for swap regret minimization with exponentially many actions.
As a special case, these algorithms achieve $\displaystyle O(T / \log T)$ strategy swap regret (defined in \Cref{sec:sf}) and yield an algorithm computing an SFCE in time exponential in $1/\epsilon$, where $\epsilon$ is an additive error of the incentive constraints.
It remains to be seen whether there exists an algorithm for SFCEs that runs in time polynomial also in $1/\epsilon$.
\citet{Fujii25} analyzed the PoA and the price of stability (PoS) for Bayes correlated equilibria in Bayesian games with submodular social welfare and showed that there is a separation in the PoA and PoS among the variants of Bayes correlated equilibria.
In \Cref{sec:related-work}, we survey further related work and connections of our results to Bayesian incentive compatible mechanisms or information design.

\section{Communication equilibria and untruthful swap regret}\label{sec:concept}
Our first goal is to define untruthful swap regret and prove that dynamics minimizing it converge to the intersection of communication equilibria and ANFCEs.
See \Cref{sec:pre} for detailed notations for Bayesian games and \Cref{sec:dynamics} for various definitions of Bayes correlated equilibria.

Note that dynamics converging to ANFCEs can be easily obtained by extending no-swap-regret dynamics to a complete-information game known as the agent normal form, where a single player with different types is hypothetically treated as distinct players (see \Cref{sec:anf} for ANFCEs and corresponding dynamics).
Therefore, our challenge is to develop a stronger notion of swap regret that guarantees the convergence to communication equilibria.

A communication equilibrium in an $n$-player Bayesian game is represented by a type-wise distribution $\pi \in \Delta(A)^\Theta$, where $A \coloneqq A_1 \times \dots \times A_n$ is the set of all action profiles and $\Theta \coloneqq \Theta_1 \times \dots \times \Theta_n$ is the set of all type profiles.
A mediator first gathers type $\theta \in \Theta$ from players and then privately recommends an action profile $a \in A$ generated from the distribution $\pi(\theta)$.
The incentive constraints for each player $i \in N \coloneqq [n]$ are specified by two different kinds of deviations:
(1) reporting an untruthful type $\psi(\theta_i) \in \Theta_i$ instead of the true type $\theta_i \in \Theta_i$ according to $\psi \colon \Theta_i \to \Theta_i$ and (2) taking an action $\phi(\theta_i, a_i)$ instead of the recommended action $a_i$ when their type is $\theta_i$ according to $\phi \colon \Theta_i \times A_i \to A_i$.
If each player cannot benefit from these deviations, a distribution $\pi \in \Delta(A)^\Theta$ is called a communication equilibrium.
\begin{definition}[$\epsilon$-Approximate communication equilibria]
For any $\epsilon \ge 0$, a type-wise distribution $\pi \in \Delta(A)^\Theta$ is an $\epsilon$-approximate communication equilibrium
if
for any $i \in N$, $\psi \colon \Theta_i \to \Theta_i$, and $\phi \colon \Theta_i \times A_i \to A_i$, it holds that
\begin{equation}\label{eq:ic-com}
\tag{$\mathrm{IC}_{\mathrm{Com}}$}
\E_{\theta \sim \rho} \left[ \E_{a \sim \pi(\theta)} \left[ v_i(\theta; a) \right] \right]
\ge
\E_{\theta \sim \rho} \left[ \E_{a \sim \pi(\psi(\theta_i),\theta_{-i})} \left[ v_i(\theta; \phi(\theta_i, a_i), a_{-i}) \right] \right] - \epsilon.
\end{equation}
Let $\PiCom^\epsilon \subseteq \Delta(A)^\Theta$ be the set of all $\epsilon$-approximate communication equilibria.
\end{definition}

To deal with these two kinds of deviations, we introduce a new variant of regret called untruthful swap regret.
In repeated play of Bayesian games, each player faces an online learning problem with rewards determined by a stochastic type, which we call \textit{online learning with stochastic types}.
In each round $t \in [T]$ of this problem, each player $i \in [n]$ decides a distribution $\pi_i^t \in \Delta(A_i)^{\Theta_i}$ that determines a (randomized) action for each type and receives reward%
\footnote{We use the term ``reward'' for online learning problems and distinguish it from the term ``payoff'' for games.}
$u_i^t(\theta_i,a_i)$ depending on their realized type $\theta_i \sim \rho_i$ and action $a_i \sim \pi_i^t(\theta_i)$, where $u_i^t \in [0,1]^{\Theta_i \times A_i}$ is the reward vector defined by the other players' distributions and $\rho_i \in \Delta(\Theta_i)$ is the marginal distribution of $\theta_i$.
Formally, we define the reward vector $u_i^t \in [0,1]^{\Theta_i \times A_i}$ by
\begin{equation*}\label{eq:reward-pi}
u_i^t(\theta_i, a_i) = \E_{\theta_{-i} \sim \rho|\theta_i} \left[ \E_{a_{-i} \sim \pi^t_{-i}(\theta_{-i})} \left[ v_i( \theta; a) \right] \right]
\end{equation*}
for each type $\theta_i \in \Theta_i$ and action $a_i \in A_i$, where
$\rho|\theta_i$ is the distribution of $\theta_{-i}$ conditioned on $\theta_i \in \Theta_i$
and $\pi_{-i}^t(\theta_{-i})$ is the product distribution that independently generates $a_j \sim \pi_j^t(\theta_j)$ for each $j \in N \setminus \{i\}$.
For this online learning problem with stochastic types, the untruthful swap regret $\RUSi^T$ for player $i \in [n]$ is defined as follows.
\begin{definition}[Untruthful swap regret]\label{def:untruthful}
For online learning with stochastic types specified by actions $A_i$, types $\Theta_i$, prior distribution $\rho_i$, and reward vector $u_i^t \in [0,1]^{\Theta_i \times A_i}$ for every round $t \in [T]$, untruthful swap regret is defined as
\begin{equation*}
\RUSi^T = \max_{\psi \colon \Theta_i \to \Theta_i} \max_{\phi \colon \Theta_i \times A_i \to A_i} \sum_{t=1}^T \E_{\theta_i \sim \rho_i} \left[ \E_{a_i \sim \pi_i^t(\psi(\theta_i))} \left[ u_i^t(\theta_i, \phi(\theta_i, a_i)) \right] - \E_{a_i \sim \pi_i^t(\theta_i)} \left[ u_i^t(\theta_i, a_i) \right] \right].
\end{equation*}
\end{definition}

\begin{algorithm}[t]
\caption{Dynamics for the intersection of communication equilibria and ANFCEs}\label{alg:com-dynamics}
	For each $i \in N$, let $\calA_i$ be a subroutine that minimizes untruthful swap regret for online learning with stochastic types.\;
	\For{each round $t = 1,\dots,T$}{
		Each player $i \in N$ decides their randomized strategy $\pi^t_{i} \in \Delta(A_i)^{\Theta_i}$ according to $\calA_{i}$ and shares it with the other players.\;
		Each player $i \in N$ computes reward $u_i^t(\theta_i, a_i) = \E_{\theta_{-i} \sim \rho|\theta_i} \left[ \E_{a_{-i} \sim \pi^t_{-i}(\theta_{-i})} \left[ v_i(\theta; a) \right] \right]$ for every $\theta_i \in \Theta_i$ and $a_i \in A_i$, where $\pi^t_{-i}(\theta_{-i}) \in \Delta(A_{-i})$ is the product distribution that independently generates $a_j \sim \pi_j^t(\theta_j)$ for each $j \in N \setminus \{i\}$.\;
		Feed $u_i^t \in [0,1]^{\Theta_i \times A_i}$ to each $\calA_{i}$ as a reward vector for round $t$.\;
	}
\end{algorithm}

The resulting dynamics are described in \Cref{alg:com-dynamics}, where each player $i \in N$ employs an algorithm $\calA_i$ to make decisions.
Let $\PiANF^\epsilon \subseteq \Delta(A)^\Theta$ denote the set of all $\epsilon$-approximate ANFCEs mapped to $\Delta(A)^\Theta$ (see \Cref{sec:anf} for the formal definition).
The following theorem (proved in \Cref{sec:dynamics-proof}) states that if every player achieves sublinear untruthful swap regret, then the dynamics converge to communication equilibria (and also to ANFCEs).
For readability, we write $\pi(\theta; a)$ and $\pi^t_i(\theta_i; a_i)$ instead of $\pi(\theta)(a)$ and $\pi^t_i(\theta_i)(a_i)$, respectively.

\begin{restatable}{theorem}{thmdynamicscom}
\label{thm:dynamics-com}
Let $\pi_{i}^t \in \Delta(A_i)^{\Theta_i}$ be the type-wise distribution of each player $i \in N$ for each round $t \in [T]$ in \Cref{alg:com-dynamics}.
If $\pi \in \Delta(A)^{\Theta}$ is the empirical distribution defined by $\pi(\theta;a) = \frac{1}{T} \sum_{t=1}^T \prod_{i \in N} \pi_{i}^t(\theta_i;a_i)$ for each $\theta \in \Theta$ and $a \in A$,
then $\pi \in \PiCom^\epsilon \cap \PiANF^\epsilon$ holds with $\epsilon = \frac{\max_{i \in N} \RUSi^T}{T}$,
where $\RUSi^T$ is the untruthful swap regret for each subroutine $\calA_i$.
\end{restatable}

\section{Algorithm for minimizing untruthful swap regret}\label{sec:ub}
In this section, we describe an efficient algorithm for minimizing untruthful swap regret, which leads to efficient computation of approximate communication equilibria.
In this section, we focus on the problem that each player $i \in N$ is faced with, and therefore, the subscript $i$ for $A_i$, $a_i$, $\Theta_i$, $\theta_i$, $\pi_i^t$, and $\RUSi^T$ is not essential.
We put the subscript $i$ just for notational consistency.
All the proofs of this section are deferred to \Cref{sec:ub-proof}.

\subsection{Reduction to $\Phi$-regret minimization}\label{sec:min1}
In the first step, the problem of minimizing untruthful swap regret is interpreted as an instance of \textit{$\Phi$-regret minimization}.
To our knowledge, the definition of $\Phi$-regret was first provided by \citet*{GJ03} when the decision space is a probability simplex and then extended to a general decision space by \citet*{SL07}.
For a set of transformations $\Phi$, the $\Phi$-regret measures how the algorithm gains by transforming their decisions with an optimal $\phi \in \Phi$ in hindsight.
To interpret untruthful swap regret as $\Phi$-regret, we vectorize the set of distributions for each type $\Delta(A_i)^{\Theta_i}$ as $\calX$ defined by
$\calX = \left\{ x \in [0,1]^{\Theta_i \times A_i} \;\middle|\; \sum_{a_i \in A_i} x(\theta_i,a_i) = 1 ~ (\forall \theta_i \in \Theta_i) \right\}$,
where we set $x^t(\theta_i,a_i) = \pi_i^t(\theta_i;a_i)$ for each $t \in [T]$, $\theta_i \in \Theta_i$, and $a_i \in A_i$.
Then the space of all transformations expressed by $\psi \colon \Theta_i \to \Theta_i$ and $\phi \colon \Theta_i \times A_i \to A_i$ can be written as $\calQ$ defined by
\begin{equation*}\label{eq:calQ}
\calQ = \left\{ Q \in [0,1]^{(\Theta_i \times A_i) \times (\Theta_i \times A_i)} 
\;\middle|\;
\begin{array}{l}
\text{there exists some $W \in [0,1]^{\Theta_i \times \Theta_i}$ such that}\\
\sum_{\theta'_i \in \Theta_i} W(\theta_i,\theta'_i) = 1 ~ \text{($\forall \theta_i \in \Theta_i$) and}\\
\sum_{a_i \in A_i} Q((\theta_i,a_i),(\theta'_i,a'_i)) = W(\theta_i,\theta'_i) ~ \text{($\forall \theta_i, \theta'_i \in \Theta_i, a'_i \in A_i$)}
\end{array}
\right\},
\end{equation*}
where $W$ represents a convex combinations of all possible $\psi$ and each block matrix $Q_{\theta_i,\psi(\theta_i)}$ represents a convex combination of all possible $\phi(\theta_i,\cdot)$.
Note that $W$ is a row stochastic matrix, while each block of $Q$ is a multiple of a column stochastic matrix.
Hence, each $Q \in \calQ$ is not a stochastic matrix, and this problem does not fit within the reduction framework of \citet*{BM07}.

Let $\bar{u}^t \in [0,1]^{\Theta_i \times A_i}$ be the reward vector weighted by prior probabilities $\rho_i$, i.e., $\bar{u}^t(\theta_i,a_i) = \rho_i(\theta_i) u_i^t(\theta_i,a_i)$ for each $\theta_i \in \Theta_i$ and $a_i \in A_i$.
\begin{restatable}{lemma}{RUSiPhi}\label{lem:RUSi-Phi}
The untruthful swap regret can be written as
\begin{equation*}
\RUSi = \max_{Q \in \calQ} \sum_{t=1}^T \langle Q x^t, \bar{u}^t \rangle - \sum_{t=1}^T \langle x^t, \bar{u}^t \rangle.
\end{equation*}
\end{restatable}

\subsection{Reduction to online linear optimization}\label{sec:min2}
In the second step, we apply the framework developed by \citet*{Gor08} that reduces $\Phi$-regret minimization to fixed point computation and online linear optimization.
First, we show that each $Q \in \calQ$ has a fixed point in $\calX$.

\begin{restatable}{lemma}{lemfixedpoint}
\label{lem:fixed-point}
For every $Q \in \calQ$, there exists $x \in \calX$ such that $Q x = x$.
\end{restatable}

Next, we show that this fixed point can be obtained by eigenvector computation.
As described later, all the subroutines are variants of the multiplicative weights, and therefore, all entries of their outputs are positive.
We can assume that each entry of $Q^t$, which is computed as a product of the subroutines' outputs, is positive for every $t \in [T]$.

\begin{restatable}{lemma}{lemeigen}
\label{lem:eigen}
If all entries of $Q \in \calQ$ are positive, then we can compute $x \in \calX$ such that $Qx = x$ in time polynomial in $|\Theta_i|$ and $|A_i|$.
\end{restatable}

Using this fixed point computation, we can solve untruthful swap regret minimization via a subroutine for online linear optimization over decision space $\calQ$.
Let $Q^t \in \calQ$ be the subroutine's output in round $t \in [T]$.
We then compute a fixed point $x^t \in \calQ$ that satisfies $Q^t x^t = x^t$, and let $x^t$ be the decision for round $t$.
For the subroutine, feed reward $U^t \in [0,1]^{(\Theta_i \times A_i) \times (\Theta_i \times A_i)}$ defined by $U^t((\theta_i,a_i),(\theta'_i,a'_i)) = \bar{u}^t(\theta_i,a_i) x^t(\theta'_i,a'_i)$ for each $\theta_i,\theta'_i \in \Theta_i$ and $a_i,a'_i \in A_i$.
We define external regret for this subroutine as
\begin{equation*}
R^T_\calQ = \max_{Q \in \calQ} \sum_{t = 1}^T \langle Q, U^t \rangle - \sum_{t = 1}^T \langle Q^t, U^t \rangle.
\end{equation*}
Then we show that untruthful swap regret $\RUSi^T$ equals this subroutine's external regret $R^T_\calQ$.

\begin{restatable}{lemma}{lemreductionone}
\label{lem:reduction1}
$\RUSi^T = R^T_\calQ$.
\end{restatable}

Since for online linear optimization over a polytope, there exist efficient algorithms with sublinear external regret (e.g., Follow-the-Perturbed-Leader \citep*{KV05} or Component Hedge \citep*{Koolen10}), this problem is already tractable.
However, they do not provide an upper bound of the optimal order.
We thus need the following third step.

\subsection{Decomposition into small external regret minimization problems}\label{sec:min3}
Next, we reduce online linear optimization with decision space $\calQ$ to small external regret minimization problems, using an approach similar to counterfactual regret minimization \citep*{Zinkevich07}.
Since $\calQ$ can be regarded as a product of block stochastic matrices and a stochastic matrix in some sense, it can be decomposed into $|\Theta_i|^2|A_i| + |\Theta_i|$ probability simplices.
For each $\theta_i, \theta'_i \in \Theta_i$ and $a'_i \in A_i$, we use a subroutine $\calE_{\theta_i,\theta'_i,a'_i}$ with decision space $A_i$.
Moreover, for each $\theta_i \in \Theta_i$, we use a subroutine $\calE_{\theta_i}$ with decision space $\Theta_i$.
In total, we use $|A_i||\Theta_i|^2 + |\Theta_i|$ subroutines.

The reduction proceeds as follows.
Let $y^t_{\theta_i,\theta'_i,a'_i} \in \Delta(A_i)$ be the output of $\calE_{\theta_i,\theta'_i,a'_i}$ and $w^t_{\theta_i} \in \Delta(\Theta_i)$ the output of $\calE_{\theta_i}$ in round $t \in [T]$.
According to these outputs, we set $Q^t((\theta_i,a_i),(\theta'_i,a'_i)) = w^t_{\theta_i}(\theta'_i) y^t_{\theta_i,\theta'_i,a'_i}(a_i)$ for each $\theta_i,\theta'_i \in \Theta_i$ and $a_i,a'_i \in A_i$.
Note that this $Q^t$ is contained in $\calQ$, which can be checked by setting $W(\theta_i,\theta'_i) = w^t_{\theta_i}(\theta'_i)$ for each $\theta_i,\theta'_i \in \Theta_i$.

In each round $t \in [T]$, based on the observed reward $U^t = \bar{u}^t (x^t)^\top$, we feed the reward for each subroutine as follows.
\begin{itemize}
\item For each $\calE_{\theta_i,\theta'_i,a'_i}$, the reward for decision $a_i \in A_i$ is $x^t(\theta'_i,a'_i) \bar{u}_i^t(\theta_i,a_i)$.
\item For each $\calE_{\theta_i}$, the reward for decision $\theta'_i \in \Theta_i$ is $\sum_{a_i,a'_i \in A_i} y^t_{\theta_i,\theta'_i,a'_i}(a_i) x^t(\theta'_i,a'_i) \bar{u}^t(\theta_i,a_i)$.
\end{itemize}

The external regret for online linear optimization with decision space $\calQ$ is bounded by the sum of the external regrets for these subroutines as follows.

\begin{restatable}{lemma}{lemreductiontwo}
\label{lem:reduction2}
\begin{equation*}
R^T_\calQ \le \sum_{\theta_i \in \Theta_i} R^T_{\theta_i} + \sum_{\theta_i \in \Theta_i} \max_{\theta'_i \in \Theta_i} \sum_{a'_i \in A_i} R^T_{\theta_i,\theta'_i,a'_i}.
\end{equation*}
\end{restatable}

These three-step reduction yields the algorithm described in \Cref{alg:untruthful}.
Finally, we obtain an upper bound on untruthful swap regret as follows.
\begin{restatable}{theorem}{thmuntruthfulub}
\label{thm:untruthful-ub}
The untruthful swap regret of \Cref{alg:untruthful} is bounded as
\begin{equation*}
\RUSi^T \le \sqrt{\frac{1}{2} T \log |\Theta_i|} + 6 \sqrt{T |A_i| \log |A_i|}.
\end{equation*}
\end{restatable}

\begin{algorithm}[t]
\caption{Algorithm for minimizing untruthful swap regret}\label{alg:untruthful}
	\KwIn{The set of types $\Theta_i$ and the set of actions $A_i$ are specified in advance. The reward vector $u_i^t \in [0,1]^{\Theta_i \times A_i}$ is given at the end of each round $t \in [T]$.}
	Initialize subroutines as follows:
	\begin{itemize}
	\item let $\calE_{\theta_i}$ be a multiplicative weights algorithm with decision space $\Theta_i$ for each $\theta_i \in \Theta_i$, and
	\item let $\calE_{\theta_i,\theta'_i,a'_i}$ be AdaHedge (see, e.g., \citep*[Section 7.6]{Orabona19}) with decision space $A_i$ for each $\theta_i, \theta'_i \in \Theta_i$ and $a'_i \in A_i$.
	\end{itemize}
	\For{each round $t = 1,\dots,T$}{
		Let $w^t_{\theta_i} \in \Delta(\Theta_i)$ be the output of $\calE_{\theta_i}$ in round $t$ for each $\theta_i \in \Theta_i$.\\
		Let $y^t_{\theta_i,\theta'_i,a'_i} \in \Delta(A_i)$ be the output of $\calE_{\theta_i,\theta'_i,a'_i}$ in round $t$ for each $\theta_i,\theta'_i \in \Theta_i$ and $a'_i \in A_i$.\\
		Define $Q^t \in [0,1]^{(\Theta_i \times A_i) \times (\Theta_i \times A_i)}$ by $Q^t((\theta_i,a_i),(\theta'_i,a'_i)) = w^t_{\theta_i}(\theta'_i) y^t_{\theta_i,\theta'_i,a'_i}(a_i)$ for each $\theta_i,\theta'_i \in \Theta_i$ and $a_i,a'_i \in A_i$.\\
		Compute an eigenvector $x^t \in \bbR^{\Theta_i \times A_i}$ of $Q^t$ such that $Q^t x^t = x^t$ and $(x^t)^\top \bfone = |\Theta_i|$.\\
		Decide the output $\pi_i^t \in \Delta(A_i)^{\Theta_i}$ by $\pi_i^t(\theta_i;a_i) = x^t(\theta_i, a_i)$ for each $\theta_i \in \Theta_i$ and $a_i \in A_i$.\\
		Observe reward vector $u_i^t \in [0,1]^{\Theta_i \times A_i}$ and feed reward vectors to subroutines as follows:
		\begin{itemize}
		\item feed $\displaystyle \sum_{a_i,a'_i \in A_i} y^t_{\theta_i,\theta'_i,a'_i}(a_i) \pi_i^t(\theta'_i;a'_i) \rho_i(\theta_i) u_i^t(\theta_i,a_i)$ as the reward for decision $\theta'_i \in \Theta_i$\\to subroutine $\calE_{\theta_i}$ for each $\theta_i \in \Theta_i$, and
		\item feed $\displaystyle \pi_i^t(\theta'_i;a'_i) \rho_i(\theta_i) u_i^t(\theta_i,a_i)$ as the reward for decision $a_i \in A_i$ to subroutine $\calE_{\theta_i,\theta'_i,a'_i}$\\for each $\theta_i,\theta'_i \in \Theta_i$ and $a'_i \in A_i$.
		\end{itemize}
	}
\end{algorithm}

\subsection{Application to equilibrium computation}\label{sec:polytime}

Now we apply the algorithm for untruthful swap regret minimization to equilibrium computation.
\Cref{thm:dynamics-com} claims that if each player minimizes the untruthful swap regret, then the dynamics converge to the intersection of communication equilibria and ANFCEs.
However, exactly simulating the dynamics requires the exact evaluation of rewards $u_i^t$, which needs exponential time in general.
Instead, we estimate these values by random sampling and bound the error using H\oe{}ffding's inequality.

\begin{restatable}{corollary}{corcom}
\label{cor:com}
For any $\epsilon > 0$, there exists an algorithm that outputs a succinct representation of $\pi \in \PiCom^\epsilon \cap \PiANF^\epsilon$ with probability at least $1-\delta$ and runs in time polynomial in $n$, $\max_{i \in N} |\Theta_i|$, $\max_{i \in N} |A_i|$, $1/\epsilon$, and $\log(1/\delta)$ with an oracle for utility functions.
\end{restatable}

\section{Lower bound for untruthful swap regret}\label{sec:lb}
Our next result is an $\Omega(\sqrt{T \log |\Theta_i|})$ lower bound on untruthful swap regret, which is tight up to a multiplicative constant in terms of the number of rounds $T$ and types $|\Theta_i|$.
In terms of the number of actions, for swap regret minimization, which is a special case of our setting (the case of $|\Theta_i| = 1$), \citet*{Ito20} provided an $\Omega(\sqrt{T |A_i| \log |A_i|})$ lower bound for adaptive adversaries, and \citet*{PR24} and \citet*{DDFG24} provided an $\Omega(\sqrt{T |A_i|})$ lower bound for oblivious adversaries.
Since our upper bound applies only to oblivious adversaries, a gap of order $\sqrt{\log |A_i|}$ remains.

Since the full proof is involved, it is deferred to \Cref{sec:lb-proof}.
In this section, we present a challenging problem instance used for our proof of the lower bound, which explains why untruthful swap regret cannot be addressed by existing techniques for proving lower bounds on swap regret.
Additionally, we outline the direction of our proof.

The proof idea for the swap regret lower bounds \citep*{BM07,Ito20} can be summarized as follows.
A problem instance is constructed such that the reward for each action $a_i \in A_i$ in each round $t \in [T]$ is independently chosen from the uniform distribution over $\{0,1\}$.
We can assume that the algorithm selects most of the actions $A_i$ for at least $\Omega(T/|A_i|)$ rounds.
This can be guaranteed by ``blocking'' operations, which set the rewards of actions selected more than a certain number of times to $0$ for the remaining rounds.
For the rounds when the algorithm selects such an action $a_i \in A_i$, the gap between the total expected rewards for $a_i$ and an optimal action $\phi(a_i) \in A_i$ is $\Omega \left(\sqrt{\frac{T}{|A_i|} \log |A_i|} \right)$ due to an anti-concentration bound for binomial distributions.
The swap regret is at least the sum of $\Omega \left(\sqrt{\frac{T}{|A_i|} \log |A_i|} \right)$ for these actions, hence $\Omega(\sqrt{T |A_i| \log |A_i|})$.

Based on a similar idea, for untruthful swap regret, we can easily prove a lower bound if the algorithm's decisions are guaranteed to be sufficiently different for different types.
For example, we consider a problem instance with two actions $A_i = \{\alpha_0,\alpha_1\}$, in which for each type $\theta_i \in \Theta_i$ and round $t \in [T]$, the reward for $\alpha_0$ is chosen from $\{0,1\}$ uniformly at random, and the reward for $\alpha_1$ is decided by flipping the corresponding reward for $\alpha_0$.
Then the total expected reward for any algorithm is $T/2$.
Assume that the algorithm's decisions are completely different for different types.
For example, we assume that for each type $\theta_i \in \Theta_i$ and round $t \in [T]$, the algorithm chooses $\alpha_0$ or $\alpha_1$ uniformly at random.
Recall that the competitor of untruthful swap regret can apply any type swap $\psi \colon \Theta_i \to \Theta_i$.
Since for each pair of $\theta_i$ and $\psi(\theta_i) \in \Theta_i$, the expected total reward obtained by applying the decision for $\psi(\theta_i)$ to the rewards for $\theta_i$ follows a binomial distribution, the optimal reward is $T/2 + \Omega(\sqrt{T \log|\Theta_i|})$ again from an anti-concentration bound, which leads to a desired lower bound.

However, it is not easy to guarantee that the algorithm's decisions are sufficiently different for different types.
For any payoffs, if the algorithm's decisions are completely the same for all types, any $\psi$ yields no untruthful swap regret.
We thus need to consider a combination of $\psi$ and $\phi$.
For example, if there is a type $\theta_i \in \Theta_i$ for which the payoff for $\alpha_0$ is always better than $\alpha_1$, the algorithm suffers untruthful swap regret due to $\phi(\theta_i,\alpha_0) = \phi(\theta_i,\alpha_1) = \alpha_0$ unless always choosing $\alpha_0$ (as formally stated in \Cref{lem:edge-type}).

\begin{figure}
\centering
\newlength{\width}
\setlength{\width}{\textwidth}
\begin{tikzpicture}[line width=0.001\paperwidth]
\node (block) [minimum width=0.6\width, minimum height=0.2\width] {};
\node (random) at (block.south) [anchor=south, minimum width=0.6\width, minimum height=0.1\width, fill=gray, draw=gray, text=white] {};
\node (branch) at (block.north) [anchor=north, minimum width=0.6\width, minimum height=0.1\width, draw=gray] {};
\node at (branch.north west) [anchor=north west, minimum width=0.15\width, minimum height=0.05\width, fill=black] {};
\node at ([xshift=0.15\width]branch.north west) [anchor=north west, minimum width=0.15\width, minimum height=0.025\width, inner sep=0, fill=black] {};
\node at ([xshift=0.15\width,yshift=-0.05\width]branch.north west) [anchor=north west, minimum width=0.15\width, minimum height=0.025\width, inner sep=0, fill=black] {};
\foreach \x in {0,1,...,3}{
\node at ([xshift=0.30\width,yshift=-0.025*\x\width]branch.north west) [anchor=north west, minimum width=0.15\width, minimum height=0.0125\width, inner sep=0, fill=black] {};
}
\foreach \x in {0,1,...,7}{
\node at ([xshift=0.45\width,yshift=-0.0125*\x\width]branch.north west) [anchor=north west, minimum width=0.15\width, minimum height=0.00625\width, inner sep=0, fill=black] {};
}
\foreach \x in {0.00625,0.0125,...,0.09375}{
\draw[gray] ([yshift=-\x\width]branch.north west) -- ([yshift=-\x\width]branch.north east);
}
\foreach \x in {-0.29375,-0.28750,...,0.29375}{
\draw[gray] ([xshift=\x\width]branch.north) -- ([xshift=\x\width]branch.south);
}
\node [text=white] at (random) {independently from the uniform distribution over $\{0,1\}$};
\node at ([xshift=-0.225\width]branch.north) [anchor=south] {$1$st block};
\node at ([xshift=-0.075\width]branch.north) [anchor=south] {$2$nd block};
\node at ([xshift=+0.075\width]branch.north) [anchor=south] {$\cdots$};
\node at ([xshift=+0.225\width]branch.north) [anchor=south] {$B$th block};
\node (par) at (branch.east) [anchor=west, inner sep=0, scale=1.5] {$\Biggr\}$};
\node at (par.east) [anchor=west, inner sep=0] {randomly shuffled};
\node at (branch.west) [anchor=east] {$\Theta'_i$};
\node at (random.west) [anchor=east] {$\Theta''_i$};
\end{tikzpicture}
\caption{
A problem instance for proving a lower bound on untruthful swap regret.
The types $\Theta_i$ are partitioned into $\Theta'_i$ and $\Theta''_i$.
The rewards for types $\Theta'_i$ are constant within each block of rounds and randomly branch at the beginning of each block.
The reward for action $\alpha_0$ is $1$ for the black cells and $0$ for the white cells.
For each $\theta''_i \in \Theta''_i$, the rewards for $\alpha_0$ are generated independently.
Every reward for action $\alpha_1$ is determined by flipping the corresponding reward for $\alpha_0$.
}\label{fig:hard-instance}
\end{figure}

Based on this observation, we create a problem instance with $A_i = \{\alpha_0,\alpha_1\}$ and the uniform distribution $\rho_i$ (\Cref{fig:hard-instance}).
We assume $|\Theta_i| = 2^{B+1}$ for some positive integer $B$ for simplicity.
The set of types $\Theta_i$ is partitioned into $\Theta'_i$ and $\Theta''_i$ of equal size.
The set of rounds $[T]$ is partitioned into $B$ blocks of equal length.
The rewards for $\Theta''_i$ are randomly determined as in the problem instance described above.
Each $\theta'_i \in \Theta'_i$ is associated with a binary sequence of length $B$ by a random bijection $\zeta \colon \Theta'_i \to \{0,1\}^B$, and the $b$th bit of this sequence indicates which of $\alpha_0$ or $\alpha_1$ has reward $1$ for $\theta'_i$ in the $b$th block.

The rewards for $\Theta'_i$ are designed so that the algorithm must use different decisions for different types.
Let $\theta^0_i,\theta^1_i \in \Theta'_i$ be the types for which $\alpha_0$'s and $\alpha_1$'s rewards are always $1$ (the completely black row and completely white row), respectively.
As mentioned above, $\phi(\theta_i,\alpha_0) = \phi(\theta_i,\alpha_1) = \alpha_0$ forces the algorithm to almost always select $\alpha_0$ for the type $\theta^0_i$.
Similarly, the algorithm almost always selects $\alpha_1$ for the type $\theta^1_i \in \Theta'_i$.
Since the rewards for $\Theta'_i$ are randomly shuffled, the algorithm does not know which are $\theta^0_i$ and $\theta^1_i$ before reaching the final block.
In the first block, since the algorithm does not know which type in the black rows is $\theta^0_i$ and which type in the white rows is $\theta^1_i$, the algorithm must almost always choose $\alpha_0$ for all the black cells and $\alpha_1$ for all the white cells.

We are tempted to expect that for the second and later blocks, the algorithm also almost always selects $\alpha_0$ for the black cells and $\alpha_1$ for the white cells.
Unfortunately, there is a counterexample to this expectation.
Suppose that the algorithm uses the same action for each contiguous \textit{two} blocks.
For example, if the first, second, third, and fourth blocks are black, white, white, and black, respectively, the algorithm uses $\alpha_0$ for the first two blocks and $\alpha_1$ for the next two blocks.
Then the algorithm successfully selects $\alpha_0$ for $\theta_i^0$ and $\alpha_1$ for $\theta_i^1$ in almost all rounds.
Furthermore, even though the algorithm obtains the total reward only $T/2$ for some type, this is optimal among all decisions that select the same action for each contiguous two blocks, and then the algorithm suffers no untruthful swap regret.%
\footnote{%
Note that a similar algorithm focusing on each contiguous \textit{three} blocks yields $\Omega(T)$ untruthful swap regret.
For example, this algorithm always selects $\alpha_0$ for the row with black, white, and white cells for the first, second, and third blocks, but selecting $\alpha_1$ all the time obtains $T/3$ better reward.
}
We prove that this example is tight, that is, if the untruthful swap regret is small, then the difference of the algorithm's decisions between the ``adjacent'' black and white rows must be at least $T/2$ approximately (see \Cref{lem:gap-type} for a formal statement).
By using this fact, we show that the algorithm's decisions must be significantly different at each branching point with some constant probability (\Cref{lem:gap-block}).

If the algorithm always selects $\alpha_0$ for the black cells and $\alpha_1$ for the white cells, by choosing a better action separately for $\theta''_i \in \Theta''_i$ in each block, we can show that the algorithm suffers $B \cdot \Omega(\sqrt{T/B}) = \Omega(\sqrt{T \log |\Theta_i|})$ untruthful swap regret.
Since we can control the algorithm's decisions only for a constant fraction of all rounds, we need to carefully determine $\psi$ so that the algorithm cannot manipulate the untruthful swap regret (\Cref{sec:lb2}) and trace the difference of cumulative rewards by using martingale analysis (\Cref{sec:lb3}).
This technique to track the algorithm's cumulative reward as a martingale and apply Doob's optional stopping theorem may be of independent interest as a widely applicable approach to analyzing regret lower bounds.
These analyses lead to the following lower bound (\Cref{sec:lb4}).

\begin{restatable}{theorem}{thmuntruthfullb}
\label{thm:untruthful-lb}
Assume $|A_i| = 2$, $|\Theta_i| = 2^{B+1}$ for some $B \in \bbZ$, and $T$ is a multiple of $B = \log_2 |\Theta_i| - 1$.
If $T \ge 2^{-49} |\Theta_i|^2 \log_2 |\Theta_i|$,
then for any randomized algorithm, there exists an oblivious adversary for which the untruthful swap regret of the algorithm is at least $2^{-28} \sqrt{T \log_2 |\Theta_i|}$.
\end{restatable}

The assumption of $T = \Omega(|\Theta_i|^2 \log |\Theta_i|)$ is essentially used in the proofs.
The other assumptions on $|A_i|$, $|\Theta_i|$, and $T$ are just for simplicity of the proofs and only affect the multiplicative constants.

\section{Smoothness and price of anarchy}\label{sec:poa}
In this section, we introduce lower bounds on the price of anarchy (PoA) in Bayesian games via the smoothness assumption.
The proofs of this section are deferred to \Cref{sec:poa-app}.

As in existing studies \citep*{Roughgarden15incomplete,Syrgkanis12,ST13}, we make the following two assumptions for all of our PoA results.
\begin{itemize}
\item Assumption 1: the prior distribution $\rho \in \Delta(\Theta)$ is a product distribution, that is, there exists some $\rho_i \in \Delta(\Theta_i)$ for each $i \in N$ such that $\rho(\theta) = \prod_{i \in N} \rho_i(\theta_i)$.
\item Assumption 2: the value of $v_i(\theta;a)$ does not depend on $\theta_{-i}$.
	Under this assumption, we write $v_i(\theta_i; a)$ in place of $v_i(\theta; a)$ by abuse of notation.
\end{itemize}
These assumptions naturally hold when type $\theta_i$ represents a preference or attribute of each player $i \in N$.
We define the price of anarchy as follows.

\begin{definition}[Price of anarchy of Bayesian games]\label{def:poa}
Given an equilibrium class $\Pi \subseteq \Delta(A)^\Theta$, its price of anarchy is defined by
\begin{equation*}
\PoA_\Pi = 
\displaystyle
\left.
\inf_{\pi \in \Pi}\E_{\theta \sim \rho} \left[ \E_{a \sim \pi(\theta)} \left[ \vSW(\theta; a) \right] \right]%
\middle/
\displaystyle\E_{\theta \sim \rho} \left[ \max_{a \in A} \vSW(\theta; a) \right]
\right.
,
\end{equation*}
where $\vSW \colon \Theta \times A \to [0,1]$ is the social welfare function.
\end{definition}

While most of the existing studies considered the case in which $\Pi$ is the set of Bayes--Nash equilibria,
we assume that $\Pi$ is the intersection of (non-approximate) communication equilibria and ANFCEs, i.e., $\PiCom^0 \cap \PiANF^0$.
As described in \Cref{cor:com}, an equilibrium in this class can be efficiently computed by simulating no-untruthful-swap-regret dynamics.

Our analysis is based on the smoothness argument, which has been widely used for bounding the PoA of coarse correlated equilibria in games with complete information.
The smoothness argument cannot be extended directly to Bayesian games, even for Bayes--Nash equilibria,
 because the denominator of the PoA takes an optimal action profile for each $\theta \in \Theta$, which can be better than the expected social welfare for the optimal strategy profile.
Therefore, we cannot directly apply the smoothness analysis of the strategic form or agent normal form.
Existing studies analyze the PoA of Bayes--Nash equilibria by making an additional assumption \citep*{Syrgkanis12,Roughgarden15incomplete,ST13}.

We extend their results from Bayes--Nash equilibria to the intersection of communication equilibria and ANFCEs for the following two cases:
the case in which the social welfare is the sum of all players' payoffs and
Bayesian games derived from various Bayesian mechanisms, where the social welfare is defined as the sum of players' valuations (not payoffs).
The first case has applications in various resource allocations games \citep*{KO22,Yao23}, and the second case in various \textit{conditionally} smooth mechanisms such as the sequential composition of smooth mechanisms \citep*{ST13}, greedy auctions with matching constraints \citep*{LS15}, draft auctions \citep*{DMSW15}.
Note that \citet*{HST15} assumed a stronger assumption of \textit{(unconditional)} smoothness of mechanisms, while existing studies including \citet*{ST13} have shown various mechanisms to satisfy only conditional smoothness.
By carefully reading their proofs of conditional smoothness, we can see that some mechanisms also satisfy (unconditional) smoothness, but it excludes several important examples mentioned above.

\subsection{PoA bounds for the sum of payoffs}\label{sec:poa-sum}

Here we state our result for the case in which the social welfare is defined as the sum of all players' payoffs, i.e., $\vSW(\theta; a) = \sum_{i \in N} v_i(\theta_i; a)$.
See \Cref{sec:poa-sum} for the proofs and \Cref{sec:poa-mechanism} for the result for mechanisms.
\citet*{Syrgkanis12} proved a PoA lower bound with (unconditional) smoothness, and \citet*{ST13} then extended it to a weaker version of smoothness, called \textit{conditional smoothness}.

\begin{definition}[$(\lambda,\mu)$-Conditional smoothness of Bayesian games {\citep*{ST13}}]
A Bayesian game is $(\lambda,\mu)$-conditionally smooth
if
there exists some $a^*_{i,\theta,a_i} \in A_i$ for each $i \in N$, $\theta \in \Theta$, and $a_i \in A_i$ such that
\begin{equation*}
\sum_{i \in N} v_i(\theta_i; a^*_{i,\theta,a_i}, a_{-i}) \ge \lambda \max_{a' \in A} \vSW(\theta; a') - \mu \vSW(\theta; a)
\end{equation*}
holds for any $\theta \in \Theta$ and $a \in A$.
\end{definition}

We extend their PoA bound to the intersection of communication equilibria and ANFCEs.
It is important to emphasize that our proof relies on properties of both communication equilibria and ANFCEs.
Since conditional smoothness considers deviations of each player depending on the current action, it seems to correspond to correlated equilibria, while (unconditional) smoothness seems to correspond to coarse correlated equilibria.
However, even under the (unconditional) smoothness assumption, the current proof requires properties of both communication equilibria and ANFCEs, and it cannot be directly extended to Bayes coarse correlated equilibria.

\begin{restatable}{theorem}{thmpoasum}
\label{thm:poa-sum}
Assume that the social welfare is the sum of all players' payoffs.
For any $(\lambda,\mu)$-conditionally smooth Bayesian game satisfying Assumptions 1 and 2, the price of anarchy of the intersection of communication equilibria and ANFCEs is at least $\lambda / (1+\mu)$.
\end{restatable}

\subsection{PoA bounds for conditionally smooth mechanisms}\label{sec:poa-mechanism}

Next, we consider applications to mechanism design.
Since it is not the focus of this paper, we do not introduce mechanism design in details.
See, e.g., \citet*{AGTmechanism} or \citet*{BMD} for basics of mechanism design and \citet*{ST13} for smooth mechanisms.

A mechanism (for auctions) is regarded as a function that maps players' type reports (bids) $a \in A$ to an allocation and payment.
Let $X_i$ be the set of allocation for player $i \in N$ and $X \subseteq \prod_{i \in N} X_i$ the set of all possible allocations.
An allocation function $f_i \colon A \to X_i$ for each $i \in N$ determines the allocation for player $i$ based on the bids $a \in A$.
We can consider the process of mechanisms as a game in which each player decides a bid $a_i \in A_i$ and then obtains a payoff based on the allocation and the payment.
We assume that each player's utility function $v_i$ is \textit{quasilinear}, i.e., there exist some valuation function $v_i^+ \colon \Theta_i \times X_i \to [0,1]$ and payment function $v_i^- \colon A \to [0,1]$ such that 
\begin{equation*}
v_i(\theta_i;a) = v_i^+(\theta_i;f_i(a)) - v_i^-(a).
\end{equation*}
An important assumption here is the payment function depends only on $a$, not on $\theta_i$.
This assumption naturally holds because the payment is determined by the mechanism that does not know each player's type.
We also assume that the payoff of each player is always non-negative, i.e., $v_i(\theta_i;a) \ge 0$ for any $a \in A$.
This should hold because each player can withdraw from the mechanism if the payoff is negative.
The social welfare function is defined as $\vSW(\theta;a) = \sum_{i \in N} v_i^+(\theta_i;f_i(a))$ for each $\theta \in \Theta$ and $a \in A$.

\citet*{ST13} proved that various simple Bayesian mechanisms, their simultaneous composition with XOS valuations, and their sequential composition with unit-demand bidders are conditionally smooth.

\begin{definition}[$(\lambda,\mu)$-Conditional smoothness of Bayesian mechanisms {\citep*{ST13}}]
A Bayesian mechanism is $(\lambda,\mu)$-conditionally smooth
if
there exists some $a^*_{i,\theta,a_i} \in A_i$ for each $i \in N$, $\theta \in \Theta$, and $a_i \in A_i$ such that
\begin{equation*}
\sum_{i \in N} v_i(\theta_i; a^*_{i,\theta,a_i}, a_{-i})
\ge
\lambda \max_{x \in X} \sum_{i \in N} v_i^+(\theta_i; x_i) - \mu \sum_{i \in N} v_i^-(a).
\end{equation*}
holds for any $\theta \in \Theta$ and $a \in A$.
\end{definition}

\citet*{ST13} proved a $\lambda / \max\{ 1, \mu \}$ lower bound on the PoA of Bayes--Nash equilibria with this assumption.
We extend it to the intersection of communication equilibria and ANFCEs.

\begin{restatable}{theorem}{thmpoamechanism}
\label{thm:poa-mechanism}
For any $(\lambda,\mu)$-conditionally smooth mechanism satisfying Assumptions 1 and 2, the price of anarchy of the intersection of communication equilibria and ANFCEs is at least $\lambda / \max\{ 1, \mu \}$.
\end{restatable}

\citet*{ST13} also proved a similar result for weak smoothness and budget constraints.
We can generalize our proof to these settings and obtain PoA bounds for more various mechanisms.
These results yield lower bounds on the PoA of the intersection of communication equilibria and ANFCEs for the simultaneous and sequential composition of first-price, second-price, all-pay, and many other auctions.

Note that the PoA lower bounds of ANFCCEs provided by \citet*{HST15} assume (unconditional) smoothness, which is stronger than conditional smoothness.
Since \citet*{ST13} proved only conditional smoothness for various mechanisms, applications of these lower bounds are limited.
For example, the sequential composition of smooth mechanisms with unit-demand bidders does not satisfy (unconditional) smoothness, and furthermore, their PoA of ANFCCEs is unbounded for some mechanisms.
This fact can been seen from the following example of the sequential first-price auction with complete information.

\begin{example}\label{ex:bad-sequential-mechanism}
We consider a sequential composition of two consecutive first-price auctions with two bidders.
A single item is sold in each stage.
In the first stage, an item that has value $0$ for the first bidder and $\epsilon \in (0,1)$ for the second bidder is sold.
In the second stage, an item that has value $1$ for the first bidder and $\epsilon$ for the second bidder is sold.

We consider a coarse correlated equilibrium of the game derived from this mechanism.
Suppose that a mediator decides recommendations as follows.
For both bidders, the mediator suggests bidding the same random number $b \epsilon$ in the first stage, where $b$ is generated from the uniform distribution over $[0,1]$, and bidding $0$ in the second stage.
Only when the first player does not bid $b \epsilon$ in the first stage, the second player is recommended bidding $100$ (any large number).
For simplicity, we employ a tie-breaking rule in which, if the bids are identical, an item is allocated to the second player.

This distribution is a coarse correlated equilibrium for the following reason.
If both players follow the recommendation, since both players' bids are identical in both stages, both items are allocated to the second player.
The payoffs for the first and second players are $0$ and $\epsilon(1-b)$, respectively.
Since both players' bids are identical, the second player cannot obtain each item if he reduces the bid.
Therefore, the second player has no incentive to deviate from the recommendation.
The first player only has the valuation for the second item.
To get the second item, he must deviate from the recommendation.
However, in a CCE, he must choose following the recommendation or decide a bid without observing the recommendation.
If he does not know the recommendation, he cannot bid $b\epsilon$ in the first stage, and then the second player blocks the first player by bidding $100$ in the second stage.
Therefore, this distribution is a CCE, and the PoA of CCEs is at most $\epsilon / 1$, which can be arbitrarily small.

\end{example}

\begin{remark}
In applications to auctions, a type $\theta_i$ represents a valuation of an item for player $i$, and an action $a_i \in A_i$ represents a bid for this good.
The sets of possible values $\Theta_i$ and $A_i$ are usually continuous intervals.
Our proposed algorithm for computing an equilibrium does not assume any structure for $\Theta_i$ and $A_i$ and requires polynomial time in $|\Theta_i|$ and $|A_i|$, which might be prohibitive.
A naive solution to this issue is to discretize $\Theta_i$ and $A_i$.
Another possible solution is to design an efficient algorithm using the continuous structure of $\Theta_i$ and $A_i$, which is left for future work.
\end{remark}

\section*{Acknowledgements}
This work was supported by JSPS KAKENHI Grant Number 22K17857.

\bibliographystyle{plainnat}
\bibliography{main}

\appendix

\section{Preliminaries}\label{sec:pre}

\paragraph{Basic notations}
The sets of reals and integers are denoted by $\bbR$ and $\bbZ$, respectively.
For any $a,b \in \bbR$, we define $[a,b] = \{ x \in \bbR \mid a \le x \le b\}$.
For any positive integer $n \in \bbZ$, we define $[n] = \{i \in \bbZ \mid 1 \le i \le n\} = \{ 1,2,\dots,n \}$.
For any finite set $A$, we denote the probability simplex by $\Delta(A) = \{ \pi \in [0,1]^A \mid \sum_{a \in A} \pi(a) = 1 \}$.
For any distribution $\pi \in \Delta(A)$, we write $a \sim \pi$ when $a \in A$ is generated from the distribution $\pi$.
By abuse of notation, we write $a \sim A$ when $a$ is generated from the uniform distribution over $A$.
To specify a random variable that a probability or an expected value is concerned with, we write $\Pr_{a \sim \pi} \left( \cdot \right)$ or $\E_{a \sim \pi} \left[ \cdot \right]$.
If the distribution is obvious, we simply write $\Pr_{a} \left( \cdot \right)$ or $\E_{a} \left[ \cdot \right]$.
For vectors $x,y \in \bbR^n$, the inner product is written as $\langle x, y \rangle = \sum_{i=1}^n x(i) y(i)$.
Similarly, for matrices $P,Q \in \bbR^{n \times m}$, the matrix inner product is written as $\langle P, Q \rangle = \sum_{i=1}^n \sum_{j=1}^m P(i,j) Q(i,j)$.

\paragraph{Bayesian games}
Let $n \in \bbZ$ be the number of players and $N = [n]$ a set of players.
Let $A_i$ be a finite set of actions for player $i \in N$ and $A = A_1 \times \dots \times A_n$ the set of action profiles of all players.
Let $\Theta_i$ be a set of finite possible types for each player $i \in N$ and $\Theta = \Theta_1 \times \dots \times \Theta_n$ the set of type profiles of all players.
Let $v_i \colon \Theta \times A \to [0,1]$ be a utility function that maps a type profile and action profile to a payoff value for player $i$.%
\footnote{%
If each player's type represents their own preference, it is more natural to assume that their payoff depends only on their own type, not the other players' types.
In this case, the utility function is defined to be $v_i \colon \Theta_i \times A \to [0,1]$ whose value is determined by their own type $\theta_i$, not by the other player's types $\theta_{-i}$.
We make this assumption to obtain PoA bounds in \Cref{sec:poa}.
}
Since the size of $\Theta$ and $A$ is exponential in $n$, we assume that each $v_i$ is given by an oracle that returns the value of $v_i(\theta;a)$ for any $\theta \in \Theta$ and $a \in A$.
For each player $i \in N$, subscript $-i$ represents the other players $N \setminus \{i\}$, i.e., $a_{-i} = (a_j)_{j \in N \setminus \{i\}}$ and $\theta_{-i} = (\theta_j)_{j \in N \setminus \{i\}}$.

In this paper, we distinguish between the two words ``strategies'' and ``actions.''
While actions are defined as above, strategies are decisions that determine an action for every type.
Formally, a strategy $s_i \colon \Theta_i \to A_i$ of player $i \in N$ determines action $s_i(\theta_i) \in A_i$ to be selected given their own type $\theta_i \in \Theta_i$.
For each player $i \in N$, let $S_i = {A_i}^{\Theta_i}$ be the set of all strategies.
We write a strategy profile as $s = (s_1,\dots,s_n)$, and let $S = S_1 \times \dots \times S_n$ be the set of all strategy profiles.
For notational simplicity, we write $s(\theta) = (s_j(\theta_j))_{j \in N}$ and $s_{-i}(\theta_{-i}) = (s_j(\theta_j))_{j \in N \setminus \{i\}}$ for each $s \in S$, $\theta \in \Theta$, and $i \in N$.

In Bayesian games, we assume that the type profile is generated from the probability distribution $\rho \in \Delta(\Theta)$ that every player knows in advance.
Since the size of $\Theta$ is exponential in $n$ in general, the distribution $\rho$ cannot be expressed in polynomial size.
In most parts of this paper, we assume that it is possible to efficiently sample $\theta_i$ from the marginal distribution of $\rho$, which is denoted by $\rho_i$, and $\theta_{-i}$ from the distribution $\rho$ conditioned on any $\theta_i \in \Theta_i$, which is denoted by $\rho|\theta_i$.
This assumption holds for a special case in which $\rho$ is a product distribution, i.e., $\rho(\theta) = \prod_{i \in N} \rho_i(\theta_i)$ holds for each $\theta \in \Theta$.
For PoA bounds in \Cref{sec:poa}, we assume that $\rho$ is a product distribution.

A Bayesian game proceeds in the following steps:
(1) First, a type profile $\theta$ is generated from $\rho \in \Delta(\Theta)$, and each player $i \in N$ is notified of their own type $\theta_i \in \Theta_i$.
(2) Each player then decides their action $a_i \in A_i$ without knowing the other players' types $\theta_{-i}$.
(3) Finally, each player obtains the payoff $v_i(\theta; a)$ depending on all players' actions and types.

There are two widely-used interpretations of Bayesian games: the strategic form (also called the random-vector model or the normal form) and the agent normal form (also called the posterior-lottery model, the Selten model, or the population interpretation).
Both interpretations reduce a Bayesian game to a normal-form game with complete information as follows.
\begin{itemize}
\item In the strategic form of a Bayesian game, the strategies $S_i$ are considered as a decision space for player $i \in N$.
Each player $i \in N$ obtains the expected payoff $\sum_{\theta \in \Theta} \rho(\theta) v_i(\theta; s(\theta))$ for strategy profile $s \in S$.

\item In the agent normal form, we hypothetically consider the same player with different types as different players, that is, the set of players are $\{ (i,\theta_i) \mid i \in N, ~ \theta_i \in \Theta_i \}$, which we denote by $N'$.
Each player $(i,\theta_i) \in N'$ selects an action from $A_i$.
A type profile $\theta \in \Theta$ is randomly generated from $\rho$, and then only one player $(i,\theta_i)$ becomes active for each $i \in N$.
Only the active players can obtain non-zero payoffs, which are determined only by the active players' actions.
The inactive players always obtain payoff $0$ and do not affect the other players' payoffs.
If we denote player $(i,\theta_i)$'s action by $a_{i,\theta_i}$, the expected payoff for player $(i,\theta_i)$ is $\sum_{\theta' \in \Theta \colon \theta'_i = \theta_i} \rho(\theta') v_i(\theta'; (a_{j,\theta'_j})_{j \in N})$.
\end{itemize}

\section{Further related work}\label{sec:related-work}

\paragraph{Computational studies on correlated equilibria}
There exist many studies for computing a correlated equilibrium in games with \textit{complete} information.
A pioneering study on this topic was given by \citet*{GZ89}, but the time complexity of the proposed algorithm is polynomial in the number of action profiles $|A|$, which is exponential in our setting.
There are mainly two approaches to designing efficient algorithms.
One approach is based on no-regret dynamics.
No-internal-regret dynamics converging to correlated equilibria were developed by \citet*{FV97} and \citet*{HM00} originally in the economics community.
It was later applied to efficiently computing correlated equilibria \citep*{AGTlearning} and extended to stronger regret notions \citep*{BM07,HK07,RST11}.
Another approach is \textit{Ellipsoid Against Hope} proposed by \citet*{PR08} (and modified by \citet*{JL15}), which computes an exact correlated equilibrium of succinctly represented games by the ellipsoid method.
On the hardness side, \citet*{HN18} proved that there is no deterministic or exact algorithm for computing a correlated equilibrium with polynomial query complexity.

\paragraph{Computational studies on Bayes correlated equilibria}
Various definitions of Bayes correlated equilibria can be found in several computer science papers, but their distinctions have not been well explored.
For Bayesian first-price auctions, \citet*{JL23} considered a class of Bayes correlated equilibria, which is equivalent to Bayesian solutions.
\citet*{AGTcrypto} discussed the connection between communication equilibria and multiparty computation.

Some studies mentioned Bayes \textit{coarse} correlated equilibria.
\citet*{HST15} considered no-external-regret dynamics in the agent normal form of Bayesian games and their convergence to agent-normal-form coarse correlated equilibria (ANFCCEs).
Since the intersection of communication equilibria and ANFCEs is included in ANFCCEs, our proposed dynamics converge to a narrower equilibrium concept.
\citet*{CKK15} and \citet*{JL23} considered a variant of ANFCCEs defined on $\Delta(A)^\Theta$, not $\Delta(S)$, for analyzing the PoA in generalized second-price auctions and the price of stability in first-price auctions, respectively.
We briefly discuss definitions of Bayes coarse correlated equilibria in \Cref{sec:bcce}.

\paragraph{Correlated equilibria in extensive-form games with imperfect information}
\citet*{Forges86} extended the notion of communication equilibria to extensive-form games, in which the mediator receives signals from players and then sends signals to players at the beginning of each node of the game tree.
Bayesian games can be regarded as a special case of extensive-form games with imperfect information by treating types as actions of a chance node.
Note that the size of the game tree obtained by this reduction is exponential:
the number of actions of this chance node is $|\Theta|$, and even in the case of complete information, the size of the game tree is at least the number of action profiles $|A|$.
There is a large amount of literature on computing extensions of correlated equilibrium to extensive-form games with imperfect information \citep*{vSF08,CG18,CCG19,FCMG22,ZS22}, but they consider algorithms running in polynomial time in the size of the game tree, which is exponential in our formulation.
\citet*{ZS22} mentioned how to design efficiently computable dynamics converging to communication equilibria as an open problem, which we resolve in the special case of Bayesian games.

The communication procedure for communication equilibria can be interpreted as an extensive-form game with imperfect information.
However, since any concept of correlated equilibria for this extensive-form game allows a ``meta-mediator'' to correlate the type reports of the players, it does not coincide with the setting of communication equilibria.
Hence, existing studies on extensive-form games do not directly apply to communication equilibria in Bayesian games.

\paragraph{Price of anarchy of Bayesian games}
The definition of the PoA for games with complete information was provided by \citet*{KP99} and then extended from Nash equilibria to (coarse) correlated equilibria by \citet*{BHLR08}.
Most existing studies for Bayesian games focus on the PoA of Bayes--Nash equilibria \citep*{CKS16,PT10,Roughgarden15incomplete,Syrgkanis12,ST13,HHT14,FFGL20,JL22}.
An exception is a study by \citet*{HST15}, which extended the smoothness framework for mechanisms developed by \citet*{ST13} to ANFCCEs.

As mentioned above, \citet*{CKK15} and \citet*{JL23} considered a variant of ANFCCEs defined on $\Delta(A)^\Theta$.
\citet*{CKK15} proved PoA bounds of this class for generalized second-price auctions.
For first-price auctions, \citet*{JL23} considered the price of stability, which is defined to be the ratio of the \textit{optimal} social welfare achieved by equilibria to the maximum social welfare.

In \Cref{sec:poa}, we assume that the players' types are independent, i.e., the prior distribution $\rho$ is a product distribution.
\citet*{Roughgarden15incomplete} called the PoA for the correlated case \textit{cPoA} (correlated PoA) and showed a lower bound if the deviation of each player $i \in [n]$ can be determined only by $\theta_i$ in the definition of smoothness.

\paragraph{Bayesian incentive compatible mechanisms}
The concept of communication equilibria is closely connected to Bayesian incentive compatible mechanisms.
Note that \citet*{Myerson82} originally called communication equilibria \textit{coordination mechanisms} with an emphasis on their relations to Bayesian incentive compatible mechanisms.
Although these two concepts are similar in that truthful type reporting is incentive-compatible, a mediator in communication equilibria is required to recommend incentive-compatible actions to players, while a mechanism determines an outcome in a top-down fashion.
A communication equilibrium can be interpreted as a solution concept combining truthful type reporting of Bayesian incentive compatible mechanisms and action recommendations of correlated equilibria.
See \citet*[Chapter 6]{Myerson} for more details of this connection.

\paragraph{Bayesian persuasion and information design}
Bayesian persuasion was proposed as a framework for information design by \citet*{KG11}.
The problem of computing a commnication equilibrium \textit{optimal} for the sender (mediator) in the information assymmetry regime is called \textit{private persuasion} by \citet*{BM16}.
Even for games with complete information, the problem of computing an optimal correlated equilibrium is in general harder than finding a correlated equilibrium \citep*{PR08,BL15}.
Most of the algorithmic studies of Bayesian persuasion make some assumption to alleviate complexity of equilibria:
there is only one player \citep*{DX21,FS22} or each player's action does not affect the other players' payoffs (no-externalities assumption) \citep*{BB17}.

\section{Proof for no-untruthful-swap-regret dynamics}\label{sec:dynamics-proof}

\thmdynamicscom*

\begin{proof}
From \Cref{prop:sr}, $\pi$ is strategy-representable.
From \Cref{prop:com-sr}, it is sufficient to prove
\begin{equation*}
\E_{\theta \sim \rho} \left[ \E_{a \sim \pi(\psi(\theta_i),\theta_{-i})} \left[ v_i(\theta; \phi(\theta_i, a_i), a_{-i}) \right] \right]
-
\E_{\theta \sim \rho} \left[ \E_{a \sim \pi(\theta)} \left[ v_i(\theta; a) \right] \right]
\le \frac{\max_{j \in N} \RUSj^T}{T}
\end{equation*}
for each $i \in N$, $\psi \colon \Theta_i \to \Theta_i$, and $\phi \colon \Theta_i \times A_i \to A_i$.
For each $t \in [T]$ and $\theta \in \Theta$, let $\pi^t(\theta) \in \Delta(A)$ be the product distribution that independently generates $a_j \sim \pi_j^t(\theta_j)$ for each $j \in N$.
Similarly, let $\pi_{-i}^t(\theta) \in \Delta(A_{-i})$ be the product distribution that independently generates $a_j \sim \pi_j^t(\theta_j)$ for each $j \in N \setminus \{i\}$.
The left-hand side can be bounded as
\begin{align*}
&\E_{\theta \sim \rho} \left[ \E_{a \sim \pi(\psi(\theta_i), \theta_{-i})} \left[ v_i(\theta; \phi(\theta_i, a_i), a_{-i}) \right] \right] - \E_{\theta \sim \rho} \left[ \E_{a \sim \pi(\theta)} \left[ v_i(\theta; a) \right] \right]\\
&=
\E_{\theta \sim \rho} \left[ \E_{t \sim [T]} \left[ \E_{a \sim \pi^t(\psi(\theta_i), \theta_{-i})} \left[ v_i(\theta; \phi(\theta_i, a_i), a_{-i}) \right] - \E_{a \sim \pi^t(\theta)} \left[ v_i(\theta; a) \right] \right] \right]\\
&=
\frac{1}{T} \sum_{t=1}^T \E_{\theta_i \sim \rho_i} \left[ \E_{\theta_{-i} \sim \rho|\theta_i} \left[ \E_{a_{-i} \sim \pi_{-i}^t(\theta_{-i})} \left[ \E_{a_i \sim \pi_i^t(\psi(\theta_i))} \left[ v_i(\theta; \phi(\theta_i,a_i),a_{-i}) \right]
- 
\E_{a_i \sim \pi_i^t(\theta_i)} \left[ v_i(\theta; a)\right] \right] \right] \right] \tag{since $\pi^t(\theta)$ is a product distribution}\\
&=
\frac{1}{T} \sum_{t=1}^T \E_{\theta_i \sim \rho_i} \left[ \E_{a_i \sim \pi_i^t(\psi(\theta_i))} \left[ u_i^t(\theta_i, \phi(\theta_i, a_i)) \right] - \E_{a_i \sim \pi_i^t(\theta_i)} \left[ u_i^t(\theta_i, a_i) \right] \right] \tag{from the definition of $u_i^t$}\\
&\le
\frac{\RUSi^T}{T}\\
&\le \frac{\max_{j \in N} \RUSj^T}{T},
\end{align*}
which completes the proof.
\end{proof}

\section{Proofs for untruthful swap regret minimization}\label{sec:ub-proof}

\RUSiPhi*

\begin{proof}
The expected reward for the algorithm is
\begin{align*}
\sum_{t=1}^T \E_{\theta_i \sim \rho_i} \left[ \E_{a_i \sim \pi_i^t(\theta_i)} \left[ u_i^t(\theta_i, a_i) \right] \right]
=
\sum_{t=1}^T \sum_{\theta_i \in \Theta_i} \rho_i(\theta_i) \sum_{a_i \in A_i} \pi_i^t(\theta_i; a_i) u_i^t(\theta_i, a_i)
=
\sum_{t=1}^T \langle x^t, \bar{u}^t \rangle.
\end{align*}

Next, we express the competitor's transformation that uses $\psi \colon \Theta_i \to \Theta_i$ and $\phi \colon \Theta_i \times A_i \to A_i$ as a linear transformation.
Given $\pi^t_i \in \Delta(A_i)^{\Theta_i}$, the competitor's expected reward in each round $t \in [T]$ is
\begin{align*}
\E_{\theta_i \sim \rho_i} \left[ \E_{a_i \sim \pi_i^t(\psi(\theta_i))} \left[ u_i^t(\theta_i, \phi(\theta_i, a_i)) \right] \right]
&=
\sum_{\theta_i \in \Theta_i} \rho_i(\theta_i) \sum_{a_i \in A_i} \pi_i^t(\psi(\theta_i); a_i) u_i^t(\theta_i, \phi(\theta_i, a_i))\\
&=
\sum_{\theta_i \in \Theta_i} \sum_{a_i \in A_i} x^t(\psi(\theta_i), a_i) \bar{u}^t(\theta_i, \phi(\theta_i, a_i))\\
&=
\left\langle Q_\psi x^t, Q_\phi \bar{u}^t \right\rangle,
\end{align*}
where $Q_\psi, Q_\phi \in [0,1]^{(\Theta_i \times A_i) \times (\Theta_i \times A_i)}$ are the matrices defined as follows.
\begin{itemize}
\item The linear transformation $Q_\psi$ for $x^t$ is defined as the Kronecker product $Q_\psi = W_\psi \otimes I$ of the zero-one stochastic matrix $W_\psi \in \{0,1\}^{\Theta_i \times \Theta_i}$, where $W_\psi(\theta_i,\theta'_i) = 1$ if $\theta'_i = \psi(\theta_i)$ and $0$ otherwise, and the identity matrix $I \in \{0,1\}^{A_i \times A_i}$.
\item The linear transformation $Q_\phi$ for $\bar{u}^t$ is defined as the block diagonal matrix, where each block $Q_{\phi(\theta_i,\cdot)} \in \{0,1\}^{A_i \times A_i}$ corresponding to $\theta_i \in \Theta_i$ is the zero-one stochastic matrix defined by $Q_{\phi(\theta_i,\cdot)}(a_i,a'_i) = 1$ if $a'_i = \phi(\theta_i,a_i)$ and $0$ otherwise.
\end{itemize}
Since $\left\langle Q_\psi x^t, Q_\phi \bar{u}^t \right\rangle = \left\langle \left( Q_\phi^\top Q_\psi \right) x^t, \bar{u}^t \right\rangle$, the competitor can be interpreted as transforming $x^t$ by using matrix $Q_\phi^\top Q_\psi$, which we denote by $Q_{\psi,\phi}$.
Hence, the set of all possible transformations for the competitor is
\begin{equation*}
\calQ' = \left\{ Q_{\psi,\phi} \;\middle|\; \psi \colon \Theta_i \to \Theta_i, \phi \colon \Theta_i \times A_i \to A_i \right\}.
\end{equation*}

Next, we prove that $\calQ$ equals the convex hull of $\calQ'$.
For each $\psi$ and $\phi$, by a simple calculation, we can see that each entry is $Q_{\psi,\phi}((\theta_i,a_i),(\theta'_i,a'_i)) = 1$ if $\theta'_i = \psi(\theta_i)$ and $a_i = \phi(\theta_i,a'_i)$, and $0$ otherwise.
All the constraints \eqref{eq:calQ} for $\calQ$ are satisfied by $Q_{\psi,\phi} \in \calQ$ for any $\psi$ and $\phi$.
The convex hull is therefore a subset of $\calQ$.

Conversely, we show that any $Q \in \calQ$ can be expressed as a convex combination of $Q_{\psi,\phi}$.
Fix any $Q \in \calQ$.
We can regard $Q$ as a block matrix with block $Q_{\theta_i,\theta'_i} \in [0,1]^{A_i \times A_i}$ for each $\theta_i,\theta'_i \in \Theta_i$.
From the definition of $\calQ$, there exists some stochastic matrices $W \in [0,1]^{\Theta_i \times \Theta_i}$ and $\tilde{Q}_{\theta_i,\theta'_i} \in [0,1]^{A_i \times A_i}$ for each $\theta_i,\theta'_i \in \Theta_i$ such that $Q_{\theta_i,\theta'_i}^\top = W(\theta_i,\theta'_i) \tilde{Q}_{\theta_i,\theta'_i}$ for each $\theta_i,\theta'_i \in \Theta_i$.
Since any stochastic matrix is a convex combination of zero-one stochastic matrices, there exists some distribution $\gamma$ over all possible $\psi$ such that $W = \sum_{\psi} \gamma(\psi) W_\psi$.
Similarly, for each $\theta_i,\theta'_i \in \Theta_i$, there exists some distribution $\kappa_{\theta_i,\theta'_i}$ over all possible $\phi(\theta_i,\cdot)$ such that $\tilde{Q}_{\theta_i,\theta'_i} = \sum_{\phi(\theta_i,\cdot)} \kappa_{\theta_i,\theta'_i}(\phi(\theta_i,\cdot)) Q_{\phi(\theta_i,\cdot)}$.
Each entry of $Q$ is
\begin{align*}
Q((\theta_i,a_i),(\theta'_i,a'_i))
&=
W(\theta_i,\theta'_i) \tilde{Q}_{\theta_i,\theta'_i}(a'_i,a_i)\\
&=
\left( \sum_{\psi} \gamma(\psi) W_\psi(\theta_i,\theta'_i) \right) \left( \sum_{\phi(\theta_i,\cdot)} \kappa_{\theta_i,\theta'_i}(\phi(\theta_i,\cdot)) Q_{\phi(\theta_i,\cdot)} (a'_i,a_i) \right)\\
&=
\left( \sum_{\psi} \gamma(\psi) \bfone_{\{\theta'_i = \psi(\theta_i) \}} \right) \left( \sum_{\phi(\theta_i,\cdot)} \kappa_{\theta_i,\theta'_i}(\phi(\theta_i,\cdot)) \bfone_{\{a_i = \phi(\theta_i,a'_i) \}} \right)\\
&=
\Pr_{\psi \sim \gamma} \left( \theta'_i = \psi(\theta_i) \right) \Pr_{\phi(\theta_i,\cdot) \sim \kappa_{\theta_i,\theta'_i}} \left( a_i = \phi(\theta_i,a'_i) \right).
\end{align*}
We consider a distribution that generates $\psi$ and $\phi$ with probability $\gamma({\psi}) \prod_{\theta_i \in \Theta_i} \kappa_{\theta_i,\psi(\theta_i)}(\phi(\theta_i,\cdot))$.
That is, $\psi$ is generated from $\gamma$ and then each $\phi(\theta_i,\cdot)$ is generated from $\kappa_{\theta_i,\psi(\theta_i)}$ independently for each $\theta_i \in \Theta_i$.
The expected value of each entry according to this distribution is
\begin{align*}
\E_{\psi,\phi} \left[ Q_{\psi,\phi}((\theta_i,a_i),(\theta'_i,a'_i)) \right]
&=
\Pr_{\psi,\phi}\left( \theta'_i = \psi(\theta_i), a_i = \phi(\theta_i,a'_i)\right)\\
&=
\Pr_{\psi \sim \gamma} \left( \theta'_i = \psi(\theta_i) \right) \Pr_{\phi(\theta_i,\cdot) \sim \kappa_{\theta_i,\theta'_i}} \left( a_i = \phi(\theta_i,a'_i)\right)\\
&=
Q((\theta_i,a_i),(\theta'_i,a'_i)).
\end{align*}
Since every $Q \in \calQ$ can be expressed as a convex combination of some matrices in $\calQ'$, we can see that $\calQ$ is a subset of the convex hull of $\calQ'$.
We conclude that $\calQ$ equals the convex hull of $\calQ'$.

The competitor's expected reward is
\begin{align*}
&\max_{\psi \colon \Theta_i \to \Theta_i} \max_{\phi \colon \Theta_i \times A_i \to A_i} \sum_{t=1}^T \E_{\theta_i \sim \rho_i} \left[ \E_{a_i \sim \pi_i^t(\psi(\theta_i))} \left[ u_i^t(\theta_i, \phi(\theta_i, a_i)) \right] \right]\\
&=
\max_{\psi \colon \Theta_i \to \Theta_i} \max_{\phi \colon \Theta_i \times A_i \to A_i} \sum_{t=1}^T \left\langle Q_{\psi,\phi} x^t, \bar{u}^t \right\rangle\\
&=
\max_{Q \in \calQ} \sum_{t=1}^T \langle Q x^t, \bar{u}^t \rangle,
\end{align*}
where the last equality holds because $\langle Q x^t, \bar{u}^t \rangle$ is linear in terms of $Q$, and $\calQ$ is the convex hull of $Q_{\psi,\phi}$ for all $\psi$ and $\phi$.
\end{proof}

\lemfixedpoint*

\begin{proof}
We prove this lemma by using Brouwer's fixed point theorem, which claims that any continuous function that maps a compact convex set to itself has a fixed point.
Since $\calX$ is a product of probability simplices, it is compact and convex.
In addition, $Q$ is continuous since it is a matrix.
Thus it is sufficient to prove that $Q$ maps every $x \in \calX$ to a vector in $\calX$.
Let $x \in \calX$ be an arbitrary vector.
Let $x_{\theta_i}$ be the block of $x$ corresponding to $\theta_i$ and $Q_{\theta_i,\theta'_i}$ be the block of $Q$ corresponding to $\theta_i,\theta'_i$.
Then $x_{\theta_i} \in \Delta(A_i)$ is a stochastic vector for each $\theta_i \in \Theta_i$.
From the definition of $\calQ$, there exists some stochastic matrices $W \in [0,1]^{\Theta_i \times \Theta_i}$ and $\tilde{Q}_{\theta_i,\theta'_i} \in [0,1]^{A_i \times A_i}$ for each $\theta_i,\theta'_i$ such that $Q_{\theta_i,\theta'_i}^\top = W(\theta_i,\theta'_i) \tilde{Q}_{\theta_i,\theta'_i}$ for each $\theta_i,\theta'_i \in \Theta_i$.
Since $\tilde{Q}_{\theta_i,\theta'_i}^\top x_{\theta'_i} \in \Delta(A_i)$ is also a stochastic vector, each block of $Q x$ corresponding to $\theta_i \in \Theta_i$ is
\begin{equation*}
\sum_{\theta'_i \in \Theta_i} Q_{\theta_i,\theta'_i} x_{\theta'_i}
= \sum_{\theta'_i \in \Theta_i} W(\theta_i,\theta'_i) \left( \tilde{Q}_{\theta_i,\theta'_i}^\top x_{\theta'_i} \right).
\end{equation*}
This is a convex sum of stochastic vectors, hence a stochastic vector in $\Delta(A_i)$.
This implies $Q x \in \calX$ for any $x \in \calX$.
From Brouwer's fixed point theorem, there exists a fixed point $x \in \calX$ such that $Q x = x$.
\end{proof}

\lemeigen*

\begin{proof}
We apply the Perron--Frobenius theorem, which claims that any positive matrix has an eigenvalue whose eigenspace is $1$-dimensional, and any positive eigenvector is an eigenvector corresponding to this eigenvalue.
From \Cref{lem:fixed-point}, there exists some $x \in \calX$ such that $Q x = x$, which implies that $Q$ has eigenvector $1$, and $x$ is its corresponding eigenvector.
Since $Q$ is a positive matrix and $x$ is a non-negative non-zero vector, $Qx=x$ is also a positive vector.
From the Perron--Frobenius theorem, the eigenspace corresponding to eigenvalue $1$ is $1$-dimensional and contains only a multiple of $x \in \calX$.
We can obtain $x \in \calX$ by computing an eigendecomposition of $Q$ and appropriately scaling the eigenvector corresponding to eigenvalue $1$.
Since $Q$ is a $|\Theta_i||A_i| \times |\Theta_i||A_i|$ matrix, its eigendecomposition can be computed in $O(|\Theta_i|^\omega |A_i|^\omega)$ time, where $\omega$ is the exponent of matrix multiplication.
\end{proof}

\lemreductionone*

\begin{proof}
From the definition of $x^t$, it holds that $Q^t x^t = x^t$ for each $t \in [T]$.
From the definition of $U^t$, the algorithm's expected reward is
\begin{align*}
\sum_{t=1}^T \langle Q^t, U^t \rangle
= 
\sum_{t=1}^T \left\langle Q^t, \bar{u}^t \left( x^t \right)^\top \right\rangle
=
\sum_{t=1}^T \left\langle Q^t x^t, \bar{u}^t \right\rangle
=
\sum_{t=1}^T \left\langle x^t, \bar{u}^t \right\rangle,
\end{align*}
where the last equality is due to $Q^t x^t = x^t$ for each $t \in [T]$.
Similarly, the competitor's expected reward is
\begin{align*}
\max_{Q \in \calQ} \sum_{t=1}^T \langle Q, U^t \rangle
=
\max_{Q \in \calQ} \sum_{t=1}^T \left\langle Q, \bar{u}^t \left( x^t \right)^\top \right\rangle
=
\max_{Q \in \calQ} \sum_{t=1}^T \left\langle Q x^t, \bar{u}^t \right\rangle.
\end{align*}
Therefore, we obtain
\begin{align*}
R^T_\calQ
=
\max_{Q \in \calQ} \sum_{t = 1}^T \langle Q, U^t \rangle - \sum_{t = 1}^T \langle Q^t, U^t \rangle
=
\max_{Q \in \calQ} \sum_{t = 1}^T \langle Q x^t, \bar{u}^t \rangle - \sum_{t = 1}^T \langle x^t, \bar{u}^t \rangle
=
\RUSi,
\end{align*}
where the last equality is due to \Cref{lem:RUSi-Phi}.
\end{proof}

\lemreductiontwo*

\begin{proof}
Recall that for each $\calE_{\theta_i,\theta'_i,a'_i}$, the reward for decision $a_i \in A_i$ is $x^t(\theta'_i,a'_i) \bar{u}_i^t(\theta_i,a_i)$.
Therefore, the external regret $R_{\theta_i,\theta'_i,a'_i}^T$ for $\calE_{\theta_i,\theta'_i,a'_i}$ is defined as
\begin{equation*}
R_{\theta_i,\theta'_i,a'_i}^T = \max_{a^*_i \in A_i} \sum_{t=1}^T x^t(\theta'_i,a'_i) \bar{u}^t(\theta_i,a^*_i) - \sum_{t=1}^T \sum_{a_i \in A_i} y^t_{\theta_i,\theta'_i,a'_i}(a_i) x^t(\theta'_i,a'_i) \bar{u}^t(\theta_i,a_i).
\end{equation*}
Since for each $\calE_{\theta_i}$, the reward for decision $\theta'_i \in \Theta_i$ is $\sum_{a_i,a'_i \in A_i} y^t_{\theta_i,\theta'_i,a'_i}(a_i) x^t(\theta'_i,a'_i) \bar{u}^t(\theta_i,a_i)$,
the external regret $R_{\theta_i}^T$ for $\calE_{\theta_i}$ is defined as
\begin{align*}
R_{\theta_i}^T
&=
\max_{\theta^*_i \in \Theta_i} \sum_{t=1}^T \sum_{a_i,a'_i \in A_i} y^t_{\theta_i,\theta^*_i,a'_i}(a_i) x^t(\theta^*_i,a'_i) \bar{u}^t(\theta_i,a_i) \\
&-
\sum_{t=1}^T \sum_{\theta'_i \in \Theta_i} w_{\theta_i}(\theta'_i) \sum_{a_i,a'_i \in A_i} y^t_{\theta_i,\theta'_i,a'_i}(a_i) x^t(\theta'_i,a'_i) \bar{u}^t(\theta_i,a_i) .
\end{align*}

Recall that $U^t = \bar{u}^t (x^t)^\top$.
Since each extreme point of $\calQ$ corresponds to a matrix $(Q_\phi^\top Q_\psi)$ for some $\psi$ and $\phi$, we obtain
\begin{align*}
&\max_{Q \in \calQ} \sum_{t = 1}^T \langle Q, U^t \rangle\\
&=
\max_{Q \in \calQ} \sum_{t = 1}^T \langle Q x^t, \bar{u}^t \rangle\\
&=
\max_{\psi \colon \Theta_i \to \Theta_i} \max_{\phi \colon \Theta_i \times A_i \to A_i} \sum_{t = 1}^T \langle Q_\psi x^t, Q_\phi \bar{u}^t \rangle\\
&= \max_{\psi \colon \Theta_i \to \Theta_i} \max_{\phi \colon \Theta_i \times A_i \to A_i} \sum_{\theta_i \in \Theta_i} \sum_{a'_i \in A_i} \sum_{t=1}^T x^t(\psi(\theta_i),a'_i) \bar{u}^t(\theta_i,\phi(\theta_i, a'_i))\\
&= \sum_{\theta_i \in \Theta_i} \max_{\theta'_i \in \Theta_i} \sum_{a'_i \in A_i} \max_{a_i \in A_i} \sum_{t=1}^T x^t(\theta'_i,a'_i) \bar{u}^t(\theta_i,a_i)\\
&= \sum_{\theta_i \in \Theta_i} \max_{\theta'_i \in \Theta_i} \sum_{a'_i \in A_i} \left\{ \sum_{t=1}^T \sum_{a_i \in A_i} y^t_{\theta_i,\theta'_i,a'_i}(a_i) x^t(\theta'_i,a'_i) \bar{u}^t(\theta_i,a_i) + R^T_{\theta_i,\theta'_i,a'_i} \right\}
\tag{due to the definition of $R^T_{\theta_i,\theta'_i,a'_i}$}\\
&\le \sum_{\theta_i \in \Theta_i} \left\{ \max_{\theta'_i \in \Theta_i} \sum_{t=1}^T \sum_{a'_i \in A_i} \sum_{a_i \in A_i} y^t_{\theta_i,\theta'_i,a'_i}(a_i) x^t(\theta'_i,a'_i) \bar{u}^t(\theta_i,a_i) + \max_{\theta'_i \in \Theta_i} \sum_{a'_i \in A_i} R^T_{\theta_i,\theta'_i,a'_i} \right\}\\
&= \sum_{\theta_i \in \Theta_i} \left\{ \sum_{t=1}^T \sum_{\theta'_i \in \Theta_i}  w^t_{\theta_i}(\theta'_i) \sum_{a'_i \in A_i} \sum_{a_i \in A_i} y^t_{\theta_i,\theta'_i,a'_i}(a_i) x^t(\theta'_i,a'_i) \bar{u}^t(\theta_i,a_i) + R^T_{\theta_i} + \max_{\theta'_i \in \Theta_i} \sum_{a'_i \in A_i} R^T_{\theta_i,\theta'_i,a'_i} \right\}
\tag{due to the definition of $R^T_{\theta_i}$}\\
&= \sum_{t=1}^T \langle Q^t, U^t \rangle + \sum_{\theta_i \in \Theta_i} R^T_{\theta_i} + \sum_{\theta_i \in \Theta_i} \max_{\theta'_i \in \Theta_i} \sum_{a'_i \in A_i} R^T_{\theta_i,\theta'_i,a'_i}, \tag{due to the definition of $Q^t$ and $U^t$}
\end{align*}
which yields an upper bound on $R^T_\calQ = \max_{Q \in \calQ} \sum_{t = 1}^T \langle Q, U^t \rangle - \sum_{t = 1}^T \langle Q^t, U^t \rangle$.
\end{proof}

To obtain an upper bound on untruthful swap regret, we use the following upper bounds for external regret minimization.
These algorithms are based on the multiplicative weights, and their output is always a positive vector.
\begin{theorem}[{\citep*[Theorem 2.2]{PLG}}]\label{thm:mwu}
For an online learning problem with rewards in $[0,1]$, there exists an algorithm such that its external regret is bounded above by $\sqrt{\frac{1}{2} T \log d}$, where $d$ is the number of possible decisions and $T$ is the number of rounds.
\end{theorem}

\begin{theorem}[see, e.g., {\citep*[Section 7.6]{Orabona19}}]\label{thm:adahedge}
For an online learning problem with reward vectors $u^1,u^2,\dots,u^T \in [0,1]$, there exists an algorithm such that its external regret is bounded above by $6 \sqrt{\left( \sum_{t=1}^T \|u^t\|_\infty \right) \log d}$, where $d$ is the number of possible decisions.%
\end{theorem}

\thmuntruthfulub*

\begin{proof}
From \Cref{lem:reduction1,lem:reduction2}, the untruthful swap regret of \Cref{alg:untruthful} is bounded as
\begin{equation*}
\RUSi^T = R^T_\calQ \le \sum_{\theta_i \in \Theta_i} R^T_{\theta_i} + \sum_{\theta_i \in \Theta_i} \max_{\theta'_i \in \Theta_i} \sum_{a'_i \in A_i} R^T_{\theta_i,\theta'_i,a'_i}.
\end{equation*}
Since the reward input to $\calE_{\theta_i}$ is bounded above by
\begin{equation*}
\max_{t \in [T]} \max_{\theta'_i \in \Theta_i} \sum_{a_i,a'_i \in A_i} y^t_{\theta_i,\theta'_i,a'_i}(a_i) x^t(\theta'_i,a'_i) \bar{u}^t(\theta_i,a_i)
\le
\max_{t \in [T]} \max_{\theta'_i \in \Theta_i} \sum_{a_i,a'_i \in A_i} y^t_{\theta_i,\theta'_i,a'_i}(a_i) x^t(\theta'_i,a'_i) \rho_i(\theta_i)
=
\rho_i(\theta_i),
\end{equation*}
we can obtain the regret upper bound $R^T_{\theta_i} \le \rho_i(\theta_i) \sqrt{\frac{1}{2} T \log |\Theta_i|}$ from \Cref{thm:mwu}.
Since the reward input to $\calE_{\theta_i,\theta'_i,a'_i}$ is similarly bounded above by $\rho_i(\theta_i)$ and the sum of the maximum rewards is bounded above as
\begin{equation*}
\sum_{t=1}^T \max_{a_i \in A_i} \pi_i^t(\theta'_i; a'_i) \bar{u}_i^t(\theta_i,a_i)
\le
\rho_i(\theta_i) \sum_{t=1}^T \pi_i^t(\theta'_i;a'_i),
\end{equation*}
we can obtain a regret upper bound $R^T_{\theta_i,\theta'_i,a'_i} = 6 \rho_i(\theta_i) \sqrt{\sum_{t=1}^T \pi_i^t(\theta'_i;a'_i) \log |A_i|}$.
By the Cauchy--Schwarz inequality, we have $\sum_{a'_i \in A_i} 1 \cdot \sqrt{\sum_{t=1}^T \pi_i^t(\theta'_i;a'_i)} \le \sqrt{|A_i| \sum_{a'_i \in A_i} \sum_{t=1}^T \pi_i^t(\theta'_i;a'_i)} = \sqrt{|A_i|T}$.
We thus obtain
\begin{align*}
&\sum_{\theta_i \in \Theta_i} R^T_{\theta_i} + \sum_{\theta_i \in \Theta_i} \max_{\theta'_i \in \Theta_i} \sum_{a'_i \in A_i} R^T_{\theta_i,\theta'_i,a'_i}\\
&\le \sum_{\theta_i \in \Theta_i} \rho_i(\theta_i) \sqrt{\frac{1}{2} T \log |\Theta_i|} + \sum_{\theta_i \in \Theta_i} \max_{\theta'_i \in \Theta_i} \rho_i(\theta_i) \left( 6 \sqrt{T |A_i| \log |A_i|} \right)\\
&= \sqrt{\frac{1}{2} T \log |\Theta_i|} + 6 \sqrt{T |A_i| \log |A_i|}.
\end{align*}
\end{proof}

\corcom*

\begin{proof}
We estimate the value of $u_i^t(\theta_i,a_i)$ for each $i \in N$, $t \in [T]$, $\theta_i \in \Theta_i$, and $a_i \in A_i$ by generating $\frac{8}{\epsilon^2} \log \frac{2nT\max_{i \in N}|\Theta_i||A_i|}{\delta}$ samples of $\theta_{-i} \sim \rho|\theta_i$ and $a_{-i} \sim \pi^t_{-i}(\theta_{-i})$ and taking the average of $v_i(\theta_i,\theta_{-i};a_i,a_{-i})$, where $T$ will be specified later.
Let $\tilde{u}_i^t(\theta_i,a_i)$ be this estimation.
By using H\oe{}ffding's inequality, we have
\begin{align*}
\Pr\left( \left| u_i^t(\theta_i,a_i) - \tilde{u}_i^t(\theta_i,a_i) \right| \ge \frac{\epsilon}{4} \right)
&\le 2 \exp \left( - \frac{2 \epsilon^2}{16} \cdot \frac{8}{\epsilon^2} \log \frac{2nT\max_{i \in N}|\Theta_i||A_i|}{\delta} \right)\\
&\le \frac{\delta}{nT \max_{i \in N} (|\Theta_i||A_i|)}
\end{align*}
for each $i \in N$, $t \in [T]$, $\theta_i \in \Theta_i$, and $a_i \in A_i$.
By the union bound, with probability at least $1-\delta$, the additive errors of all the rewards are at most $\epsilon/4$.

We apply \Cref{alg:untruthful} to these estimated rewards.
For each player $i \in N$, by using the upper bound on untruthful swap regret in \Cref{thm:untruthful-ub}, we obtain
\begin{align*}
&\frac{\RUSi}{T}
= \frac{1}{T} \max_{\psi \colon \Theta_i \to \Theta_i} \max_{\phi \colon \Theta_i \times A_i \to A_i} \sum_{t=1}^T \E_{\theta_i \sim \rho_i} \left[ \E_{a_i \sim \pi_i^t(\psi(\theta_i))} \left[ u_i^t(\theta_i, \phi(\theta_i, a_i)) \right] - \E_{a_i \sim \pi_i^t(\theta_i)} \left[ u_i^t(\theta_i, a_i) \right] \right]\\
&\le
\frac{1}{T} \left\{ \max_{\psi \colon \Theta_i \to \Theta_i} \max_{\phi \colon \Theta_i \times A_i \to A_i} \sum_{t=1}^T \E_{\theta_i \sim \rho_i} \left[ \E_{a_i \sim \pi_i^t(\psi(\theta_i))} \left[ \tilde{u}_i^t(\theta_i, \phi(\theta_i, a_i)) \right] - \E_{a_i \sim \pi_i^t(\theta_i)} \left[ \tilde{u}_i^t(\theta_i, a_i) \right] \right] + T \cdot \frac{\epsilon}{2} \right\}\\
&\le \sqrt{\frac{1}{2T} \log |\Theta_i|} + 6 \sqrt{\frac{|A_i| \log |A_i|}{T}} + \frac{\epsilon}{2}\\
&\le \epsilon,
\end{align*}
where we set $T = \max\{ \frac{18}{\epsilon^2} \log(\max_{i \in N}|\Theta_i|), \frac{144^2}{\epsilon^2} \max_{i \in N}|A_i| \log |A_i| \}$.
From \Cref{thm:dynamics-com}, $\pi \in \PiCom^\epsilon \cap \PiANF^\epsilon$ holds.
\end{proof}

\section{Proof of lower bound for untruthful swap regret}\label{sec:lb-proof}
Here, we provide a full proof of the $\Omega(\sqrt{T \log |\Theta_i|})$ lower bound on untruthful swap regret (\Cref{thm:untruthful-lb}).
This section focuses on an online learning problem with stochastic types for a single player $i \in N$.
The subscript $i$ for $A_i$, $a_i$, $\Theta_i$, $\theta_i$, $\pi_i^t$, and $\RUSi^T$ is not essential in this section, but we put it for notational consistency.

We want to prove that for any randomized algorithm, there exists a deterministic adversary for which the algorithm's untruthful swap regret is lower bounded.
From Yao's minimax principle, it is sufficient to show there exists a randomized adversary for which any deterministic algorithm's expected untruthful swap regret is lower bounded.
In the following, we construct such a randomized adversary.

Let $A_i = \{\alpha_0,\alpha_1\}$ be the set of actions and $\rho_i \in \Delta(\Theta_i)$ the uniform distribution over $\Theta_i$.
The set of time rounds $[T]$ is partitioned into $B$ blocks of equal length $L = T/B$, where $B = \log_2 |\Theta_i|-1$.
For each $b \in [B]$, we denote the $b$th block by $\calT_b = \{ t \in [T] \mid (b-1)L < t \le bL \}$.

The set of types $\Theta_i$ is partitioned into $\Theta'_i$ of size $2^B$ and $\Theta''_i$ of size $2^B$.
The rewards for types $\Theta'_i$ are determined by a bijection $\zeta \colon \Theta'_i \to \{0,1\}^B$ that associates each type in $\theta'_i \in \Theta'_i$ with binary sequence $\zeta(\theta'_i)$ of length $B$.
A bijection $\zeta$ is latently generated from the uniform distribution over all such bijections and unknown to the algorithm in the beginning.
For each time round $t \in \calT_b$ in the $b$th block, the reward $u_i^t(\theta_i,\alpha_0)$ is defined by $\zeta$ such that $u_i^t(\theta'_i,\alpha_0) = \zeta(\theta'_i)(b)$ for each $\theta'_i \in \Theta'_i$.
On the other hand, 
the rewards for types $\Theta''_i$ are determined by a random map $\xi \colon \Theta''_i \to \{0,1\}^T$.
For each $\theta''_i \in \Theta''_i$ and $t \in [T]$, the reward $u_i^t(\theta''_i,\alpha_0) = \xi(\theta''_i)(t)$ independently follows the uniform distribution over $\{0,1\}$.
For both $\Theta'_i$ and $\Theta''_i$, the reward for $\alpha_1$ is defined by flipping the reward for $\alpha_0$, i.e., $u_i^t(\theta_i,\alpha_1) = 1 - u_i^t(\theta_i,\alpha_0)$ for every type $\theta_i \in \Theta_i$ and round $t \in [T]$.
See \Cref{fig:hard-instance} for an illustration.

Once $\zeta$ and $\xi$ are fixed, the problem instance is determined.
Then, since the algorithm is deterministic, the algorithm's decisions $(\pi_i^t)_{t \in [T]}$ are also determined.
Furthermore, since in each round $t \in [T]$, the algorithm decides $\pi_i^t$ according to the rewards observed so far,
$\pi_i^t$ is deterministic if $(\zeta(\theta'_i)(b))_{\theta'_i \in \Theta'_i, b \in [\lceil (t-1)/L \rceil]}$ and $(\xi(\theta''_i)(t'))_{\theta''_i \in \Theta''_i, t' \in [t-1]}$ are fixed.

The untruthful swap regret $\RUSi^T$ depends on the problem instance parametrized by $\zeta$ and $\xi$.
We will develop an $\Omega(\sqrt{T \log |\Theta_i|})$ lower bound on the expected value $\E_{\zeta,\xi} \left[\RUSi^T\right]$.
This lower bound implies that there exists some $\zeta$ and $\xi$ such that $\RUSi^T = \Omega(\sqrt{T \log |\Theta_i|})$, which leads to the theorem.

\subsection{Analysis for randomly branching rewards}\label{sec:lb1}

The goal of this subsection is to prove that the algorithm with low untruthful swap regret must make significantly different decisions for types $\Theta'_i$ (\Cref{lem:gap-block}).
Toward this goal, we first show that the algorithm must select $\alpha_0$ or $\alpha_1$ for the type with all $0$ or all $1$ binary sequence, respectively (\Cref{lem:edge-type}).
We then show that the algorithm must make significantly different decisions for ``adjacent'' types (\Cref{lem:gap-type}).
By taking the expectation over $\zeta$ and $\xi$, we prove \Cref{lem:gap-block}.

First, we prove that if the untruthful swap regret is small, for the type $\theta^0_i \in \Theta'_i$ with $\zeta(\theta^0_i) = 00 \cdots 00$, the algorithm must choose $\alpha_0$ for most of the rounds.
Similarly, for the type $\theta^1_i \in \Theta'_i$ with $\zeta(\theta^1_i) = 11 \cdots 11$, the algorithm must choose $\alpha_1$ for most of the rounds.
\begin{lemma}\label{lem:edge-type}
Fix $\zeta$ and $\xi$.
For $\theta^0_i \in \Theta'_i$ such that $\zeta(\theta^0_i)(b) = 0$ for all $b \in [B]$, it must hold that
$\sum_{t=1}^T \pi_i^t(\theta^0_i;\alpha_0) \ge T - |\Theta_i| \RUSi^T$.
For $\theta^1_i \in \Theta'_i$ such that $\zeta(\theta^1_i)(b) = 1$ for all $b \in [B]$, it must hold that
$\sum_{t=1}^T \pi_i^t(\theta^1_i;\alpha_1) \ge T - |\Theta_i| \RUSi^T$.
\end{lemma}

\begin{proof}
Let $\theta^0_i \in \Theta'_i$ be the type satisfying $\zeta(\theta^0_i)(b) = 0$ for all $b \in [B]$.
We consider $\psi$ and $\phi$ in the definition of untruthful swap regret such that $\psi(\theta^0_i) = \theta^0_i$ and $\phi(\theta^0_i,\alpha_0) = \phi(\theta^0_i,\alpha_1) = \alpha_0$.
Let $\psi(\theta_i) = \theta_i$ and $\phi(\theta_i,\alpha) = \alpha$ for all the other types $\theta_i \in \Theta_i$ and any action $\alpha \in A_i$.
We thus obtain
\begin{align*}
\RUSi^T
&= \max_{\psi \colon \Theta_i \to \Theta_i} \max_{\phi \colon \Theta_i \times A_i \to A_i} \sum_{t=1}^T \E_{\theta_i \sim \rho_i} \left[ \E_{\alpha \sim \pi_i^t(\psi(\theta_i))} \left[ u_i^t(\theta, \phi(\theta_i, \alpha)) \right] - \E_{\alpha \sim \pi_i^t(\theta_i)} \left[ u_i^t(\theta_i, \alpha) \right] \right] \\
&\ge \frac{1}{|\Theta_i|} \sum_{t=1}^T \left\{ u_i^t(\theta^0_i, \alpha_0) - \E_{\alpha \sim \pi_i^t(\theta^0_i)} \left[ u_i^t(\theta^0_i, \alpha) \right] \right\}.
\end{align*}
Since $\zeta(\theta^0_i)(b) = 0$ for all $b \in [B]$, we have $u_i^t(\theta^0_i,\alpha_0) = 1$ for all $t \in [T]$, which implies
\begin{equation*}
\RUSi^T
\ge \frac{1}{|\Theta_i|} \left\{ T - \sum_{t =1}^T \pi_i^t(\theta^0_i; \alpha_0) \right\}.
\end{equation*}
Rearranging this inequality yields 
$\sum_{t=1}^T \pi_i^t(\theta^0_i;\alpha_0) \ge T - |\Theta_i| \RUSi^T$.
By applying the same argument to $\theta^1_i \in \Theta'_i$ such that $\zeta(\theta^1_i)(b) = 1$ for all $b \in [B]$, we obtain 
$\sum_{t=1}^T \pi_i^t(\theta^1_i;\alpha_1) \ge T - |\Theta_i| \RUSi^T$.
\end{proof}

\begin{figure}[t]
\centering
\begin{tikzpicture}
\newlength{\cw}
\setlength{\cw}{0.025\paperwidth}
\coordinate (o) at (0,0);
\node(t1) at (o) [align=right, anchor=north east]{$\nu_\zeta(\theta_i,b_1)$};
\node(t2) at (t1.south east) [align=right, anchor=north east]{$\nu_\zeta(\theta_i,b_2)$};
\node(t3) at (t2.south east) [align=right, anchor=north east]{$\theta_i = \nu_\zeta(\theta_i,b_3) = \nu_\zeta(\theta_i,b'_4)$};
\node(t4) at (t3.south east) [align=right, anchor=north east]{$\nu_\zeta(\theta_i,b'_3)$};
\node(t5) at (t4.south east) [align=right, anchor=north east]{$\nu_\zeta(\theta_i,b'_2)$};
\node(t6) at (t5.south east) [align=right, anchor=north east]{$\nu_\zeta(\theta_i,b'_1)$};
\node (r11) at (o) [anchor=north west, minimum width=\cw]{1};
\node (r21) at (t2.east) [anchor=west, minimum width=\cw]{1};
\node (r31) at (t3.east) [anchor=west, minimum width=\cw]{1};
\node (r41) at (t4.east) [anchor=west, minimum width=\cw]{1};
\node (r51) at (t5.east) [anchor=west, minimum width=\cw]{1};
\node (r61) at (t6.east) [anchor=west, minimum width=\cw]{0};
\node (r12) at (r11.east) [anchor=west, minimum width=\cw]{1};
\node (r22) at (r21.east) [anchor=west, minimum width=\cw]{0};
\node (r32) at (r31.east) [anchor=west, minimum width=\cw]{0};
\node (r42) at (r41.east) [anchor=west, minimum width=\cw]{0};
\node (r52) at (r51.east) [anchor=west, minimum width=\cw]{0};
\node (r62) at (r61.east) [anchor=west, minimum width=\cw]{0};
\node (r13) at (r12.east) [anchor=west, minimum width=\cw]{1};
\node (r23) at (r22.east) [anchor=west, minimum width=\cw]{1};
\node (r33) at (r32.east) [anchor=west, minimum width=\cw]{1};
\node (r43) at (r42.east) [anchor=west, minimum width=\cw]{1};
\node (r53) at (r52.east) [anchor=west, minimum width=\cw]{0};
\node (r63) at (r62.east) [anchor=west, minimum width=\cw]{0};
\node (r14) at (r13.east) [anchor=west, minimum width=\cw]{1};
\node (r24) at (r23.east) [anchor=west, minimum width=\cw]{1};
\node (r34) at (r33.east) [anchor=west, minimum width=\cw]{1};
\node (r44) at (r43.east) [anchor=west, minimum width=\cw]{0};
\node (r54) at (r53.east) [anchor=west, minimum width=\cw]{0};
\node (r64) at (r63.east) [anchor=west, minimum width=\cw]{0};
\node (r15) at (r14.east) [anchor=west, minimum width=\cw]{1};
\node (r25) at (r24.east) [anchor=west, minimum width=\cw]{1};
\node (r35) at (r34.east) [anchor=west, minimum width=\cw]{0};
\node (r45) at (r44.east) [anchor=west, minimum width=\cw]{0};
\node (r55) at (r54.east) [anchor=west, minimum width=\cw]{0};
\node (r65) at (r64.east) [anchor=west, minimum width=\cw]{0};
\node (r71) at ([yshift=-0.01\paperwidth]r61.south) [anchor=north, minimum width=\cw]{$b'_1$};
\node (r72) at ([yshift=-0.01\paperwidth]r62.south) [anchor=north, minimum width=\cw]{$b_1$};
\node (r73) at ([yshift=-0.01\paperwidth]r63.south) [anchor=north, minimum width=\cw]{$b'_2$};
\node (r74) at ([yshift=-0.01\paperwidth]r64.south) [anchor=north, minimum width=\cw]{$b'_3$};
\node (r75) at ([yshift=-0.01\paperwidth]r65.south) [anchor=north, minimum width=\cw]{$b_2$};
\node (rho) at (o) [anchor=south west] {$\zeta(\cdot)$};
\draw ([xshift=-0.1\paperwidth]o) -- (r15.north east);
\draw (rho.north west) -- ([yshift=-0.03\paperwidth]t6.south east);
\draw [dotted] (r11.north east) -- ([yshift=-0.03\paperwidth]r61.south east);
\draw [dotted] (r12.north east) -- ([yshift=-0.03\paperwidth]r62.south east);
\draw [dotted] (r13.north east) -- ([yshift=-0.03\paperwidth]r63.south east);
\draw [dotted] (r14.north east) -- ([yshift=-0.03\paperwidth]r64.south east);
\end{tikzpicture}
\caption{%
An example of the definition of $\nu_\zeta$.
Given $\theta_i \in \Theta'_i$, we denote the block indices corresponding to $0$ and $1$ by $b_1,b_2,\dots,b_{B_0}$ and $b'_1,b'_2,\dots,b'_{B_1}$, respectively.
We define the series of types $\nu_\zeta(\theta_i,b_1),\nu_\zeta(\theta_i,b_2),\dots,\nu_\zeta(\theta_i,b_{B_0}),\nu_\zeta(\theta_i,b_{B_0+1})=\theta_i$ such that their corresponding binary sequences gradually change from $11\cdots11$ to $\zeta(\theta_i)$.
Similarly, we define the series of types $\nu_\zeta(\theta_i,b'_1),\nu_\zeta(\theta_i,b'_2),\dots,\nu_\zeta(\theta_i,b'_{B_1}),\nu_\zeta(\theta_i,b'_{B_1+1})=\theta_i$ such that their corresponding binary sequences gradually change from $00\cdots00$ to $\zeta(\theta_i)$.
}\label{fig:type-sequence}
\end{figure}

The previous lemma shows that for the type with $\zeta(\theta_i) = 00\cdots00$ or $\zeta(\theta_i) = 11\cdots11$, the algorithm must choose the optimal action for most of the rounds.
For the other types, we cannot prove a similar claim since there exists a counterexample (see \Cref{sec:lb}).
However, we can show that the algorithm must make sufficiently different decisions for some pair of ``adjacent'' types.

Fix any $\theta_i \in \Theta'_i$.
To define the adjacency for $\theta_i$, we define a sequence of types as follows (see \Cref{fig:type-sequence} for example).
Let $B_0$ and $B_1$ be the numbers of blocks $b \in [B]$ such that $\zeta(\theta_i)(b) = 0$ and $\zeta(\theta_i)(b) = 1$, respectively.
We define the sorted indices $b_1 \le b_2 \le \dots \le b_{B_0}$ of blocks such that $\zeta(\theta_i)(b_k) = 0$ for each $k \in [B_0]$.
Similarly, we define the sorted indices $b'_1 \le b'_2 \le \dots \le b'_{B_1}$ of blocks such that $\zeta(\theta_i)(b'_{l}) = 1$ for each $l \in [B_1]$.
For notational convenience, we define $b_0 = b'_0 = 0$ and $b_{B_0+1} = b'_{B_1+1} = B+1$.

For each $k =1,2,\dots,B_0+1$, we define the type $\nu_\zeta(\theta_i,b_k) \in \Theta'_i$ such that $\zeta(\nu_\zeta(\theta_i,b_k))$ is the binary sequence whose first $b_{k-1}$ bits are identical to the prefix of $\zeta(\theta_i)$ and the remaining bits are all $1$, i.e., $\zeta(\nu_\zeta(\theta_i,b_k))(b) = \zeta(\theta_i)(b)$ for each $b \le b_{k-1}$ and $\zeta(\nu_\zeta(\theta_i,b_k))(b) = 1$ for each $b > b_{k-1}$.
Similarly, for each $l =1,2,\dots,B_1+1$, we define the type $\nu_\zeta(\theta_i,b'_l) \in \Theta'_i$ such that $\zeta(\nu_\zeta(\theta_{i},b'_l))$ is the binary sequence whose first $b'_{l-1}$ bits are identical to the prefix of $\zeta(\theta_i)$ and the remaining bits are all $0$.

In the following lemma, we show that the algorithm must make significantly different decisions for $\theta_i$ and $\nu_\zeta(\theta_i,b)$ in each block $b \in [B]$ if $\RUSi^T$ is small.

\begin{lemma}\label{lem:gap-type}
Fix $\zeta$ and $\xi$.
Let $\theta_i \in \Theta'_i$ be an arbitrary type.
Then it holds that
\begin{align*}
\sum_{k=1}^{B_0} \sum_{t \in \calT_{b_{k}}} \left\{ \pi_i^t(\theta_i;\alpha_0) - \pi_i^t(\nu_\zeta(\theta_i,b_{k});\alpha_0) \right\}
+
\sum_{l=1}^{B_1} \sum_{t \in \calT_{b'_{l}}} \left\{ \pi_i^t(\theta_i;\alpha_1) - \pi_i^t(\nu_\zeta(\theta_i,b'_{l});\alpha_1) \right\}
\ge \frac{T}{2} - 2 |\Theta_i| \RUSi^T.
\end{align*}
\end{lemma}

\begin{proof}
We consider $\psi$ in the definition of untruthful swap regret such that $\psi(\nu(\theta_i,b_{k+1})) = \nu(\theta_i,b_k)$ for each $k \in [B_0]$.
Let $\psi(\theta_i) = \theta_i$ for all the other types $\theta_i \in \Theta_i$.
We use the identity map for actions, i.e., $\phi(\theta_i,\alpha) = \alpha$ for all types $\theta_i \in \Theta_i$ and actions $\alpha \in A_i$.
Then we obtain a lower bound on the untruthful swap regret as
\begin{align}
\RUSi^T
&\ge \frac{1}{|\Theta_i|} \sum_{k=1}^{B_0} \sum_{t=1}^T \left\{ \E_{\alpha \sim \pi_i^t(\nu_\zeta(\theta_i,b_k))}\left[ u_i^t(\nu_\zeta(\theta_i,b_{k+1}),\alpha) \right] - \E_{\alpha \sim \pi_i^t(\nu_\zeta(\theta_i,b_{k+1}))} \left[ u_i^t(\nu_\zeta(\theta_i,b_{k+1}), \alpha) \right] \right\}. \label{eq:regret-sequence}
\end{align}

From the definition of $\nu_\zeta(\theta_i,b_{k+1})$, the binary sequence $\zeta(\nu_\zeta(\theta_i,b_{k+1}))$ is identical to $\zeta(\theta_i)$ for the first $b_k$ blocks and all $1$ for the remaining $B - b_k$ blocks.
This implies that $u_i^t(\nu_\zeta(\theta_i,b_{k+1}),\alpha_0) = 1$ for all blocks with indices $b_1,b_2,\dots,b_k$, and $u_i^t(\nu_\zeta(\theta_i,b_{k+1}),\alpha_1) = 1$ for all blocks with indices $b_{k+1},b_{k+2}\dots,b_{B_0}$ and $b'_1,b'_2,\dots,b'_{B_1}$.
Each term on the right-hand side of \eqref{eq:regret-sequence} for $k \in [B_0]$ can be decomposed as
\begin{align*}
&\sum_{t=1}^T \left\{ \E_{\alpha \sim \pi_i^t(\nu_\zeta(\theta_i,b_k))} \left[ u_i^t(\nu_\zeta(\theta_i,b_{k+1}),\alpha) \right] - \E_{\alpha \sim \pi_i^t(\nu_\zeta(\theta_i,b_{k+1}))} \left[ u_i^t(\nu_\zeta(\theta_i,b_{k+1}),\alpha) \right] \right\}\\
&= \underbrace{\sum_{k'=1}^{k} \sum_{t \in \calT_{b_{k'}}} \left\{ \pi_i^t(\nu_\zeta(\theta_i,b_k);\alpha_0)-\pi_i^t(\nu_\zeta(\theta_i,b_{k+1});\alpha_0) \right\}}_{(A)}\\
&\quad + \underbrace{\sum_{k'=k+1}^{B_0} \sum_{t \in \calT_{b_{k'}}} \left\{ \pi_i^t(\nu_\zeta(\theta_i,b_k);\alpha_1)-\pi_i^t(\nu_\zeta(\theta_i,b_{k+1});\alpha_1) \right\}}_{(B)}\\
&\quad + \underbrace{\sum_{l=1}^{B_1} \sum_{t \in \calT_{b'_{l}}} \left\{ \pi_i^t(\nu_\zeta(\theta_i,b_k);\alpha_1)-\pi_i^t(\nu_\zeta(\theta_i,b_{k+1});\alpha_1) \right\}}_{(C)}.
\end{align*}
We take the summation of these terms over all $k = 1,\dots,B_0$.
The sum of (A) for all $k = 1,\dots,B_0$ is
\begin{align*}
&\sum_{k=1}^{B_0} \sum_{k'=1}^{k} \sum_{t \in \calT_{b_{k'}}} \left\{ \pi_i^t(\nu_\zeta(\theta_i,b_k);\alpha_0)-\pi_i^t(\nu_\zeta(\theta_i,b_{k+1});\alpha_0) \right\}\\
&=
\sum_{k'=1}^{B_0} \sum_{k=k'}^{B_0} \sum_{t \in \calT_{b_{k'}}} \left\{ \pi_i^t(\nu_\zeta(\theta_i,b_k);\alpha_0)-\pi_i^t(\nu_\zeta(\theta_i,b_{k+1});\alpha_0) \right\}\\
&=
\sum_{k'=1}^{B_0} \sum_{t \in \calT_{b_{k'}}} \left\{ \pi_i^t(\nu_\zeta(\theta_i,b_{k'});\alpha_0)-\pi_i^t(\nu_\zeta(\theta_i,b_{B_0+1});\alpha_0) \right\},
\end{align*}
where the first equality is obtained by changing the order of summations and the second equality is obtained by the telescoping sum.
Similarly, the sum of (B) for all $k = 1,\dots,B_0$ is
\begin{align*}
&\sum_{k=1}^{B_0} \sum_{k'=k+1}^{B_0} \sum_{t \in \calT_{b_{k'}}} \left\{ \pi_i^t(\nu_\zeta(\theta_i,b_k);\alpha_1)-\pi_i^t(\nu_\zeta(\theta_i,b_{k+1});\alpha_1) \right\}\\
&=
\sum_{k'=1}^{B_0} \sum_{k=1}^{k'-1} \sum_{t \in \calT_{b_{k'}}} \left\{ \pi_i^t(\nu_\zeta(\theta_i,b_k);\alpha_1)-\pi_i^t(\nu_\zeta(\theta_i,b_{k+1});\alpha_1) \right\}\\
&=
\sum_{k'=1}^{B_0} \sum_{t \in \calT_{b_{k'}}} \left\{ \pi_i^t(\nu_\zeta(\theta_i,b_1);\alpha_1)-\pi_i^t(\nu_\zeta(\theta_i,b_{k'});\alpha_1) \right\}.
\end{align*}
By using the telescoping sum, the sum of (C) for all $k = 1,\dots,B_0$ is 
\begin{align*}
&\sum_{k=1}^{B_0} \sum_{l=1}^{B_1} \sum_{t \in \calT_{b'_{l}}} \left\{ \pi_i^t(\nu_\zeta(\theta_i,b_k);\alpha_1)-\pi_i^t(\nu_\zeta(\theta_i,b_{k+1});\alpha_1) \right\}\\
&=
\sum_{l=1}^{B_1} \sum_{t \in \calT_{b'_{l}}} \left\{ \pi_i^t(\nu_\zeta(\theta_i,b_1);\alpha_1) - \pi_i^t(\nu_\zeta(\theta_i,b_{B_0+1});\alpha_1) \right\}.
\end{align*}
By substituting these terms into \eqref{eq:regret-sequence}, we obtain
\begin{align*}
|\Theta_i| \RUSi^T \ge \sum_{k=1}^{B_0} \sum_{t \in \calT_{b_k}} \left\{ \pi_i^t(\nu_\zeta(\theta_i,b_k);\alpha_0) - \pi_i^t(\theta_i;\alpha_0) + \pi_i^t(\nu_\zeta(\theta_i,b_1);\alpha_1) - \pi_i^t(\nu_\zeta(\theta_i,b_k);\alpha_1) \right\}\\
+ \sum_{l=1}^{B_1} \sum_{t \in \calT_{b'_l}} \left\{ \pi_i^t(\nu_\zeta(\theta_i,b_1);\alpha_1) - \pi_i^t(\theta_i;\alpha_1) \right\},
\end{align*}
where we used $\nu_\zeta(\theta_i,b_{B_0+1}) = \theta_i$.
Recall $\zeta(\nu_\zeta(\theta_i,b_1))(b) = 1$ for all $b \in [B]$.
From \Cref{lem:edge-type}, we have
\begin{equation*}
\sum_{k=1}^{B_0} \sum_{t \in \calT_{b_k}} \pi_i^t(\nu_\zeta(\theta_i,b_1);\alpha_1)
+ \sum_{l=1}^{B_1} \sum_{t \in \calT_{b'_l}} \pi_i^t(\nu_\zeta(\theta_i,b_1);\alpha_1)
\ge T - |\Theta_i| \RUSi^T.
\end{equation*}
By summing these two inequalities and rearranging, we obtain
\begin{align*}
\sum_{k=1}^{B_0} \sum_{t \in \calT_{b_k}} \left\{ \pi_i^t(\theta_i;\alpha_0) - \pi_i^t(\nu_\zeta(\theta_i,b_k);\alpha_0) + \pi_i^t(\nu_\zeta(\theta_i,b_k);\alpha_1) \right\}
+ \sum_{l=1}^{B_1} \sum_{t \in \calT_{b'_l}} \pi_i^t(\theta_i;\alpha_1)
\ge T - 2 |\Theta_i| \RUSi^T,
\end{align*}
where we removed $\sum_{k=1}^{B_0} \sum_{t \in \calT_{b_k}} \pi_i^t(\nu_\zeta(\theta_i,b_1);\alpha_1)$ and $\sum_{l=1}^{B_1} \sum_{t \in \calT_{b'_l}} \pi_i^t(\nu_\zeta(\theta_i,b_1);\alpha_1)$ from both sides.
Finally, by using $\pi_i^t(\nu_\zeta(\theta_i,b_k);\alpha_0) + \pi_i^t(\nu_\zeta(\theta_i,b_k);\alpha_1) = 1$ for each $t \in [T]$, we obtain
\begin{align*}
\sum_{k=1}^{B_0} \sum_{t \in \calT_{b_k}} \left\{ \pi_i^t(\theta_i;\alpha_0) - 2 \pi_i^t(\nu_\zeta(\theta_i,b_k);\alpha_0) + 1\right\}
+ \sum_{l=1}^{B_1} \sum_{t \in \calT_{b'_l}} \pi_i^t(\theta_i;\alpha_1)
&\ge T - 2 |\Theta_i| \RUSi^T.
\end{align*}
By using the symmetric analysis, we obtain
\begin{align*}
\sum_{k=1}^{B_0} \sum_{t \in \calT_{b_k}} \pi_i^t(\theta_i;\alpha_0)
+ \sum_{l=1}^{B_1} \sum_{t \in \calT_{b'_l}} \left\{ \pi_i^t(\theta_i;\alpha_1) - 2 \pi_i^t(\nu_\zeta(\theta_i,b'_l);\alpha_1) + 1\right\}
\ge T - 2 |\Theta_i| \RUSi^T.
\end{align*}
By summing these two inequalities, we obtain
\begin{align*}
\sum_{k=1}^{B_0} \sum_{t \in \calT_{b_k}} \left\{ 2 \pi_i^t(\theta_i;\alpha_0) - 2 \pi_i^t(\nu_\zeta(\theta_i,b_k);\alpha_0) + 1\right\}
+ \sum_{l=1}^{B_1} \sum_{t \in \calT_{b'_l}} \left\{ 2 \pi_i^t(\theta_i;\alpha_1) - 2 \pi_i^t(\nu_\zeta(\theta_i,b'_l);\alpha_1) + 1\right\}\\
\ge 2 T - 4 |\Theta_i| \RUSi^T.
\end{align*}
Since $\sum_{k=1}^{B_0} \sum_{t \in \calT_{b_k}} 1 + \sum_{l=1}^{B_1} \sum_{t \in \calT_{b'_l}} 1 = \sum_{t \in [T]} 1 = T$, by dividing both sides by $2$, we complete the proof.
\end{proof}

The previous lemma claims that compared to the decisions for the adjacent type $\nu_\zeta(\theta_i,b)$, the algorithm must choose the optimal action for type $\theta_i$ more frequently in each block $b \in [B]$.
The next lemma compares the algorithm's decisions in block $b \in [B]$ for types whose binary sequences first branch in block $b$, not only adjacent types.

For each $b \in [B]$ and a prefix $p \in \{0,1\}^{b-1}$, 
let $\calI_{p,1}$ and $\calI_{p,1}$ be the sets of types whose corresponding binary sequences have the prefix $p$ and the next bit is $0$ or $1$, respectively.
Formally,
we define $\calI_{p,0} = \{ \theta_i \in \Theta'_i \mid \zeta(\theta_i)(b') = p(b') ~ (\forall b' < b), ~ \zeta(\theta_i)(b) = 0 \}$
and $\calI_{p,1} = \{ \theta_i \in \Theta'_i \mid \zeta(\theta_i)(b') = p(b') ~ (\forall b' < b), ~ \zeta(\theta_i)(b) = 1 \}$.
Note that $\calI_{p,0}$ and $\calI_{p,1}$ are random variables determined by $\zeta_b$.
Let $\calI_p = \calI_{p,0} \cup \calI_{p,1}$ be their union, which is determined by $\zeta_{b-1}$.

In block $b$, the optimal action is $\alpha_0$ for types in $\calI_{p,0}$ and $\alpha_1$ for types in $\calI_{p,1}$.
In the following lemma, we show that the average decisions for $\calI_{p,0}$ and $\calI_{p,1}$ are sufficiently different with some constant probability.

\begin{lemma}\label{lem:gap-block}
Assume
\begin{equation*}
\E_{\zeta,\xi} \left[ \RUSi^T \right] < 2^{-28} \sqrt{T \log_2 |\Theta_i|}.
\end{equation*}
Then it holds that
\begin{align*}
\Pr_{\substack{ b \sim [B], \zeta_b,\xi_{b-1}, \\ p \sim \{0,1\}^{b-1} }} \left(
\E_{\xi_b} \left[ \E_{\theta_i \sim \calI_{p,0}} \left[ \sum_{t \in \calT_b} \pi_i^t(\theta_i;\alpha_0) \right] - \E_{\theta_i \sim \calI_{p,1}} \left[ \sum_{t \in \calT_b} \pi_i^t(\theta_i;\alpha_0) \right] \;\middle|\; \zeta_b,\xi_{b-1} \right]
\ge \frac{L}{8}
\right) \ge \frac{1}{8}.
\end{align*}
\end{lemma}

\begin{proof}
Since \Cref{lem:gap-type} holds for any $\theta_i \in \Theta'_i$, it holds in expectation for any distribution of $\theta_i$.
In this proof, assuming that $\theta_i$ follows the uniform distribution over $\Theta'_i$, we take the expectation of \Cref{lem:gap-type} with respect to $\zeta$, $\xi$, and $\theta_i$.

The expected value of the right-hand side of \Cref{lem:gap-type} is bounded as
\begin{align}
\E_{\zeta,\xi,\theta_i} \left[ \frac{T}{2} - 2 |\Theta_i| \RUSi^T \right]
=
\frac{T}{2} - 2 |\Theta_i| \E_{\zeta,\xi} \left[ \RUSi^T \right]
\ge
\frac{T}{4}, \label{eq:gap-type-rhs}
\end{align}
where the inequality is due to the assumptions $\E_{\zeta,\xi} \left[ \RUSi^T \right] < 2^{-28} \sqrt{T \log_2 |\Theta_i|}$ and $T \ge 2^{-49} |\Theta_i|^2 \log_2 |\Theta_i|$.

The expected value of the left-hand side of \Cref{lem:gap-type} can be decomposed for each block as
\begin{align}
&\E_{\zeta,\xi,\theta_i} \left[ \sum_{k=1}^{B_0} \sum_{t \in \calT_{b_{k}}} \left\{ \pi_i^t(\theta_i;\alpha_0) - \pi_i^t(\nu_\zeta(\theta_i,b_k);\alpha_0) \right\}
+
\sum_{l=1}^{B_1} \sum_{t \in \calT_{b'_{l}}} \left\{ \pi_i^t(\theta_i;\alpha_1) - \pi_i^t(\nu_\zeta(\theta_i,b'_l);\alpha_1) \right\} \right] \nonumber \\
&=
\sum_{b=1}^{B} 
\left\{ 
\Pr_{\zeta,\xi,\theta_i} (\zeta(\theta_i)(b)=0)
\E_{\zeta,\xi,\theta_i} \left[ \sum_{t \in \calT_b} \left\{ \pi_i^t(\theta_i;\alpha_0) - \pi_i^t(\nu_\zeta(\theta_i,b);\alpha_0) \right\} \;\middle|\; \zeta(\theta_i)(b)=0 \right] \right. \nonumber \\
&\qquad +
\left. \Pr_{\zeta,\xi,\theta_i} (\zeta(\theta_i)(b)=1)
\E_{\zeta,\xi,\theta_i} \left[ \sum_{t \in \calT_b} \left\{ \pi_i^t(\nu_\zeta(\theta_i,b);\alpha_0) - \pi_i^t(\theta_i;\alpha_0) \right\} \;\middle|\; \zeta(\theta_i)(b)=1 \right]
\right\}, \label{eq:gap-type-lhs}
\end{align}
where we used $\pi_i^t(\theta_i;\alpha_0) + \pi_i^t(\theta_i;\alpha_1) = 1$ and $\pi_i^t(\nu_\zeta(\theta_i,b);\alpha_0) + \pi_i^t(\nu_\zeta(\theta_i,b);\alpha_1) = 1$ for each $t \in [T]$.
We take the expectation separately for each $b \in [B]$.

First, we take the expectation conditioned on observations from the first block to the $b$th block.
Let $\zeta_b$ and $\xi_b$ be the part of $\zeta$ and $\xi$ that is revealed until the end of block $b \in [B]$.
Formally, we define $\zeta_b = (\zeta(\theta_i)(b'))_{\theta_i \in \Theta'_i, b' \in [b]}$ and $\xi_b = (\xi(\theta_i)(t))_{\theta_i \in \Theta''_i, t \in [bL]}$.
If we condition on $\zeta_b$ and $\xi_b$, the algorithm's outputs $(\pi_i^t)_{t \in [bL]}$ until the end of block $b$ are deterministic.
Therefore, conditioned on $\zeta_b$, $\xi_b$, and $\theta_i$, only $\nu_\zeta(\theta_i,b)$ is stochastic in the term of which we take the expectation in \eqref{eq:gap-type-lhs}.
Recall that $\nu_\zeta(\theta_i,b)$ is defined to be the type with the corresponding binary sequence equal to $\zeta(\theta_i)$ before block $b-1$ and all $1$ (or $0$) after block $b$ if $\zeta(\theta_i)(b)$ is $0$ (or $1$, respectively).
Let $p \in \{0,1\}^{b-1}$ be a random variable that represents the prefix of $\zeta(\theta_i)$.
Since $\zeta$ revealed after block $b$ is uniformly distributed, if the $b$th bit $\zeta(\theta_i)(b)$ is $0$, then $\nu_\zeta(\theta_i,b)$ is distributed uniformly over $\calI_{p,1}$.
We thus obtain
\begin{align}
\E_{\zeta,\xi} \left[ \sum_{t \in \calT_b} \left\{ \pi_i^t(\theta_i;\alpha_0) - \pi_i^t(\nu_\zeta(\theta_i,b);\alpha_0) \right\} \;\middle|\; \zeta_b,\xi_b \right]
=
\sum_{t \in \calT_b} \left\{ \pi_i^t(\theta_i;\alpha_0) - \E_{\theta'_i \sim \calI_{p,1}} \left[ \pi_i^t(\theta'_i;\alpha_0) \right] \right\}. \label{eq:gap-type-lhs1}
\end{align}
Next, we take the expectation over $\theta_i$ conditioned on $\zeta_b$ and $\xi_b$.
Since $\theta_i$ is uniformly distributed over $\Theta'_i$, the prefix $p$ follows the uniform distribution over $\{0,1\}^{b-1}$.
Then we obtain
\begin{align}
&\E_{\theta_i} \left[ \sum_{t \in \calT_b} \left\{ \pi_i^t(\theta_i;\alpha_0) - \E_{\theta'_i \sim \calI_{p,1}} \left[ \pi_i^t(\theta'_i;\alpha_0) \right] \right\} \;\middle|\; \zeta(\theta_i)(b) = 0, \zeta_b,\xi_b \right] \nonumber \\
&= \E_{p \sim \{0,1\}^{b-1}} \left[ \E_{\theta_i} \left[ \sum_{t \in \calT_b} \left\{ \pi_i^t(\theta_i;\alpha_0) - \E_{\theta'_i \sim \calI_{p,1}} \left[ \pi_i^t(\theta'_i;\alpha_0) \right] \right\} \;\middle|\; \zeta(\theta_i)(b) = 0, \text{$p$ is a prefix of $\zeta(\theta_i)$}, \zeta_b,\xi_b \right] \right] \nonumber \\
&= \E_{p \sim \{0,1\}^{b-1}} \left[ \sum_{t \in \calT_b} \left\{ \E_{\theta_i \sim \calI_{p,0}} \left[ \pi_i^t(\theta_i;\alpha_0) \right] - \E_{\theta_i \sim \calI_{p,1}} \left[ \pi_i^t(\theta_i;\alpha_0) \right] \right\} \right]. \label{eq:gap-type-lhs2}
\end{align}
Finally, we take the expectation over $\zeta_b$ and $\xi_b$, which determine $\pi_i^t$ for each $t \in [bL]$ and the partition $(\calI_{p,0}, \calI_{p,1})$ for each $p \in \{0,1\}^{b-1}$.
By combining it with \eqref{eq:gap-type-lhs1} and \eqref{eq:gap-type-lhs2}, we obtain
\begin{align*}
&\E_{\zeta,\xi,\theta_i} \left[ \sum_{t \in \calT_b} \left\{ \pi_i^t(\theta_i;\alpha_0) - \pi_i^t(\nu_\zeta(\theta_i,b);\alpha_0) \right\} \;\middle|\; \zeta(\theta_i)(b) = 0 \right]\\
&= \E_{\zeta,\xi,p \sim \{0,1\}^{b-1}} \left[ \sum_{t \in \calT_b} \left\{ \E_{\theta_i \sim \calI_{p,0}} \left[ \pi_i^t(\theta_i;\alpha_0) \right] - \E_{\theta_i \sim \calI_{p,1}} \left[ \pi_i^t(\theta_i;\alpha_0) \right] \right\} \right].
\end{align*}
By applying the same argument to the case of $\zeta(\theta_i)(b) = 1$, we obtain
\begin{align*}
&\E_{\zeta,\xi,\theta_i} \left[ \sum_{t \in \calT_b} \left\{ \pi_i^t(\nu_\zeta(\theta_i,b);\alpha_0) - \pi_i^t(\theta_i;\alpha_0) \right\} \;\middle|\; \zeta(\theta_i)(b) = 1 \right]\\
&= \E_{\zeta,\xi,p \sim \{0,1\}^{b-1}} \left[ \sum_{t \in \calT_b} \left\{ \E_{\theta_i \sim \calI_{p,0}} \left[ \pi_i^t(\theta_i;\alpha_0) \right] - \E_{\theta_i \sim \calI_{p,1}} \left[ \pi_i^t(\theta_i;\alpha_0) \right] \right\} \right].
\end{align*}
By plugging these expected values into \eqref{eq:gap-type-lhs} and bounding it by the right-hand side \eqref{eq:gap-type-rhs}, we obtain
\begin{align*}
&\sum_{b \in [B]} \E_{\zeta,\xi,p \sim \{0,1\}^{b-1}} \left[
\sum_{t \in \calT_{b}} \left\{ \E_{\theta_i \sim \calI_{p,0}} \left[ \pi_i^t(\theta_i;\alpha_0) \right] - \E_{\theta_i \sim \calI_{p,1}} \left[ \pi_i^t(\theta_i;\alpha_0) \right]
\right\} \right]
\ge \frac{T}{4}.
\end{align*}
By dividing both sides by $B$ and taking the conditional expectation, we obtain
\begin{align*}
&\E_{b \sim [B],\zeta_b,\xi_{b-1},p \sim \{0,1\}^{b-1}} \left[
\E_{\zeta,\xi} \left[
\sum_{t \in \calT_{b}} \left\{ \E_{\theta_i \sim \calI_{p,0}} \left[ \pi_i^t(\theta_i;\alpha_0) \right] - \E_{\theta_i \sim \calI_{p,1}} \left[ \pi_i^t(\theta_i;\alpha_0) \right]
\right\} \;\middle|\; \zeta_b,\xi_{b-1} \right] \right]
\ge \frac{L}{4}.
\end{align*}
Since $\calI_{p,0}$, $\calI_{p,1}$, and $(\pi_i^t)_{t \in \calT_b}$ are determined only by $\zeta_b$ and $\xi_b$, we can replace the expectation with respect to $\zeta$ and $\xi$ by the expectation with respect to $\xi_b$.
By applying \Cref{lem:gap-lower-bound} proved below, we complete the proof.
\end{proof}

In this proof, we used the following lemma.
\begin{lemma}\label{lem:gap-lower-bound}
Fix $x, y \in \bbR$ such that $0 < y < x$.
Let $X$ be a random variable whose value is at most $x$.
If $\E[X] \ge y$, then $\Pr(X \ge y/2) \ge \frac{y}{2x}$.
\end{lemma}

\begin{proof}
By Markov's inequality to $x - X$, for any $\epsilon > 0$, we obtain
\begin{equation*}
\Pr\left(X \le \frac{y}{2} + \epsilon \right)
= \Pr\left(x - X \ge x - \frac{y}{2} - \epsilon \right)
\le \frac{\E[x-X]}{x-\frac{y}{2}-\epsilon}
\le \frac{x-y}{x-\frac{y}{2}-\epsilon}
= 1 - \frac{\frac{y}{2}-\epsilon}{x-\frac{y}{2}-\epsilon}
\end{equation*}
By setting sufficiently small $\epsilon \in (0, \frac{y^2}{4x-2y})$, we obtain
$\Pr\left(X \le \frac{y}{2} + \epsilon \right)
\le 1 - \frac{y}{2x}$.
By considering the complementary event, we obtain $\Pr(X \ge y/2) \ge \Pr(X > y/2 + \epsilon) \ge \frac{y}{2x}$.
\end{proof}

\subsection{Assignment of types}\label{sec:lb2}

The goal of this subsection is to define an assignment $\psi'_{\zeta,\xi} \colon \Theta''_i \to \Theta'_i$ such that against the reward sequence for each $\theta''_i \in \Theta''_i$, the decision for $\psi'_{\zeta,\xi}(\theta''_i) \in \Theta'_i$ achieves significantly better than $T/2$.
In the previous subsection, we showed that if the expected untruthful swap regret is small, the algorithm must take decisions different for $\calI_{p,0} \subseteq \Theta'_i$ and $\calI_{p,1} \subseteq \Theta'_i$ with some constant probability (\Cref{lem:gap-block}).
We consider assigning these decisions to the rewards for each $\theta''_i \in \Theta''_i$.
Since these rewards are determined by the independent random variables $\xi$, the total reward obtained by any decision is $L/2$ in expectation for every block.
Intuitively, if the decisions for $\calI_{p,0}$ and $\calI_{p,1}$ are sufficiently different, then the better one fits the rewards for $\theta''_i \in \Theta''_i$, thus achieving the expected reward $L/2 + \Omega(\sqrt{L})$ due to anti-concentration.

A simple but failed approach is to decide $\psi'_{\zeta,\xi}(\theta''_i)$ independently for each $\theta''_i \in \Theta''_i$.
For each $\zeta$ and $\xi$, this approach iteratively selects $\calI_{p,0}$ or $\calI_{p,1}$ whose average decision achieves better for the rewards of $\theta''_i$ as follows.
First, let $p$ be the empty sequence.
For each $b = 1,2,\dots,B$, we assign $\theta''_i$ to $0$ if the average decision for $\calI_{p,0}$ is better than that for $\calI_{p,1}$ against the reward sequence of $\theta''_i$, i.e.,
\begin{equation*}
\E_{\theta'_i \sim \calI_{p,0}} \left[ \sum_{t \in \calT_b} \E_{\alpha \sim \pi_i^t(\theta'_i)} \left[ u_i^t(\theta''_i, \alpha) \right] \right]
\ge
\E_{\theta'_i \sim \calI_{p,1}} \left[ \sum_{t \in \calT_b} \E_{\alpha \sim \pi_i^t(\theta'_i)} \left[ u_i^t(\theta''_i, \alpha) \right] \right]
\end{equation*}
and $1$ otherwise, and we then add the selected $0$ or $1$ to the end of $p$.
Finally, we set $\psi'_{\zeta,\xi}(\theta''_i) = \zeta^{-1}(p)$, where $\zeta^{-1} \colon \{0,1\}^B \to \Theta'_i$ is the inverse map of the bijection $\zeta$.
This assignment process is done for each $\theta''_i \in \Theta''_i$ independently.
This approach has the following two issues.
\begin{itemize}
\item
Since the assignment is independent for each $\theta''_i \in \Theta''_i$, the assigned types $(\psi'_{\zeta,\xi}(\theta''_i))_{\theta''_i \in \Theta''_i}$ might be unbalanced, that is, for some $\zeta_{b-1}$, $\xi_{b-1}$, and $p \in \{0,1\}^{b-1}$, there is no $\theta''_i \in \Theta''_i$ that is assigned to $\calI_{p,0}$ or $\calI_{p,1}$ conditioned on $\zeta_{b-1}$ and $\xi_{b-1}$.
If the average decisions for $\calI_{p,0}$ and $\calI_{p,1}$ are sufficiently different \textit{for all $p \in \{0,1\}^{b-1}$ with any $b \in [B]$}, then this assignment is enough.
However, \Cref{lem:gap-block} guarantees this condition only with a constant fraction of $p$ and $b$.
If the algorithm is manipulatively designed, this condition may fail to hold with a significant probability for each $\zeta_{b-1}$, $\xi_{b-1}$, and $p$ that $\theta''_i \in \Theta''_i$ is assigned to.

\item
The assignment for each block is decided according to the average decisions for $\calI_{p,0}$ and $\calI_{p,1}$.
However, to obtain a lower bound on the untruthful swap regret, we need to prove a lower bound for the decision for $\psi_{\zeta,\xi}(\theta''_i)$, not the average decision for $\calI_{p,0}$ or $\calI_{p,1}$.
The selection of $\psi_{\zeta,\xi}(\theta''_i)$ from $\calI_{p,0}$ or $\calI_{p,1}$ depends on the assignment in the future blocks $b,b+1,\dots,B$.
If the algorithm is manipulatively designed, the distribution of $\psi_{\zeta,\xi}(\theta''_i)$ can be far from the uniform distribution over $\calI_{p,0}$ or $\calI_{p,1}$.
\end{itemize}

To overcome these two obstacles, we use another definition of $\psi'_{\zeta,\xi}$.
The assignment process for each $\zeta$ and $\xi$ proceeds as follows.
Let $\calJ = \Theta''_i$, which is supposed to be $\calJ_p$ with the empty sequence $p$.
For each $b = 1,2,\dots,B$ and each $p \in \{0,1\}^{b-1}$, we partition $\calJ_p$ into two subsets $\calJ_{p,0}$ and $\calJ_{p,1}$ of equal size.
To decide this partition, we add an arbitrary index for each type in $\calJ_p$ as $\calJ_p = \{j_1,j_2,\dots,j_{2^{B-b+1}}\}$.
For each $k \in [2^{B-b}]$, we add one of $\{j_{2k-1}, j_{2k}\}$ into $\calJ_{p,0}$ and the other into $\calJ_{p,1}$.
Conditioned on $\zeta_b$ and $\xi_{b-1}$, we consider the distribution of
\begin{equation}
\label{eq:quantile}
\E_{\theta'_i \sim \calI_{p,0}} \left[ \sum_{t \in \calT_b} \E_{\alpha \sim \pi_i^t(\theta'_i)} \left[ u_i^t(j_{2k-1}, \alpha) - u_i^t(j_{2k}, \alpha) \right] \right]
-
\E_{\theta'_i \sim \calI_{p,1}} \left[ \sum_{t \in \calT_b} \E_{\alpha \sim \pi_i^t(\theta'_i)} \left[ u_i^t(j_{2k-1}, \alpha) - u_i^t(j_{2k}, \alpha) \right] \right],
\end{equation}
whose randomness comes from the randomness of $\xi_b$.
This value compares the average decisions for $\calI_{p,0}$ and $\calI_{p,1}$ in terms of the performance difference for the rewards of $j_{2k-1}$ and $j_{2k}$.
For the upper $50\%$ of this distribution, we set $j_{2k-1} \in \calJ_{p,0}$ and $j_{2k} \in \calJ_{p,1}$.
For the lower $50\%$ of this distribution, we set $j_{2k-1} \in \calJ_{p,1}$ and $j_{2k} \in \calJ_{p,0}$.
Since this is a discrete distribution, there might be arbitrariness in how to divide the mass on the median, but any division is allowed.
In the end of this process, for each $p \in \{0,1\}^B$, the assigned type $\calJ_p \subseteq \Theta''_i$ is a singleton.
Recall that $\calI_p$ is also a singleton for each $p \in \{0,1\}^B$.
For each $p \in \{0,1\}^B$, we define $\psi'_{\zeta,\xi}(\theta''_i) = \theta'_i$, where $\theta'_i \in \calI_p$ and $\theta''_i \in \calJ_p$ are the unique element of $\calI_p$ and $\calJ_p$, respectively.

This assignment addresses the two issues described above as follows.
\begin{itemize}
\item
The first issue is caused when $(\psi'_{\zeta,\xi}(\theta''_i))_{\theta''_i \in \Theta''_i}$ is unbalanced over $\Theta'_i$.
To avoid this issue, we consider all types in $\Theta''_i$ at the same time and partition them into equally sized subsets in each block.
This guarantees that $\psi'_{\zeta,\xi} \colon \Theta''_i \to \Theta'_i$ is always a bijection.
Therefore, for any $b \in [B]$, $p \in \{0,1\}^{b-1}$, $\zeta_{b-1}$, and $\xi_{b-1}$ such that the average decisions are sufficiently different for $\calI_{p,0}$ and $\calI_{p,1}$, there is some $\theta''_i$ assigned to it.

\item
The second issue is addressed by assigning each type in $\calJ_p$ to $\calJ_{p,0}$ and $\calJ_{p,1}$ with equal probability.
Since this partition applies recursively, $\psi'_{\zeta,\xi}(\theta''_i)$ follows the uniform distribution over $\calI_p$.
This is formally stated in the following lemma.
\end{itemize}

\begin{lemma}\label{lem:psi-dist}
Let $b \in [B]$ be an arbitrary block index.
Let $\zeta_b$ and $\xi_b$ be any realization of $\zeta$ and $\xi$ for the first $b$ blocks.
Let $p \in \{0,1\}^b$ be an arbitrary bit sequence of length $b$.
For any pair of $\theta'_i \in \calI_p$ and $\theta''_i \in \calJ_p$, it holds that
\begin{equation*}
\Pr_{\zeta,\xi} \left(\psi'_{\zeta,\xi}(\theta''_i) = \theta'_i \;\middle|\; \zeta_b,\xi_b \right) = \frac{1}{2^{B-b}}.
\end{equation*}
\end{lemma}

\begin{proof}
Fix any $b \in [B]$, $\zeta_b$, $\xi_b$, $p \in \{0,1\}^b$, $\theta'_i \in \calI_p$, and $\theta''_i \in \calJ_p$.
Let $\calF_{b'}$ be the event that $\theta'_i$ and $\psi'_{\zeta,\xi}(\theta''_i)$ are in the same partition in the $b'$th block.
Formally, $\calF_{b'}$ is the event that $\theta'_i \in \calI_{p'}$ and $\theta''_i \in \calJ_{p'}$ hold for some $p' \in \{0,1\}^{b'}$.
By induction, we prove
\begin{equation}\label{eq:induction}
\Pr_{\zeta,\xi} \left(\calF_{b'} \;\middle|\; \zeta_b,\xi_b \right) = \frac{1}{2^{b'-b}}
\end{equation}
for each $b' \in [B]$ such that $b \le b' \le B$.

For the case of $b' = b$, since $\theta'_i \in \calI_p$ and $\theta''_i \in \calJ_p$, we obtain
\begin{equation*}
\Pr_{\zeta,\xi} \left(\calF_b \;\middle|\; \zeta_b,\xi_b \right) = 1,
\end{equation*}
which proves the base case.

Fix any $b' \in [B]$ such that $b \le b' \le B-1$.
Assume \eqref{eq:induction} holds for $b'$.
For any $p' \in \{0,1\}^{b'}$ that has $p$ as a prefix, we have
\begin{align}
&\Pr_{\zeta,\xi} \left( \theta'_i \in \calI_{p',0}, \theta''_i \in \calJ_{p',0} \;\middle|\; \zeta_b,\xi_b \right) \nonumber\\
&= 
\Pr_{\zeta,\xi} \left( \theta'_i \in \calI_{p'}, \theta''_i \in \calJ_{p'} \;\middle|\; \zeta_b,\xi_b \right) \nonumber\\
&\times \Pr_{\zeta,\xi} \left( \theta'_i \in \calI_{p',0} \;\middle|\; \zeta_b,\xi_b,\theta'_i \in \calI_{p'}, \theta''_i \in \calJ_{p'} \right)
\Pr_{\zeta,\xi} \left( \theta''_i \in \calJ_{p',0} \;\middle|\; \zeta_b,\xi_b,\theta'_i \in \calI_{p',0}, \theta''_i \in \calJ_{p'} \right). \label{eq:assignment}
\end{align}
Conditioned on $\zeta_{b'}$ and $\xi_{b'}$, whether $\theta'_i \in \calI_{p',0}$ or $\theta'_i \in \calI_{p',1}$ holds is determined only by the randomness of the $(b'+1)$th block of $\zeta$.
The second factor of \eqref{eq:assignment} is
\begin{align}
&\Pr_{\zeta,\xi} \left( \theta'_i \in \calI_{p',0} \;\middle|\; \zeta_b,\xi_b,\theta'_i \in \calI_{p'}, \theta''_i \in \calJ_{p'} \right) \nonumber \\
&= \E_{\zeta_{b'},\xi_{b'}} \left[
\Pr_{\zeta,\xi} \left( \theta'_i \in \calI_{p',0} \;\middle|\; \zeta_{b'},\xi_{b'},\theta'_i \in \calI_{p'}, \theta''_i \in \calJ_{p'} \right)
\;\middle|\; \zeta_{b},\xi_{b},\theta'_i \in \calI_{p'}, \theta''_i \in \calJ_{p'} \right]
\nonumber \\
&= \E_{\zeta_{b'},\xi_{b'}} \left[
\frac{1}{2}
\;\middle|\; \zeta_{b},\xi_{b},\theta'_i \in \calI_{p'}, \theta''_i \in \calJ_{p'} \right]
\nonumber \\
&= \frac{1}{2}. \label{eq:assignment1}
\end{align}
Conditioned on $\zeta_{b'+1}$ and $\xi_{b'}$, whether $\theta''_i \in \calJ_{p',0}$ or $\theta''_i \in \calJ_{p',1}$ is determined only by the randomness of $\xi_{b'+1}$.
Since we decide it according to whether the realized $\xi_{b'+1}$ is in the upper $50\%$ or lower $50\%$ of the distribution of some random variable, the third factor of \eqref{eq:assignment} is
\begin{align}
&\Pr_{\zeta,\xi} \left( \theta''_i \in \calJ_{p',0} \;\middle|\; \zeta_b,\xi_b,\theta'_i \in \calI_{p',0}, \theta''_i \in \calJ_{p'} \right) \nonumber \\
&= \E_{\zeta_{b'+1},\xi_{b'}} \left[
\Pr_{\zeta,\xi} \left( \theta''_i \in \calJ_{p',0} \;\middle|\; \zeta_{b'+1},\xi_{b'},\theta'_i \in \calI_{p',0}, \theta''_i \in \calJ_{p'} \right)
\;\middle|\; \zeta_{b},\xi_{b},\theta'_i \in \calI_{p',0}, \theta''_i \in \calJ_{p'} \right]
\nonumber \\
&= \E_{\zeta_{b'+1},\xi_{b'}} \left[
\frac{1}{2}
\;\middle|\; \zeta_{b},\xi_{b},\theta'_i \in \calI_{p',0}, \theta''_i \in \calJ_{p'} \right]
\nonumber \\
&= \frac{1}{2}.
\label{eq:assignment2}
\end{align}
By plugging \eqref{eq:assignment1} and \eqref{eq:assignment2} into \eqref{eq:assignment}, we obtain
\begin{align*}
\Pr_{\zeta,\xi} \left( \theta'_i \in \calI_{p',0}, \theta''_i \in \calJ_{p',0} \;\middle|\; \zeta_b,\xi_b \right)
= 
\frac{1}{4} \Pr_{\zeta,\xi} \left( \theta'_i \in \calI_{p'}, \theta''_i \in \calJ_{p'} \;\middle|\; \zeta_b,\xi_b \right).
\end{align*}
In the same way, we can obtain
\begin{equation*}
\Pr_{\zeta,\xi} \left( \theta'_i \in \calI_{p',1}, \theta''_i \in \calJ_{p',1} \;\middle|\; \zeta_b,\xi_b \right)
= \frac{1}{4} \Pr_{\zeta,\xi} \left( \theta'_i \in \calI_{p'}, \theta''_i \in \calJ_{p'} \;\middle|\; \zeta_b,\xi_b \right).
\end{equation*}
Then the probability of $\calF_{b'+1}$ is
\begin{align*}
&\Pr_{\zeta,\xi} \left(\calF_{b'+1} \;\middle|\; \zeta_b,\xi_b \right)\\
&=
\sum_{\substack{p' \in \{0,1\}^{b'+1}}}
\Pr_{\zeta,\xi} \left( \theta'_i \in \calI_{p'}, \theta''_i \in \calJ_{p'} \;\middle|\; \zeta_b,\xi_b \right)\\
&=
\sum_{\substack{p' \in \{0,1\}^{b'}}}
\left\{
\Pr_{\zeta,\xi} \left( \theta'_i \in \calI_{p',0}, \theta''_i \in \calJ_{p',0} \;\middle|\; \zeta_b,\xi_b \right)
+
\Pr_{\zeta,\xi} \left( \theta'_i \in \calI_{p',1}, \theta''_i \in \calJ_{p',1} \;\middle|\; \zeta_b,\xi_b \right)
\right\} \\
&=
\sum_{\substack{p' \in \{0,1\}^{b'}}}
\left\{
\frac{1}{4} \Pr_{\zeta,\xi} \left( \theta'_i \in \calI_{p'}, \theta''_i \in \calJ_{p'} \;\middle|\; \zeta_b,\xi_b \right)
+
\frac{1}{4} \Pr_{\zeta,\xi} \left( \theta'_i \in \calI_{p'}, \theta''_i \in \calJ_{p'} \;\middle|\; \zeta_b,\xi_b \right)
\right\} \\
&=
\frac{1}{2} \sum_{\substack{p' \in \{0,1\}^{b'}}}
\Pr_{\zeta,\xi} \left( \theta'_i \in \calI_{p'}, \theta''_i \in \calJ_{p'} \;\middle|\; \zeta_b,\xi_b \right)\\
&=
\frac{1}{2} \Pr_{\zeta,\xi} \left( \calF_{b'} \;\middle|\; \zeta_b,\xi_b \right)\\
&= \frac{1}{2} \cdot \frac{1}{2^{b'-b}} \tag{from the induction hypothesis}\\
&= \frac{1}{2^{b'-b+1}},
\end{align*}
which proves the induction step.

By induction, \eqref{eq:induction} holds for each $b \le b' \le B$.
By considering the case of $b' = B$, we obtain
\begin{equation*}
\Pr_{\zeta,\xi} \left(\calF_B \;\middle|\; \zeta_b,\xi_b \right) = \frac{1}{2^{B-b}}.
\end{equation*}
From the definition of $\psi'_{\zeta,\xi}$, the event $\calF_B$ implies $\psi'_{\zeta,\xi}(\theta''_i) = \theta'_i$, which completes the proof.
\end{proof}

\subsection{Analysis of the expected reward in each block}\label{sec:lb3}

In \Cref{sec:lb1}, we showed that the average decisions for $\calI_{p,0}$ and $\calI_{p,1}$ are sufficiently different with a constant probability.
In \Cref{sec:lb2}, we defined the assignment $\psi'_{\zeta,\xi} \colon \Theta''_i \to \Theta'_i$ such that $\calJ_{p,0}$ and $\calJ_{p,1}$ correspond to $\calI_{p,0}$ and $\calI_{p,1}$ in each block, respectively.
The next lemma claims that if the average decisions for $\calI_{p,0}$ and $\calI_{p,1}$ are sufficiently different, the average decisions for $\calI_{p,0}$ and $\calI_{p,1}$ achieve $L/2 + \Omega(\sqrt{L})$ for the rewards of $\calJ_{p,0}$ and $\calJ_{p,1}$, respectively.

The proof proceeds in the following steps.
First, by introducing variables $Y_{\theta''_i}^t$ and $Z_{\theta''_i}^t$ for each $\theta''_i \in \Theta''_i$ and $t \in \calT_b$, we express the expected value achieved by the assignment as \eqref{eq:sqrt-gap-lhs1}.
The first term of \eqref{eq:sqrt-gap-lhs1} is equal to $L/2$ as shown in \eqref{eq:sqrt-gap-lhs4}.
The remaining term of \eqref{eq:sqrt-gap-lhs1} can be written as a random walk term $D_k^\tmax$ with a sign $h_k(\xi_b)$ as shown in \eqref{eq:sqrt-gap-lhs2}.
To bound it by a term without $h_k(\xi_b)$, we define a more tractable $h'_k(\xi_b)$ and obtain \eqref{eq:sqrt-gap-lhs3}.
We bound \eqref{eq:sqrt-gap-lhs3} by using Doob's optional stopping theorem.
Combining all these inequalities, we prove the following lemma.

\begin{lemma}\label{lem:martingale}
Let $b \in [B]$ be an arbitrary block index.
Let $\zeta_b$ be any partial realization of $\zeta$ revealed until the end of the $b$th block.
Let $\xi_{b-1}$ be any partial realization of $\xi$ revealed until the end of the $(b-1)$th block.
Let $p \in \{0,1\}^{b-1}$ be an arbitrary bit sequence of length $b-1$.
Let $c \in [0,1]$ be an arbitrary non-negative constant.
If
\begin{equation*}
\E_{\xi_b} \left[ \sum_{t \in \calT_b} \left| \E_{\theta_i \sim \calI_{p,0}} \left[ \pi_i^t(\theta_i; \alpha_0) \right]
-
\E_{\theta_i \sim \calI_{p,1}} \left[ \pi_i^t(\theta_i; \alpha_0) \right] \right| \;\middle|\; \zeta_b,\xi_{b-1} \right]
\ge c L,
\end{equation*}
then
\begin{align}
\E_{\xi_b} \left[
\frac{1}{2} \E_{\substack{\theta'_i \sim \calI_{p,0} \\ \theta''_i \sim \calJ_{p,0} }} \left[ \sum_{t \in \calT_b} \E_{\alpha \sim \pi_i^t(\theta'_i)} \left[ u_i^t(\theta''_i, \alpha) \right] \right]
+
\frac{1}{2} \E_{\substack{ \theta'_i \sim \calI_{p,1} \\ \theta''_i \sim \calJ_{p,1} }} \left[ \sum_{t \in \calT_b} \E_{\alpha \sim \pi_i^t(\theta'_i)} \left[ u_i^t(\theta''_i, \alpha) \right] \right]
\;\middle|\; \zeta_b, \xi_{b-1} \right] \nonumber \\
\ge \frac{L}{2} + 2^{-10} c^{9/2}\sqrt{L}. \label{eq:sqrt-gap}
\end{align}
\end{lemma}

\begin{proof}
Fix $b \in [B]$ and $p \in \{0,1\}^{b-1}$.
We define random variables
\begin{equation*}
X_0^t = \E_{\theta'_i \sim \calI_{p,0}} \left[ \pi_i^t(\theta'_i;\alpha_0) \right]
\qquad \text{and} \qquad
X_1^t = \E_{\theta'_i \sim \calI_{p,1}} \left[ \pi_i^t(\theta'_i;\alpha_0) \right]
\end{equation*}
that represent the average decisions for $\calI_{p,0}$ and $\calI_{p,1}$ for each $t \in \calT_b$.
Note that they are deterministic if $\zeta_b$ and $\xi_b$ are fixed.
If we apply the average decision for $\calI_{p,0}$ to the rewards for each $\theta''_i \in \calJ_p$, the expected total reward in round $t \in \calT_b$ is
\begin{align*}
&\E_{\theta'_i \sim \calI_{p,0}} \left[ \E_{\alpha \sim \pi_i^t(\theta'_i)} \left[ u_i^t(\theta''_i,\alpha) \right] \right]\\
&= \E_{\theta'_i \sim \calI_{p,0}} \left[ \pi_i^t(\theta'_i;\alpha_0) u_i^t(\theta''_i,\alpha_0) + \pi_i^t(\theta'_i;\alpha_1) u_i^t(\theta''_i,\alpha_1) \right]\\
&= \E_{\theta'_i \sim \calI_{p,0}} \left[ \pi_i^t(\theta'_i;\alpha_0) \right]  u_i^t(\theta''_i,\alpha_0)
+ \E_{\theta'_i \sim \calI_{p,0}} \left[ \pi_i^t(\theta'_i;\alpha_1) \right] u_i^t(\theta''_i,\alpha_1)\\
&= X_0^t \xi(\theta''_i)(t) + (1-X_0^t) (1 - \xi(\theta''_i)(t)).
\end{align*}
Similarly, if we apply the average decision for $\calI_{p,1}$, the expected total reward is
\begin{align*}
\E_{\theta'_i \sim \calI_{p,1}} \left[ \E_{\alpha \sim \pi_i^t(\theta'_i)} \left[ u_i^t(\theta''_i,\alpha) \right] \right] 
= X_1^t \xi(\theta''_i)(t) + (1-X_1^t) (1 - \xi(\theta''_i)(t)).
\end{align*}
For each $\theta''_i \in \Theta''_i$ and $t \in \calT_b$, we define random variables
\begin{equation*}
Y_{\theta''_i}^t = \frac{X_0^t+X_1^t}{2} \xi(\theta''_i)(t) + \left( 1-\frac{X_0^t+X_1^t}{2}\right) (1 - \xi(\theta''_i)(t))
\end{equation*}
and
\begin{equation*}
Z_{\theta''_i}^t = \frac{X_0^t-X_1^t}{2} \xi(\theta''_i)(t) - \frac{X_0^t-X_1^t}{2} (1 - \xi(\theta''_i)(t)),
\end{equation*}
which are determined by $\zeta_b$ and $\xi_b$.
The total reward of the average decisions for $\calI_{p,0}$ in block $b$ is
\begin{align*}
\sum_{t \in \calT_b} \E_{\theta'_i \sim \calI_{p,0}} \left[ \E_{\alpha \sim \pi_i^t(\theta'_i)} \left[ u_i^t(\theta''_i,\alpha) \right] \right] 
&= \sum_{t \in \calT_b} \left\{ X^t_0 \xi(\theta''_i)(t) + (1-X^t_0) (1 - \xi(\theta''_i)(t)) \right\}\\
&= \sum_{t \in \calT_b} (Y_{\theta''_i}^t + Z_{\theta''_i}^t).
\end{align*}
Similarly, the total reward of the average decisions for $\calI_{p,1}$ is
\begin{align*}
\sum_{t \in \calT_b} 
\E_{\theta'_i \sim \calI_{p,1}} \left[ \E_{\alpha \sim \pi_i^t(\theta'_i)} \left[ u_i^t(\theta''_i,\alpha) \right] \right]
&=
\sum_{t \in \calT_b} \left\{ X_1^t \xi(\theta''_i)(t) + (1-X_1^t) (1 - \xi(\theta''_i)(t)) \right\}\\
&= \sum_{t \in \calT_b} (Y_{\theta''_i}^t - Z_{\theta''_i}^t).
\end{align*}
The left-hand side of the inequality \eqref{eq:sqrt-gap} is
\begin{align}
&\E_{\xi_b} \left[ \frac{1}{2} \E_{\substack{\theta'_i \sim \calI_{p,0} \\ \theta''_i \sim \calJ_{p,0} }} \left[ \sum_{t \in \calT_b} \E_{\alpha \sim \pi_i^t(\theta'_i)} \left[ u_i^t(\theta''_i, \alpha) \right] \right]
+
\frac{1}{2} \E_{\substack{ \theta'_i \sim \calI_{p,1} \\ \theta''_i \sim \calJ_{p,1} } } \left[ \sum_{t \in \calT_b} \E_{\alpha \sim \pi_i^t(\theta'_i)} \left[ u_i^t(\theta''_i, \alpha) \right] \right] \;\middle|\; \zeta_b,\xi_{b-1} \right] \nonumber \\
&=
\E_{\xi_b} \left[ \frac{1}{2} \E_{\theta''_i \sim \calJ_{p,0}} \left[ \sum_{t \in \calT_b} (Y_{\theta''_i}^t + Z_{\theta''_i}^t) \right]
+
\frac{1}{2} \E_{\theta''_i \sim \calJ_{p,1}} \left[ \sum_{t \in \calT_b} (Y_{\theta''_i}^t - Z_{\theta''_i}^t) \right] \;\middle|\; \zeta_b,\xi_{b-1} \right] \nonumber \\
&=
\E_{\xi_b} \left[ \E_{\theta''_i \sim \calJ_p} \left[ \sum_{t \in \calT_b} Y_{\theta''_i}^t \right]
+
\frac{1}{2} \E_{\theta''_i \sim \calJ_{p,0}} \left[ \sum_{t \in \calT_b} Z_{\theta''_i}^t \right]
-
\frac{1}{2} \E_{\theta''_i\sim \calJ_{p,1}} \left[ \sum_{t \in \calT_b} Z_{\theta''_i}^t \right] \;\middle|\; \zeta_b,\xi_{b-1} \right], \label{eq:sqrt-gap-lhs1}
\end{align}
where the second equality is due to $\calJ_p = \calJ_{p,0} \cup \calJ_{p,1}$.

The first term of \eqref{eq:sqrt-gap-lhs1} is
\begin{align*}
&\E_{\xi_b} \left[ \E_{\theta''_i \sim \calJ_p} \left[ \sum_{t \in \calT_b} Y_{\theta''_i}^t \right] \;\middle|\; \zeta_b,\xi_{b-1} \right] \\
&= \E_{\xi_b} \left[ \E_{\theta''_i \sim \calJ_p} \left[ \sum_{t \in \calT_b} \left\{ \frac{X_0^t+X_1^t}{2} \xi(\theta''_i)(t) + \left( 1-\frac{X_0^t+X_1^t}{2}\right) (1 - \xi(\theta''_i)(t)) \right\} \right] \;\middle|\; \zeta_b,\xi_{b-1} \right]\\
&= \E_{\theta''_i \sim \calJ_p} \left[ \sum_{t \in \calT_b} \E_{\xi_b} \left[ \frac{X_0^t+X_1^t}{2} \xi(\theta''_i)(t) + \left( 1-\frac{X_0^t+X_1^t}{2}\right) (1 - \xi(\theta''_i)(t)) \;\middle|\; \zeta_b,\xi_{b-1} \right] \right].
\end{align*}
Here, we note that $\calJ_p$ is determined by $\zeta_{b-1}$ and $\xi_{b-1}$.
Since $X_0^t$ and $X_1^t$ are determined by the algorithm before $\xi(\theta''_i)(t)$ is revealed, $\frac{X_0^t+X_1^t}{2}$ and $\xi(\theta''_i)(t)$ are independent.
Then we can take the expectation of $\frac{X_0^t+X_1^t}{2}$ and $\xi(\theta''_i)(t)$ separately.
It is also the case for $1-\frac{X_0^t+X_1^t}{2}$ and $1 - \xi(\theta''_i)(t)$.
Since each entry of $\xi$ follows the uniform distribution over $\{0,1\}$ independently, $\E_{\xi_b}\left[ \xi(\theta''_i)(t) \;\middle|\; \zeta_b,\xi_{b-1} \right] = \E_{\xi_b}\left[ 1-\xi(\theta''_i)(t) \;\middle|\; \zeta_b,\xi_{b-1} \right] = 1/2$.
We thus obtain
\begin{align}
&\E_{\xi_b} \left[ \E_{\theta''_i \sim \calJ_p} \left[ \sum_{t \in \calT_b} Y_{\theta''_i}^t \right] \;\middle|\; \zeta_b,\xi_{b-1} \right] \nonumber\\
&= \E_{\theta''_i \sim \calJ_p} \left[ \sum_{t \in \calT_b} \left\{ \frac{1}{2} \E_{\xi_b} \left[ \frac{X_0^t+X_1^t}{2} \;\middle|\; \zeta_b,\xi_{b-1} \right]
+ \frac{1}{2} \E_{\xi_b} \left[ 1-\frac{X_0^t+X_1^t}{2} \;\middle|\; \zeta_b,\xi_{b-1} \right] \right\} \right]\nonumber\\
&= \frac{|\calT_b|}{2} \nonumber\\
&= \frac{L}{2}.\label{eq:sqrt-gap-lhs4}
\end{align}

Next, we give a lower bound on the second and third terms of \eqref{eq:sqrt-gap-lhs1}.
Since $Z_{\theta''_i}^t$ can be expressed as
\begin{equation*}
Z_{\theta''_i}^t
=
\frac{X_0^t-X_1^t}{2} \xi(\theta''_i)(t) - \frac{X_0^t-X_1^t}{2} (1 - \xi(\theta''_i)(t))
=
\begin{cases}
\frac{X_0^t-X_1^t}{2} & \text{if $\xi(\theta''_i)(t) = 1$}\\
- \frac{X_0^t-X_1^t}{2} & \text{if $\xi(\theta''_i)(t) = 0$},
\end{cases}
\end{equation*}
the cumulative sum $\left( \sum_{t' \in \calT_b \colon t' \le t} Z_{\theta''_i}^{t'} \right)_{t \in \calT_b}$ can be regarded as a random walk with varying step size $\left| \frac{X_0^t-X_1^t}{2} \right|$ for each $\theta''_i \in \calJ_p$.
The step size in round $t \in \calT_b$ is determined by the algorithm that has not observed $\xi(\theta''_i)(t)$ for any $\theta''_i \in \Theta''_i$ yet.
Since $X_0^t \in [0,1]$ and $X_1^t \in [0,1]$, the step size is bounded as $\left| \frac{X_0^t-X_1^t}{2} \right| \in \left[ 0,\frac{1}{2} \right]$ for each $t \in \calT_b$.

Recall that how to partition $\{j_{2k-1},j_{2k}\}$ into $\calJ_{p,0}$ and $\calJ_{p,1}$ is determined independently for each $k \in [2^{B-b}]$.
This is decided by the value of \eqref{eq:quantile},
which can be expressed as
\begin{align*}
&\E_{\theta'_i \sim \calI_{p,0}} \left[ \sum_{t \in \calT_b} \E_{\alpha \sim \pi_i^t(\theta'_i)} \left[ u_i^t(j_{2k-1}, \alpha) - u_i^t(j_{2k}, \alpha) \right] \right]
-
\E_{\theta'_i \sim \calI_{p,1}} \left[ \sum_{t \in \calT_b} \E_{\alpha \sim \pi_i^t(\theta'_i)} \left[ u_i^t(j_{2k-1}, \alpha) - u_i^t(j_{2k}, \alpha) \right] \right]\\
&=
\left\{ \sum_{t \in \calT_b} \left( Y_{j_{2k-1}}^t + Z_{j_{2k-1}}^t \right)
-
\sum_{t \in \calT_b} \left( Y_{j_{2k}}^t + Z_{j_{2k}}^t \right) \right\}
-
\left\{ \sum_{t \in \calT_b} \left( Y_{j_{2k-1}}^t - Z_{j_{2k-1}}^t \right)
-
\sum_{t \in \calT_b} \left( Y_{j_{2k}}^t - Z_{j_{2k}}^t \right) \right\}\\
&=
2 \sum_{t \in \calT_b} \left( Z_{j_{2k-1}}^t - Z_{j_{2k}}^t \right).
\end{align*}
By considering the distribution of this value conditioned on $\zeta_b$ and $\xi_{b-1}$, we set $j_{2k-1} \in \calJ_{p,0}$ for the upper half of this distribution and $j_{2k-1} \in \calJ_{p,1}$ for the lower half.
We define the function $h_k$ such that $h_k(\xi_b) = 1$ if $j_{2k-1} \in \calJ_{p,0}$ and $h_k(\xi_b) = -1$ if $j_{2k-1} \in \calJ_{p,1}$.
Note that $h_k(\xi_b)$ might be random on the median of the distribution.

Let $D_k^t = \sum_{t' \in \calT_b \colon t' \le t} \left( Z_{j_{2k-1}}^{t'} - Z_{j_{2k}}^{t'} \right)$ be half of the cumulative sum of this value until round $t$.
From the definition of $Z_{\theta''_i}^t$, we have
\begin{equation}
D_k^t = \begin{cases}
D_k^{t-1} & \text{with probability $1/2$ (if $\xi(j_{2k-1})(t) = \xi(j_{2k})(t)$)}\\
D_k^{t-1} + (X_0^t-X_1^t) & \text{with probability $1/4$ (if $\xi(j_{2k-1})(t) = 1$ and $\xi(j_{2k})(t) = 0$)}\\
D_k^{t-1} - (X_0^t-X_1^t) & \text{with probability $1/4$ (if $\xi(j_{2k-1})(t) = 0$ and $\xi(j_{2k})(t) = 1$)},
\end{cases}\label{eq:random-walk}
\end{equation}
where $X_0^t-X_1^t$ is determined by the algorithm that has not observed $\xi(j_{2k-1})(t)$ and $\xi(j_{2k})(t)$ yet.

The second and third terms of \eqref{eq:sqrt-gap-lhs1} are expressed as
\begin{align}
&\E_{\xi_b} \left[ \frac{1}{2} \E_{\theta''_i \sim \calJ_{p,0}} \left[ \sum_{t \in \calT_b} Z_{\theta''_i}^t \right]
-
\frac{1}{2} \E_{\theta''_i\sim \calJ_{p,1}} \left[ \sum_{t \in \calT_b} Z_{\theta''_i}^t \right] \;\middle|\; \zeta_b,\xi_{b-1} \right] \nonumber \\
&=
\frac{1}{2^{B-b}} \sum_{k=1}^{2^{B-b}} \frac{1}{2} \E_{\xi_b} \left[ h_k(\xi_b) \left\{ \sum_{t \in \calT_b} Z_{j_{2k-1}}^t
-
\sum_{t \in \calT_b} Z_{j_{2k}}^t \right\} \;\middle|\; \zeta_b,\xi_{b-1} \right] \nonumber \\
&=
\frac{1}{2^{B-b+1}} \sum_{k=1}^{2^{B-b}} \E_{\xi_b} \left[ h_k(\xi_b) D_k^\tmax \;\middle|\; \zeta_b,\xi_{b-1} \right],\label{eq:sqrt-gap-lhs2}
\end{align}
where $\tmax = \max \calT_b$ is the index of the last round in block $b$.

In the remaining part of the proof, we fix $k \in [2^{B-b}]$.
Let $\calF$ be the event that $\left| D_k^t \right| > \frac{\sqrt{c^3L}}{4}$ for some $t \in \calT_b$.
Let $\calF_+$ be the event that $D_k^t > \frac{\sqrt{c^3L}}{4}$ for some $t \in \calT_b$
and $\calF_-$ the event that $D_k^t < \frac{\sqrt{c^3L}}{4}$ for some $t \in \calT_b$.
Since $\Pr_{\xi_b}(\calF \mid \zeta_b,\xi_{b-1}) \le \Pr_{\xi_b}(\calF_+ \mid \zeta_b,\xi_{b-1}) + \Pr_{\xi_b}(\calF_- \mid \zeta_b,\xi_{b-1})$, one of $\calF_+$ or $\calF_-$ happens with probability at least $\Pr_{\xi_b}(\calF \mid \zeta_b,\xi_{b-1})/2$.
Assume $\Pr_{\xi_b}(\calF_+ \mid \zeta_b,\xi_{b-1}) \ge \Pr_{\xi_b}(\calF \mid \zeta_b,\xi_{b-1})/2$ without loss of generality.
For the case of $\Pr_{\xi_b}(\calF_- \mid \zeta_b,\xi_{b-1}) \ge \Pr_{\xi_b}(\calF \mid \zeta_b,\xi_{b-1})/2$, we can prove the statement in the same way.

We define $h'_k$ as follows.
\begin{itemize}
\item
First, we consider the case of $\Pr(\calF_+) \ge 0.5$.
In this case, we set $h'_k(\xi_b) = 1$ for the upper part of $0.5 / \Pr(\calF_+)$ on the distribution of $D_k^\tmax$ conditioned on $\zeta_b$, $\xi_{b-1}$, and $\calF_+$.
We set $h'_k(\xi_b) = -1$ for the remaining part of $\calF_+$ and the complementary event $\calF_+^c$.

\item
Next, we consider the case of $\Pr(\calF_+) < 0.5$.
If $\calF_+$ holds, we set $h'_k(\xi_k) = 1$.
If $\calF_+$ does not hold, we decide $h'_k(\xi_k) = 1$ or $h'_k(\xi_k) = -1$ by considering the distribution of $D_k^\tmax$.
For the upper part of probability $(0.5-\Pr(\calF_+))/(1-\Pr(\calF_+))$ on the distribution of $D_k^\tmax$ conditioned on $\calF_+^c$, we set $h'_k(\xi_b) = 1$.
For the lower part of probability $0.5/(1-\Pr(\calF_+))$, we set $h'_k(\xi_b) = -1$.
\end{itemize}
In both cases, the mass on the boundary can be arbitrarily partitioned so that the above constraint is satisfied.
Note that $h'_k$ satisfies the constraint $\Pr_{\xi_b}(h'_k(\xi_b) = 1 \mid \zeta_b,\xi_{b-1}) = 0.5$.
Since $h_k$ maximizes $\E_{\xi_b}\left[ h_k(\xi_b) D_k^\tmax \;\middle|\; \zeta_b,\xi_{b-1} \right]$ under this constraint, we have
\begin{equation*}
\E_{\xi_b}\left[ h_k(\xi_b) D_k^\tmax \;\middle|\; \zeta_b,\xi_{b-1} \right]
\ge 
\E_{\xi_b}\left[ h'_k(\xi_b) D_k^\tmax \;\middle|\; \zeta_b,\xi_{b-1} \right].
\end{equation*}

In the following, we prove a lower bound
\begin{equation}
\E_{\xi_b} \left[ h_k(\xi_b) D_k^\tmax \;\middle|\; \zeta_b,\xi_{b-1} \right]
\ge \frac{1}{2} \Pr_{\xi_b} \left(\calF \;\middle|\; \zeta_b,\xi_{b-1} \right) \E_{\xi_b} \left[ D_k^\tmax \;\middle|\; \zeta_b,\xi_{b-1},\calF_+ \right] \label{eq:sqrt-gap-lhs3}
\end{equation}
on each term of \eqref{eq:sqrt-gap-lhs2}.
Since $D_k^t$ is a random walk specified as \eqref{eq:random-walk}, it is a martingale.
Therefore, $\E_{\xi_b}\left[D_k^\tmax \;\middle|\; \zeta_b,\xi_{b-1}\right] = 0$.
From the optimality of $h_k$, the left hand side of \eqref{eq:sqrt-gap-lhs3} is non-negative.
If $\E_{\xi_b} \left[ D_k^\tmax \;\middle|\; \zeta_b,\xi_{b-1},\calF_+ \right] < 0$, \eqref{eq:sqrt-gap-lhs3} immediately follows.
In the following, assuming $\E_{\xi_b} \left[ D_k^\tmax \;\middle|\; \zeta_b,\xi_{b-1},\calF_+ \right] \ge 0$, we prove \eqref{eq:sqrt-gap-lhs3} for each case separately.
Let $D_k^*$ be the maximum value of $D_k^\tmax$ when $h'_k(\xi_b) = -1$ and $\calF_+^c$.

\begin{itemize}
\item
First, we consider the case of $\Pr(\calF_+) \ge 0.5$.
It holds that
\begin{align*}
&\E_{\xi_b}\left[D_k^\tmax \;\middle|\; \zeta_b,\xi_{b-1}\right]\\
&=
\frac{1}{2}\E_{\xi_b}\left[D_k^\tmax \;\middle|\; \zeta_b,\xi_{b-1},h'_k(\xi_b)=1\right]
+
\frac{1}{2}\E_{\xi_b}\left[D_k^\tmax \;\middle|\; \zeta_b,\xi_{b-1},h'_k(\xi_b)=-1\right],
\end{align*}
where we used the fact that the probabilities of $h'_k(\xi_b) = 1$ and $h'_k(\xi_b) = -1$ are both $1/2$.
Combining it with $\E_{\xi_b}\left[D_k^\tmax \;\middle|\; \zeta_b,\xi_{b-1}\right] = 0$, we obtain
\begin{align*}
\frac{1}{2}\E_{\xi_b}\left[D_k^\tmax \;\middle|\; \zeta_b,\xi_{b-1},h'_k(\xi_b)=1\right]
=
- \frac{1}{2}\E_{\xi_b}\left[D_k^\tmax \;\middle|\; \zeta_b,\xi_{b-1},h'_k(\xi_b)=-1\right].
\end{align*}
We thus obtain
\begin{align*}
&\E_{\xi_b} \left[ h_k(\xi_b) D_k^\tmax \;\middle|\; \zeta_b,\xi_{b-1} \right]\\
&\ge \E_{\xi_b} \left[ h'_k(\xi_b) D_k^\tmax \;\middle|\; \zeta_b,\xi_{b-1} \right]\tag{due to the optimality of $h_k$}\\
&=
\frac{1}{2}\E_{\xi_b}\left[D_k^\tmax \;\middle|\; \zeta_b,\xi_{b-1},h'_k(\xi_b)=1\right]
-
\frac{1}{2}\E_{\xi_b}\left[D_k^\tmax \;\middle|\; \zeta_b,\xi_{b-1},h'_k(\xi_b)=-1\right]\\
&=
\E_{\xi_b}\left[D_k^\tmax \;\middle|\; \zeta_b,\xi_{b-1},h'_k(\xi_b)=1\right]\\
&\ge
\E_{\xi_b}\left[D_k^\tmax \;\middle|\; \zeta_b,\xi_{b-1},\calF_+\right]
\tag{since $h'_k(\xi_b) = 1$ holds for the upper part of $\calF_+$}\\
&\ge
\frac{1}{2} \Pr_{\xi_b} \left(\calF \;\middle|\; \zeta_b,\xi_{b-1} \right) 
\E_{\xi_b}\left[D_k^\tmax \;\middle|\; \zeta_b,\xi_{b-1},\calF_+ \right].
\end{align*}

\item
If $\Pr(\calF_+) < 0.5$ and $D^*_k \ge 0$, the former part of the proof in the previous case also applies to this case.
The remaining part is
\begin{align*}
&\E_{\xi_b} \left[ h_k(\xi_b) D_k^\tmax \;\middle|\; \zeta_b,\xi_{b-1} \right]\\
&\ge
\E_{\xi_b}\left[D_k^\tmax \;\middle|\; \zeta_b,\xi_{b-1},h'_k(\xi_b)=1\right]\\
&= \Pr_{\xi_b} \left(\calF_+ \;\middle|\; \zeta_b,\xi_{b-1},h'_k(\xi_b)=1 \right) \E_{\xi_b} \left[ D_k^\tmax \;\middle|\; \zeta_b,\xi_{b-1},\calF_+,h'_k(\xi_b)=1 \right] \\
&\qquad + \Pr_{\xi_b} \left(\calF^c_+ \;\middle|\; \zeta_b,\xi_{b-1},h'_k(\xi_b)=1 \right) \E_{\xi_b} \left[ D_k^\tmax \;\middle|\; \zeta_b,\xi_{b-1},\calF_+^c,h'_k(\xi_b)=1 \right] \\
&\ge \Pr_{\xi_b} \left(\calF_+ \;\middle|\; \zeta_b,\xi_{b-1},h'_k(\xi_b)=1 \right) \E_{\xi_b} \left[ D_k^\tmax \;\middle|\; \zeta_b,\xi_{b-1},\calF_+,h'_k(\xi_b)=1 \right]\\
&= \frac{ \Pr_{\xi_b} \left( \calF_+\;\middle|\; \zeta_b,\xi_{b-1} \right) }{ \Pr_{\xi_b} \left(h'_k(\xi_b)=1 \;\middle|\; \zeta_b,\xi_{b-1} \right) }  \E_{\xi_b} \left[ D_k^\tmax \;\middle|\; \zeta_b,\xi_{b-1},\calF_+ \right]\\
&= 2 \Pr_{\xi_b} \left( \calF_+\;\middle|\; \zeta_b,\xi_{b-1} \right)  \E_{\xi_b} \left[ D_k^\tmax \;\middle|\; \zeta_b,\xi_{b-1},\calF_+ \right]\\
&\ge \frac{1}{2} \Pr_{\xi_b} \left(\calF \;\middle|\; \zeta_b,\xi_{b-1} \right) \E_{\xi_b} \left[ D_k^\tmax \;\middle|\; \zeta_b,\xi_{b-1},\calF_+ \right].
\end{align*}
The first inequality is due to the non-negativity of $D_k^\tmax$ conditioned on $\calF^c_+$ and $h'_k(\xi_b) = 1$, which is implied by $D_k^* \ge 0$.
The second equality holds because $\calF_+$ holds only when $h'_k(\xi_b) = 1$.
The last inequality is due to the assumption $\Pr_{\xi_b}(\calF_+ \mid \zeta_b,\xi_{b-1}) \ge \Pr_{\xi_b}(\calF \mid \zeta_b,\xi_{b-1})/2$.

\item
If $\Pr(\calF_+) < 0.5$ and $D^*_k < 0$, we obtain
\begin{align*}
&\E_{\xi_b} \left[ h_k(\xi_b) D_k^\tmax \;\middle|\; \zeta_b,\xi_{b-1} \right] \\
&\ge \E_{\xi_b} \left[ h'_k(\xi_b) D_k^\tmax \;\middle|\; \zeta_b,\xi_{b-1} \right] \tag{due to the optimality of $h_k$}\\
&= \Pr_{\xi_b} \left(\calF_+ \;\middle|\; \zeta_b,\xi_{b-1} \right) \E_{\xi_b} \left[ D_k^\tmax \;\middle|\; \zeta_b,\xi_{b-1},\calF_+ \right]\\
& \qquad + \Pr_{\xi_b} \left(\calF_+^c, h'_k(\xi_b)=1 \;\middle|\; \zeta_b,\xi_{b-1}\right) \E_{\xi_b} \left[ D_k^\tmax \;\middle|\; \zeta_b,\xi_{b-1},\calF_+^c,h'_k(\xi_b) = 1 \right]\\
& \qquad - \Pr_{\xi_b} \left(\calF_+^c, h'_k(\xi_b)=-1 \;\middle|\; \zeta_b,\xi_{b-1} \right) \E_{\xi_b} \left[ D_k^\tmax \;\middle|\; \zeta_b,\xi_{b-1},\calF_+^c,h'_k(\xi_b) = -1 \right]\\
&\ge \Pr_{\xi_b} \left(\calF_+ \;\middle|\; \zeta_b,\xi_{b-1} \right) \E_{\xi_b} \left[ D_k^\tmax \;\middle|\; \zeta_b,\xi_{b-1},\calF_+ \right]\\
& \qquad + \Pr_{\xi_b} \left(\calF_+^c, h'_k(\xi_b)=1 \;\middle|\; \zeta_b,\xi_{b-1}\right) D_k^*
- \Pr_{\xi_b} \left(\calF_+^c, h'_k(\xi_b)=-1 \;\middle|\; \zeta_b,\xi_{b-1} \right) D_k^*\\
&\ge \Pr_{\xi_b} \left(\calF_+ \;\middle|\; \zeta_b,\xi_{b-1} \right) \E_{\xi_b} \left[ D_k^\tmax \;\middle|\; \zeta_b,\xi_{b-1},\calF_+ \right] \\
&\ge \frac{1}{2} \Pr_{\xi_b} \left(\calF \;\middle|\; \zeta_b,\xi_{b-1} \right) \E_{\xi_b} \left[ D_k^\tmax \;\middle|\; \zeta_b,\xi_{b-1},\calF_+ \right].
\end{align*}
The second inequality is due to the definition of $D_k^*$.
The third inequality is due to the negativity of $D_k^*$ and $\Pr_{\xi_b} \left(\calF_+^c, h'_k(\xi_b)=1 \;\middle|\; \zeta_b,\xi_{b-1}\right) \le \Pr_{\xi_b} \left(\calF_+^c, h'_k(\xi_b)=-1 \;\middle|\; \zeta_b,\xi_{b-1} \right)$.
\end{itemize}

We thus proved \eqref{eq:sqrt-gap-lhs3}.
In the remaining part of the proof, we give a lower bound on each of $\Pr_{\xi_b}\left(\calF \;\middle|\; \zeta_b,\xi_{b-1} \right)$ and $\E_{\xi_b} \left[ D_k^\tmax \;\middle|\; \zeta_b,\xi_{b-1},\calF_+ \right]$.

First, we show that the expected value of $D_k^\tmax$ is at least $\frac{\sqrt{c^3L}}{4}$ conditioned on $\calF_+$.
Let $t^*_+$ be the first time round that $D^t_k \ge \frac{\sqrt{c^3L}}{4}$ holds, which is well-defined conditioned on $\calF_+$.
For each $t \in \calT_b$, if we condition on $\calF_+$ and $t_+^* = t$, we have
\begin{align*}
\E_{\xi_b} \left[ D_k^{t'+1} \;\middle|\; \zeta_b,\xi_{b-1},\calF_+, t_+^*=t \right]
=
\E_{\xi_b} \left[ D_k^{t'} \;\middle|\; \zeta_b,\xi_{b-1},\calF_+, t_+^*=t \right]
\end{align*}
for any $t' \ge t$
because whether the events $\calF_+$ and $t^* = t$ hold is determined until round $t'$.
By considering the expected value for different $t_+^*$ separately,
\begin{align}
\E_{\xi_b} \left[ D_k^\tmax \;\middle|\; \zeta_b,\xi_{b-1},\calF_+ \right]
&= \sum_{t \in \calT_b} \Pr_{\xi_b} \left(t_+^*=t \;\middle|\; \zeta_b,\xi_{b-1},\calF_+ \right) \E_{\xi_b} \left[ D_k^\tmax \;\middle|\; \zeta_b,\xi_{b-1},\calF_+, t_+^*=t \right] \nonumber \\
&= \sum_{t \in \calT_b} \Pr_{\xi_b} \left(t_+^*=t \;\middle|\; \zeta_b,\xi_{b-1},\calF_+ \right) \E_{\xi_b} \left[ D_k^t \;\middle|\; \zeta_b,\xi_{b-1},\calF_+, t_+^*=t \right] \nonumber \\
&\ge \sum_{t \in \calT_b} \Pr_{\xi_b}(t_+^*=t \mid \calF_+) ~ \frac{\sqrt{c^3L}}{4} \tag{from the definition of $t^*_+$} \nonumber \\
&= \frac{\sqrt{c^3L}}{4}. \label{eq:expectation-lb}
\end{align}

Next, we provide a lower bound on $\Pr_{\xi_b}\left(\calF \;\middle|\; \zeta_b,\xi_{b-1}\right)$ by using the assumption
\begin{equation*}
\E_{\xi_b} \left[ \sum_{t \in \calT_b} \left| \E_{\theta_i \sim \calI_{p,0}} \left[ \pi_i^t(\theta_i; \alpha_0) \right]
-
\E_{\theta_i \sim \calI_{p,1}} \left[ \pi_i^t(\theta_i; \alpha_0) \right] \right| \;\middle|\; \zeta_b,\xi_{b-1} \right]
\ge cL.
\end{equation*}
Using random variables $X_0^t$ and $X_1^t$, this assumption can be written as
\begin{equation*}
\E_{t \sim \calT_b, \xi_b} \left[ \left| X_0^t - X_1^t \right| \;\middle|\; \zeta_b,\xi_{b-1} \right]
\ge c.
\end{equation*}
By applying \Cref{lem:gap-lower-bound}, we obtain
\begin{equation*}
\Pr_{t \sim \calT_b, \xi_b} \left( \left( X_0^t - X_1^t \right)^2 \ge \frac{c^2}{4} \;\middle|\; \zeta_b,\xi_{b-1} \right)
=
\Pr_{t \sim \calT_b, \xi_b} \left( \left| X_0^t - X_1^t \right| \ge \frac{c}{2} \;\middle|\; \zeta_b,\xi_{b-1} \right)
\ge \frac{c}{2}.
\end{equation*}
We thus obtain
\begin{equation}
\E_{\xi_b} \left[ \sum_{t \in \calT_b} \left( X_0^t - X_1^t \right)^2 \;\middle|\; \zeta_b,\xi_{b-1} \right]
=
L \E_{t \sim \calT_b, \xi_b} \left[ \left( X_0^t - X_1^t \right)^2 \;\middle|\; \zeta_b,\xi_{b-1} \right]
= L \cdot \frac{c}{2} \cdot \frac{c^2}{4}
\ge \frac{c^3 L}{8}. \label{eq:square-lb}
\end{equation}

We show that a stochastic process $(D_k^t)^2 - \frac{1}{2} \sum_{t' \in \calT_b \colon t' \le t} \left( X_0^{t'} - X_1^{t'} \right)^2$ for $t \in \calT_b$ is a martingale.
We denote by $\calF_t$ all the events that happen until the end of time round $t$.
Conditioned on $\calF_{t-1}$, the expected difference is
\begin{align*}
&\E_{\xi_b} \left[ (D_k^t)^2 - \frac{1}{2} \sum_{t' \in \calT_b \colon t' \le t} \left( X_0^{t'} - X_1^{t'} \right)^2 \;\middle|\; \zeta_b,\xi_{b-1},\calF_{t-1} \right]
- \left\{ (D_k^{t-1})^2 - \frac{1}{2} \sum_{t' \in \calT_b \colon t' \le t-1} \left( X_0^{t'} - X_1^{t'} \right)^2 \right\}\\
&= \E_{\xi_b} \left[ (D_k^t)^2 - (D_k^{t-1})^2 \;\middle|\; \zeta_b,\xi_{b-1},\calF_{t-1} \right] - \frac{1}{2} \left( X_0^t - X_1^t \right)^2\\
&= \E_{\xi_b} \left[ (D_k^t - D_k^{t-1})(D_k^t + D_k^{t-1}) \;\middle|\; \zeta_b,\xi_{b-1},\calF_{t-1} \right] - \frac{1}{2} \left( X_0^t - X_1^t \right)^2\\
&= \E_{\xi_b} \left[ (D_k^t - D_k^{t-1})(2 D_k^{t-1} + (D_k^t - D_k^{t-1})) \;\middle|\; \zeta_b,\xi_{b-1},\calF_{t-1} \right] - \frac{1}{2} \left( X_0^t - X_1^t \right)^2 \\
&= 2 D_k^{t-1} \E_{\xi_b} \left[ D_k^t - D_k^{t-1} \;\middle|\; \zeta_b,\xi_{b-1},\calF_{t-1} \right] + \E_{\xi_b} \left[ (D_k^t - D_k^{t-1})^2 \;\middle|\; \zeta_b,\xi_{b-1},\calF_{t-1} \right] - \frac{1}{2} \left( X_0^t - X_1^t \right)^2\\
&= 0 + \frac{1}{2} \left( X_0^t - X_1^t \right)^2 - \frac{1}{2} \left( X_0^t - X_1^t \right)^2 \tag{since \eqref{eq:random-walk}} \\
&= 0.
\end{align*}

We define the random variable $t^* \in \calT_b$ such that
\begin{equation*}
t^* = \begin{cases}
\min \left\{ t \in \calT_b \;\middle|\; |D_k^t| \ge \frac{\sqrt{c^3L}}{4} \right\} & \text{if $\calF$ holds}\\
\tmax & \text{otherwise}
\end{cases}
\end{equation*}
is the time round when $|D_k^t|$ exceeds $\frac{\sqrt{c^3L}}{4}$ for the first time if $\calF$ holds.
Since the stopping time $t^*$ is bounded above by $\tmax$, we can apply Doob's optional stopping theorem to the martingale defined above and obtain
\begin{equation*}
\E_{\xi_b} \left[ (D_k^{t^*})^2 - \sum_{t \in \calT_b \colon t \le t^*} \left( X_0^t - X_1^t \right)^2 \;\middle|\; \zeta_b,\xi_{b-1} \right] = 0.
\end{equation*}
Since $\sum_{t \in \calT_b \colon t > t^*} \left( X_0^t - X_1^t \right)^2$ is $0$ if $\calF$ does not hold and bounded above by $L$ if $\calF$ holds, we have
\begin{align*}
&\E_{\xi_b} \left[ \sum_{t \in \calT_b} \left( X_0^t - X_1^t \right)^2  \;\middle|\; \zeta_b,\xi_{b-1} \right]\\
&= \E_{\xi_b} \left[ \sum_{t \in \calT_b \colon t \le t^*} \left( X_0^t - X_1^t \right)^2  \;\middle|\; \zeta_b,\xi_{b-1}\right] + \E_{\xi_b} \left[ \sum_{t \in \calT_b \colon t > t^*} \left( X_0^t - X_1^t \right)^2  \;\middle|\; \zeta_b,\xi_{b-1} \right]\\
&\le \E_{\xi_b} \left[ (D_k^{t^*})^2 \;\middle|\; \zeta_b,\xi_{b-1} \right] 
+ \Pr_{\xi_b}\left(\calF \;\middle|\; \zeta_b,\xi_{b-1} \right) L\\
&= \Pr_{\xi_b}\left(\calF \;\middle|\; \zeta_b,\xi_{b-1} \right) \E_{\xi_b} \left[ (D_k^{t^*})^2 \;\middle|\; \zeta_b,\xi_{b-1},\calF \right]
+ \Pr_{\xi_b}\left(\calF^c \;\middle|\; \zeta_b,\xi_{b-1} \right) \E_{\xi_b} \left[ (D_k^{t^*})^2 \;\middle|\; \zeta_b,\xi_{b-1},\calF^c \right]
+ \Pr_{\xi_b}\left(\calF \;\middle|\; \zeta_b,\xi_{b-1} \right) L\\
&\le \Pr_{\xi_b}\left(\calF \;\middle|\; \zeta_b,\xi_{b-1} \right) \left(\frac{\sqrt{c^3L}}{4} + 1 \right)^2
+\Pr_{\xi_b}\left(\calF^c \;\middle|\; \zeta_b,\xi_{b-1} \right) \left(\frac{\sqrt{c^3L}}{4}\right)^2
+ \Pr_{\xi_b}\left(\calF \;\middle|\; \zeta_b,\xi_{b-1} \right) L\\
&\le 4 L \Pr_{\xi_b}\left(\calF \;\middle|\; \zeta_b,\xi_{b-1} \right)
+ \frac{c^3L}{16},
\end{align*}
where we used $c \le 1$, $\sqrt{L} \le L$, $L \ge 1$, and $\Pr_{\xi_b}\left(\calF^c \;\middle|\; \zeta_b,\xi_{b-1} \right) \le 1$ for the last inequality.
Combining it with \eqref{eq:square-lb}, we obtain
\begin{equation}
\Pr_{\xi_b}\left(\calF \;\middle|\; \zeta_b,\xi_{b-1} \right) \ge \frac{c^3L/8 - c^3L/16}{4L} = \frac{c^3}{64}.\label{eq:prob-lb}
\end{equation}

Plugging \eqref{eq:expectation-lb} and \eqref{eq:prob-lb} into \eqref{eq:sqrt-gap-lhs3}, we obtain
\begin{align*}
&\E_{\xi_b} \left[
\frac{1}{2} \E_{\substack{\theta'_i \sim \calI_{p,0} \\ \theta''_i \sim \calJ_{p,0} }} \left[ \sum_{t \in \calT_b} \E_{\alpha \sim \pi_i^t(\theta'_i)} \left[ u_i^t(\theta''_i, \alpha) \right] \right]
+
\frac{1}{2} \E_{\substack{ \theta'_i \sim \calI_{p,1} \\ \theta''_i \sim \calJ_{p,1} }} \left[ \sum_{t \in \calT_b} \E_{\alpha \sim \pi_i^t(\theta'_i)} \left[ u_i^t(\theta''_i, \alpha) \right] \right]
\;\middle|\; \zeta_b, \xi_{b-1} \right]\\
&=
\E_{\xi_b} \left[ \E_{\theta''_i \sim \calJ_p} \left[ \sum_{t \in \calT_b} Y_{\theta''_i}^t \right]
+
\frac{1}{2} \E_{\theta''_i \sim \calJ_{p,0}} \left[ \sum_{t \in \calT_b} Z_{\theta''_i}^t \right]
-
\frac{1}{2} \E_{\theta''_i\sim \calJ_{p,1}} \left[ \sum_{t \in \calT_b} Z_{\theta''_i}^t \right] \;\middle|\; \zeta_b,\xi_{b-1} \right]
\tag{since \eqref{eq:sqrt-gap-lhs1}}\\
&= \frac{L}{2} +
\frac{1}{2^{B-b+1}} \sum_{k=1}^{2^{B-b}} \E_{\xi_b} \left[ h_k(\xi_b) D_k^\tmax \;\middle|\; \zeta_b,\xi_{b-1} \right]
\tag{since \eqref{eq:sqrt-gap-lhs4} and \eqref{eq:sqrt-gap-lhs2}}
\\
&\ge \frac{L}{2} +
\frac{1}{2^{B-b+1}} \sum_{k=1}^{2^{B-b}} \frac{1}{2} \Pr_{\xi_b} \left(\calF \;\middle|\; \zeta_b,\xi_{b-1} \right) \E_{\xi_b} \left[ D_k^\tmax \;\middle|\; \zeta_b,\xi_{b-1},\calF_+ \right]
\tag{since \eqref{eq:sqrt-gap-lhs3}}
\\
&\ge \frac{L}{2} +
\frac{1}{2^{B-b+1}} \sum_{k=1}^{2^{B-b}} \frac{1}{2} \cdot \frac{c^3}{64} \cdot \frac{\sqrt{c^3L}}{4}
\tag{since \eqref{eq:expectation-lb} and \eqref{eq:prob-lb}}
\\
&= \frac{L}{2} + 2^{-10} c^{9/2}\sqrt{L},
\end{align*}
which completes the proof.
\end{proof}

\subsection{Proof of lower bound}\label{sec:lb4}

Finally, we prove the theorem using \Cref{lem:psi-dist,lem:gap-block,lem:martingale}.

\thmuntruthfullb*

\begin{proof}
Since the rewards for type $\theta''_i \in \Theta''_i$ are determined by independent Bernoulli random variables $\xi$, and therefore the expected reward obtained by the algorithm for each round is equal to $1/2$.
Formally, since the algorithm's decision $\pi_i^t(\theta''_i)$ is independent from $u_i^t(\theta''_i,\alpha_0) = 1 - u_i^t(\theta''_i,\alpha_1) = \xi(\theta''_i)(t)$, the expected reward is
\begin{align*}
\E_{\zeta,\xi} \left[ \E_{\alpha \sim \pi_i^t(\theta''_i)} \left[ u_i^t(\theta''_i, \alpha) \right] \right]
&=
\E_{\zeta,\xi} \left[ \pi_i^t(\theta''_i;\alpha_0) \xi(\theta''_i)(t) + (1-\pi_i^t(\theta''_i;\alpha_0)) (1-\xi(\theta''_i)(t)) \right]\\
&=
\pi_i^t(\theta''_i;\alpha_0) \cdot \frac{1}{2} + (1-\pi_i^t(\theta''_i;\alpha_0)) \cdot \frac{1}{2} \\
&= \frac{1}{2}.
\end{align*}
We consider $\phi$ in the definition of untruthful swap regret such that $\phi(\theta_i,\cdot)$ is the identity map for every $\theta_i \in \Theta_i$.
By considering only the types in $\Theta''_i$, we obtain
\begin{align*}
\E_{\zeta,\xi} \left[ \RUSi^T \right]
&\ge \E_{\zeta,\xi} \left[ \max_{\psi \colon \Theta_i \to \Theta_i}\sum_{t=1}^T \E_{\theta_i \sim \rho_i} \left[ \E_{\alpha \sim \pi_i^t(\psi(\theta_i))} \left[ u_i^t(\theta_i, \alpha) \right] - \E_{\alpha \sim \pi_i^t(\theta_i)} \left[ u_i^t(\theta_i, \alpha) \right] \right] \right]\\
&\ge \E_{\zeta,\xi} \left[ \max_{\psi \colon \Theta_i \to \Theta_i} \frac{1}{|\Theta_i|} \sum_{\theta''_i \in \Theta''_i} \sum_{t=1}^T \left\{ \E_{\alpha \sim \pi_i^t(\psi(\theta''_i))} \left[ u_i^t(\theta''_i,\alpha) \right] - \E_{\alpha \sim \pi_i^t(\theta''_i)} \left[ u_i^t(\theta''_i, \alpha) \right] \right\} \right]\\
&= \E_{\zeta,\xi} \left[ \max_{\psi \colon \Theta_i \to \Theta_i} \frac{1}{|\Theta_i|} \sum_{\theta''_i \in \Theta''_i} \sum_{t=1}^T \E_{\alpha \sim \pi_i^t(\psi(\theta''_i))} \left[ u_i^t(\theta''_i,\alpha) \right] \right] - \frac{1}{|\Theta_i|} \sum_{\theta''_i \in \Theta''_i} \sum_{t=1}^T \frac{1}{2}\\
&= \frac{1}{2} \E_{\zeta,\xi} \left[ \max_{\psi \colon \Theta_i \to \Theta_i} \frac{1}{|\Theta''_i|} \sum_{\theta''_i \in \Theta''_i} \sum_{t=1}^T \E_{\alpha \sim \pi_i^t(\psi(\theta''_i))} \left[ u_i^t(\theta''_i,\alpha) \right] \right] - \frac{T}{4},
\end{align*}
where the last equality is due to $|\Theta''_i| = |\Theta_i|/2$.
By using $\psi'_{\zeta,\xi}$ defined in \Cref{sec:lb2} for each $\zeta$ and $\xi$, we obtain
\begin{align}
\E_{\zeta,\xi} \left[ \RUSi^T \right]
&\ge \frac{1}{2} \E_{\zeta,\xi} \left[ \frac{1}{|\Theta''_i|} \sum_{\theta''_i \in \Theta''_i} \sum_{t=1}^T \E_{\alpha \sim \pi_i^t(\psi'_{\zeta,\xi}(\theta''_i))} \left[ u_i^t(\theta''_i, \alpha) \right] \right] - \frac{T}{4} \nonumber \\
&= \frac{1}{2} \sum_{b \in [B]} \E_{\zeta,\xi} \left[ \frac{1}{|\Theta''_i|} \sum_{\theta''_i \in \Theta''_i} \sum_{t \in \calT_b} \E_{\alpha \sim \pi_i^t(\psi'_{\zeta,\xi}(\theta''_i))} \left[ u_i^t(\theta''_i, \alpha) \right] - \frac{L}{2} \right] \nonumber \\
&= \frac{B}{2} \left\{ \E_{b \sim [B]} \left[ \E_{\zeta,\xi} \left[ \E_{\theta''_i \sim \Theta''_i} \left[ \sum_{t \in \calT_b} \E_{\alpha \sim \pi_i^t(\psi'_{\zeta,\xi}(\theta''_i))} \left[ u_i^t(\theta''_i, \alpha) \right] \right] \right] \right] - \frac{L}{2} \right\}. \label{eq:usr-lb}
\end{align}
For each $b \in [B]$, since $(\calJ_p)_{p \in \{0,1\}^{b-1}}$ is the equally sized partition of $\Theta''_i$, each term can be expressed as
\begin{align*}
&\E_{\zeta,\xi} \left[ \E_{\theta''_i \sim \Theta''_i} \left[ \sum_{t \in \calT_b} \E_{\alpha \sim \pi_i^t(\psi'_{\zeta,\xi}(\theta''_i))} \left[ u_i^t(\theta''_i, \alpha) \right] \right] \right]\\
&= \E_{\zeta,\xi} \left[ \E_{p \sim \{0,1\}^{b-1}} \left[ \E_{\theta''_i \sim \calJ_p} \left[ \sum_{t \in \calT_b} \E_{\alpha \sim \pi_i^t(\psi'_{\zeta,\xi}(\theta''_i))} \left[ u_i^t(\theta''_i, \alpha) \right] \right] \right] \right] \\
&= \E_{\substack{\zeta_b,\xi_{b-1} \\ p \sim \{0,1\}^{b-1} }} \left[ \E_{\zeta,\xi} \left[ \E_{\theta''_i \sim \calJ_p} \left[ \sum_{t \in \calT_b} \E_{\alpha \sim \pi_i^t(\psi'_{\zeta,\xi}(\theta''_i))} \left[ u_i^t(\theta''_i, \alpha) \right] \right] \;\middle|\; \zeta_b,\xi_{b-1} \right] \right]\\
&= \E_{\substack{\zeta_b,\xi_{b-1} \\ p \sim \{0,1\}^{b-1} }} \left[ \E_{\zeta,\xi} \left[ \frac{1}{2} \E_{\theta''_i \sim \calJ_{p,0}} \left[ \sum_{t \in \calT_b} \E_{\alpha \sim \pi_i^t(\psi'_{\zeta,\xi}(\theta''_i))} \left[ u_i^t(\theta''_i, \alpha) \right] \right] \right. \right.\\
&\qquad \qquad \left. \left. + \frac{1}{2} \E_{\theta''_i \sim \calJ_{p,1}} \left[ \sum_{t \in \calT_b} \E_{\alpha \sim \pi_i^t(\psi'_{\zeta,\xi}(\theta''_i))} \left[ u_i^t(\theta''_i, \alpha) \right] \right] \;\middle|\; \zeta_b,\xi_{b-1} \right] \right]\\
&= \E_{\substack{\zeta_b,\xi_{b-1} \\ p \sim \{0,1\}^{b-1} }} \left[ \E_{\xi_b} \left[ \frac{1}{2} \E_{\theta''_i \sim \calJ_{p,0}} \left[ \E_{\zeta,\xi} \left[ \sum_{t \in \calT_b} \E_{\alpha \sim \pi_i^t(\psi'_{\zeta,\xi}(\theta''_i))} \left[ u_i^t(\theta''_i, \alpha) \right] \;\middle|\; \zeta_b,\xi_b \right] \right. \right. \right.\\
&\qquad\qquad \left. \left. + \frac{1}{2} \E_{\theta''_i \sim \calJ_{p,1}} \left[ \E_{\zeta,\xi} \left[ \sum_{t \in \calT_b} \E_{\alpha \sim \pi_i^t(\psi'_{\zeta,\xi}(\theta''_i))} \left[ u_i^t(\theta''_i, \alpha) \right] \;\middle|\; \zeta_b,\xi_b \right] \right] \;\middle|\; \zeta_b,\xi_{b-1} \right] \right].
\end{align*}
Conditioned on $\zeta_b$ and $\xi_b$, all of $\calJ_{p,0}$, $\calJ_{p,1}$, $(\pi_i^t)_{t \in \calT_b}$, and $(u_i^t)_{t \in \calT_b}$ are deterministic.
For taking the expectation with respect to $\zeta$ and $\xi$ conditioned on $\zeta_b$ and $\xi_b$,
we need to consider only the randomness of $\psi'_{\zeta,\xi}(\theta''_i)$.
From \Cref{lem:psi-dist}, $\psi'_{\zeta,\xi}(\theta''_i)$ is distributed uniformly over $\calI_{p,0}$ for each $\theta''_i \in \calJ_{p,0}$.
Similarly, $\psi'_{\zeta,\xi}(\theta''_i)$ is distributed uniformly over $\calI_{p,1}$ for each $\theta''_i \in \calJ_{p,1}$.
Therefore, the above term can be expressed as
\begin{align*}
&\E_{\zeta,\xi} \left[ \E_{\theta''_i \sim \Theta''_i} \left[ \sum_{t \in \calT_b} \E_{\alpha \sim \pi_i^t(\psi'_{\zeta,\xi}(\theta''_i))} \left[ u_i^t(\theta''_i, \alpha) \right] \right] \right]\\
&= \E_{\substack{\zeta_b,\xi_{b-1} \\ p \sim \{0,1\}^{b-1} }} \left[ \E_{\xi_b} \left[ \frac{1}{2} \E_{\substack{ \theta'_i \sim \calI_{p,0} \\ \theta''_i \sim \calJ_{p,0} }} \left[ \sum_{t \in \calT_b} \E_{\alpha \sim \pi_i^t(\theta'_i)} \left[ u_i^t(\theta''_i, \alpha) \right] \right] 
+ \frac{1}{2} \E_{\substack{ \theta'_i \sim \calI_{p,0} \\ \theta''_i \sim \calJ_{p,1} }} \left[ \sum_{t \in \calT_b} \E_{\alpha \sim \pi_i^t(\theta'_i)} \left[ u_i^t(\theta''_i, \alpha) \right] \right] \;\middle|\; \zeta_b,\xi_{b-1} \right] \right].
\end{align*}
For each $b \in [B]$, $\zeta_b$, $\xi_{b-1}$, and $p \in \{0,1\}^{b-1}$, let $\calF_{b,\zeta_b,\xi_{b-1},p}$ be the event that
\begin{equation*}
\E_{\xi_b} \left[ \sum_{t \in \calT_b} \left| \E_{\theta_i \sim \calI_{p,0}} \left[ \pi_i^t(\theta_i;\alpha_0) \right] - \E_{\theta_i \sim \calI_{p,1}} \left[ \pi_i^t(\theta_i;\alpha_0) \right] \right| \;\middle|\; \zeta_b,\xi_{b-1} \right]
\ge \frac{L}{8}.
\end{equation*}

Toward a contradiction, assume $\E_{\zeta,\xi} \left[ \RUSi^T \right] < 2^{-28} \sqrt{T \log_2 |\Theta_i|}$.
From \Cref{lem:gap-block}, we have $\Pr_{b,\zeta_b,\xi_{b-1},p}(\calF_{b,\zeta_b,\xi_{b-1},p}) \ge 1/8$.
Next, we apply \Cref{lem:martingale} to each $b$, $\zeta_b$, $\xi_{b-1}$ and $p$.
If $\calF_{b,\zeta_b,\xi_{b-1},p}$ holds, we apply \Cref{lem:martingale} with $c = 1/8$.
If $\calF_{b,\zeta_b,\xi_{b-1},p}$ does not hold, since the assumption of \Cref{lem:martingale} always holds for $c=0$, we apply \Cref{lem:martingale} with $c = 0$.
Then we obtain a lower bound on the above term as
\begin{align*}
&\E_{b \sim [B]} \left[ \E_{\zeta,\xi} \left[ \E_{\theta''_i \sim \Theta''_i} \left[ \sum_{t \in \calT_b} \E_{\alpha \sim \pi_i^t(\psi'_{\zeta,\xi}(\theta''_i))} \left[ u_i^t(\theta''_i, \alpha) \right] \right] \right] \right]\\
&\ge \Pr_{\substack{ b \sim [B],\zeta_b,\xi_{b-1},p \sim \{0,1\}^{b-1}}} \left( \calF_{b,\zeta_b,\xi_{b-1},p} \right) \left( \frac{L}{2} + 2^{-24} \sqrt{L} \right)
+
\Pr_{\substack{ b \sim [B],\zeta_b,\xi_{b-1}, p \sim \{0,1\}^{b-1}}} \left( \calF^c_{b,\zeta_b,\xi_{b-1},p} \right) \frac{L}{2}\\
&= \frac{L}{2} + \Pr_{\substack{ b \sim [B], \zeta_b,\xi_{b-1}, p \sim \{0,1\}^{b-1}}} \left( \calF_{b,\zeta_b,\xi_{b-1},p} \right) \cdot 2^{-24} \sqrt{L}\\
&\ge \frac{L}{2} + 2^{-27} \sqrt{L}.
\end{align*}
Finally, substituting this into \eqref{eq:usr-lb}, we obtain
\begin{align*}
\E_{\zeta,\xi} \left[ \RUSi^T \right]
&\ge \frac{B}{2} \cdot 2^{-27} \sqrt{L}\\
&= 2^{-28} \cdot B \cdot \sqrt{T/B} \tag{since $L = T/B$}\\
&= 2^{-28} \sqrt{TB}\\
&\ge 2^{-28} \sqrt{T \log_2 |\Theta_i|},
\end{align*}
which contradicts the assumption 
$\E_{\zeta,\xi} \left[ \RUSi^T \right] < 2^{-28} \sqrt{T \log_2 |\Theta_i|}$.
Then there exist some $\zeta$ and $\xi$ such that for the problem instance with parameter $\zeta$ and $\xi$, it holds $\RUSi^T \ge 2^{-28} \sqrt{T \log_2 |\Theta_i|}$.
\end{proof}

\section{Proofs for smoothness and price of anarchy}\label{sec:poa-app}

First, we present characterizations of the intersection of communication equilibria and ANFCEs, which will be used for analyzing the PoA.
Note that $\PiCom^\epsilon \subseteq \Delta(A)^\Theta$ is the set of $\epsilon$-approximate communication equilibria and $\PiANF^\epsilon \subseteq \Delta(A)^\Theta$ is the set of $\epsilon$-ANFCEs mapped to $\Delta(A)^\Theta$.
See \Cref{sec:sr} for the formal definition of strategy representability and \Cref{sec:anf} for $\PiANF^\epsilon$.

\begin{proposition}
\label{prop:com-sr}
For any type-wise distribution $\pi \in \Delta(A)^\Theta$, the following are equivalent:
\begin{itemize}
\item[(i)] $\pi \in \PiCom^\epsilon$ and $\pi$ is strategy-representable.
\item[(ii)] $\pi \in \PiCom^\epsilon \cap \PiANF^\epsilon$.
\item[(iii)] There exists some $\sigma \in \Delta(S)$ such that $\eta(\sigma) = \pi$ and
\begin{equation}\label{eq:ic-comanf}
\tag{$\mathrm{IC}_{\mathrm{ComSR}}$}
\E_{\theta \sim \rho} \left[ \E_{s \sim \sigma} \left[ v_i(\theta_i; s(\theta)) \right] \right]
\ge                                        
\E_{\theta \sim \rho} \left[ \E_{s \sim \sigma} \left[ v_i(\theta_i; \phi(\theta_i, s_i(\psi(\theta_i))), s_{-i}(\theta_{-i})) \right] \right] - \epsilon.
\end{equation}
holds for any $i \in N$, $\psi \colon \Theta_i \to \Theta_i$, and $\phi \colon \Theta_i \times A_i \to A_i$.
\end{itemize}
\end{proposition}

\begin{proof}
First, we prove (ii) assuming (i).
Suppose $\pi \in \PiCom^\epsilon$ and $\pi$ is strategy-representable.
Since \eqref{eq:ic-anf-pi} is a weaker condition than \eqref{eq:ic-com}, $\pi$ satisfies \eqref{eq:ic-anf-pi} for any $\phi \colon \Theta_i \times A_i \to A_i$.
We thus obtain $\pi \in \PiANF^\epsilon$ from \Cref{prop:anf-pi}.

Next, we prove (iii) assuming (ii).
Suppose $\pi \in \PiCom^\epsilon \cap \PiANF^\epsilon$.
Since $\PiANF^\epsilon = \eta(\SigmaANF^\epsilon)$ from the definition, $\pi \in \PiANF^\epsilon$ implies that there exists $\sigma \in \SigmaANF^\epsilon$ such that $\eta(\sigma) = \pi$.
Moreover, since $\pi \in \PiCom^\epsilon$, \eqref{eq:ic-com} holds for any $i \in N$, $\psi \colon \Theta_i \to \Theta_i$, and $\phi \colon \Theta_i \times A_i \to A_i$.
Since $\eta$ is defined by $\pi(\theta;a) = \Pr_{s \sim \sigma} (s(\theta) = a)$ for each $\theta \in \Theta$ and $a \in A$, the left-hand side of \eqref{eq:ic-com} is equal to the left-hand side of \eqref{eq:ic-comanf}.
Then it is sufficient to show the equality of their right-hand sides.
Again from the definition of $\eta$, for any $\theta \in \Theta$, we have
\begin{align*}
\pi \left( \psi(\theta_i),\theta_{-i}; a\right)
=
\Pr_{s \sim \sigma} \left( s_i(\psi(\theta_i)) = a_i, ~ s_{-i}(\theta_{-i}) = a_{-i} \right).
\end{align*}
Hence, for each fixed $\theta \in \Theta$, the distribution of $a$ in the right-hand side of \eqref{eq:ic-com} equals the distribution of $(s_i(\psi(\theta_i)), s_{-i}(\theta_{-i}))$ in the right-hand side of \eqref{eq:ic-comanf}, and then the right-hand sides are equal.

Finally, we prove (i) assuming (iii).
The existence of $\sigma \in \Delta(S)$ such that $\eta(\sigma) = \pi$ implies the strategy representability of $\pi$.
It is sufficient to prove that \eqref{eq:ic-com} holds for any $i \in N$, $\psi \colon \Theta_i \to \Theta_i$, and $\phi \colon \Theta_i \times A_i \to A_i$.
Since \eqref{eq:ic-comanf} and \eqref{eq:ic-com} are equivalent as proved above, we obtain $\pi \in \PiCom^\epsilon$.
\end{proof}

We prepare a lemma that is commonly used in both cases.
The proof is based on the ones provided by \citet*{Roughgarden15incomplete,Syrgkanis12,ST13}, but we here extend them from Bayes--Nash equilibria to the intersection of communication equilibria and ANFCEs.
We use the characterization of this class given in \Cref{prop:com-sr}.

\begin{lemma}\label{lem:poa}
Suppose that a Bayesian game satisfies Assumptions 1 and 2.
Fix any $i \in N$ and $\pi \in \PiCom^0 \cap \PiANF^0$.
Given any $a_{i,\theta,a_i}^* \in A_i$ for each $\theta \in \Theta$ and $a_i \in A_i$, it holds that
\begin{equation*}
\E_{\theta \sim \rho} \left[ \E_{a \sim \pi(\theta)} \left[ v_i(\theta_i; a) \right] \right]
\ge
\E_{\theta \sim \rho} \left[ \E_{\theta' \sim \rho} \left[ \E_{a \sim \pi(\theta')} \left[ v_i(\theta_i; a^*_{i,\theta,a_i}, a_{-i}) \right] \right] \right].
\end{equation*}
\end{lemma}

\begin{proof}
Let $\pi \in \PiCom^0 \cap \PiANF^0$.
From \Cref{prop:com-sr}, there exists $\sigma \in \Delta(S)$ that satisfies $\eta(\sigma) = \pi$ and \eqref{eq:ic-comanf} for each $i \in N$, $\psi \colon \Theta_i \to \Theta_i$, and $\phi \colon \Theta_i \times A_i \to A_i$.
Recall that $\eta \colon \Delta(S) \to \Delta(A)^\Theta$ is defined as $(\eta(\sigma))(\theta;a) = \Pr_{s \sim \sigma} (s(\theta) = a)$ for each $\theta \in \Theta$ and $a \in A$ (see \Cref{sec:sr}).
Since $\eta(\sigma) = \pi$, the expected payoff for player $i$ in this equilibrium is
\begin{align}
\E_{\theta \sim \rho} \left[ \E_{a \sim \pi(\theta)} \left[ v_i(\theta_i; a) \right] \right]
&=
\E_{\theta \sim \rho} \left[ \E_{s \sim \sigma} \left[ v_i(\theta_i; s(\theta)) \right] \right]. \label{eq:payoff-sum}
\end{align}
Next, we apply \eqref{eq:ic-comanf} for each $i \in N$.
Fix any $\theta' \in \Theta$.
If we set $\psi(\theta_i) = \theta'_i$ for each $\theta_i \in \Theta_i$ and $\phi(\theta_i,a_i) = a^*_{i,(\theta_i,\theta'_{-i}),a_i}$ for each $\theta_i \in \Theta_i$ and $a_i \in A_i$, then \eqref{eq:ic-comanf} implies
\begin{equation*}
\E_{\theta \sim \rho} \left[ \E_{s \sim \sigma} \left[ v_i(\theta_i; s(\theta)) \right] \right]
\ge
\E_{\theta \sim \rho} \left[ \E_{s \sim \sigma} \left[ v_i(\theta_i; a^*_{i,(\theta_i,\theta'_{-i}),s_i(\theta'_i)}, s_{-i}(\theta_{-i})) \right] \right].
\end{equation*}
By taking the expectation over $\theta' \sim \rho$, we obtain
\begin{equation*}
\E_{\theta \sim \rho} \left[ \E_{s \sim \sigma} \left[ v_i(\theta_i; s(\theta)) \right] \right]
\ge
\E_{\theta' \sim \rho} \left[ \E_{\theta \sim \rho} \left[ \E_{s \sim \sigma} \left[ v_i(\theta_i; a^*_{i,(\theta_i,\theta'_{-i}),s_i(\theta'_i)}, s_{-i}(\theta_{-i})) \right] \right] \right].
\end{equation*}
Since $\rho$ is a product distribution, we can swap $\theta_{-i}$ and $\theta'_{-i}$ on the right-hand side and obtain
\begin{equation*}
\E_{\theta \sim \rho} \left[ \E_{s \sim \sigma} \left[ v_i(\theta_i; s(\theta)) \right] \right]
\ge
\E_{\theta \sim \rho} \left[ \E_{\theta' \sim \rho} \left[ \E_{s \sim \sigma} \left[ v_i(\theta_i; a^*_{i,\theta,s_i(\theta'_i)}, s_{-i}(\theta'_{-i})) \right] \right] \right].
\end{equation*}
By plugging this inequality into \eqref{eq:payoff-sum}, we obtain
\begin{align*}
\E_{\theta \sim \rho} \left[ \E_{a \sim \pi(\theta)} \left[ v_i(\theta_i; a) \right] \right] \nonumber
&\ge
\E_{\theta \sim \rho} \left[ \E_{\theta' \sim \rho} \left[ \E_{s \sim \sigma} \left[ v_i(\theta_i; a^*_{i,\theta,s_i(\theta'_i)}, s_{-i}(\theta'_{-i})) \right] \right] \right]\\
&=
\E_{\theta \sim \rho} \left[ \E_{\theta' \sim \rho} \left[ \E_{a \sim \pi(\theta')} \left[ v_i(\theta_i; a^*_{i,\theta,a_i}, a_{-i}) \right] \right] \right],
\end{align*}
where the equality holds since the distribution of $s(\theta')$ equals that of $a \sim \pi(\theta')$ due to $\pi = \eta(\sigma)$.
\end{proof}

\subsection{Proof of PoA bounds for the sum of payoffs}\label{sec:poa-sum-proof}

\thmpoasum*

\begin{proof}
Let $\pi \in \PiCom^0 \cap \PiANF^0$.
By taking the summation of \Cref{lem:poa} over $i \in N$, we obtain
\begin{align}
\E_{\theta \sim \rho} \left[ \E_{a \sim \pi(\theta)} \left[ \vSW(\theta; a) \right] \right]
&=
\sum_{i \in N} \E_{\theta \sim \rho} \left[ \E_{a \sim \pi(\theta)} \left[ v_i(\theta_i; a) \right] \right] \tag{due to the assumption} \nonumber \\
&\ge
\sum_{i \in N} \E_{\theta \sim \rho} \left[ \E_{\theta' \sim \rho} \left[ \E_{a \sim \pi(\theta')} \left[ v_i(\theta_i; a^*_{i,\theta,a_i}, a_{-i}) \right] \right] \right] \nonumber \\
&\ge
\E_{\theta \sim \rho} \left[ \E_{\theta' \sim \rho} \left[ \E_{a \sim \pi(\theta')} \left[ \lambda \max_{a' \in A} \vSW(\theta;a') - \mu \vSW(\theta;a) \right] \right] \right] \tag{due to conditional smoothness} \nonumber \\
&=
\lambda \E_{\theta \sim \rho} \left[ \max_{a' \in A} \vSW(\theta;a') \right]
- \mu \E_{\theta \sim \rho} \left[ \E_{\theta' \sim \rho} \left[ \E_{a \sim \pi(\theta')} \left[ \vSW(\theta;a) \right] \right] \right]. \label{eq:sum-payoff-bound}
\end{align}
As in the proof of \Cref{lem:poa}, we use $\sigma \in \Delta(S)$ defined in \Cref{prop:com-sr}.
By using the assumption of $\vSW$, we obtain
\begin{align}
\E_{\theta \sim \rho} \left[ \E_{\theta' \sim \rho} \left[ \E_{a \sim \pi(\theta')} \left[ \vSW(\theta;a) \right] \right] \right] \nonumber
&=
\E_{\theta \sim \rho} \left[ \E_{\theta' \sim \rho} \left[ \E_{s \sim \sigma} \left[ \vSW(\theta;s(\theta')) \right] \right] \right] \nonumber\\
&=
\E_{\theta \sim \rho} \left[ \E_{\theta' \sim \rho} \left[ \E_{s \sim \sigma} \left[ \sum_{i \in N} v_i(\theta_i;s(\theta')) \right] \right] \right] \nonumber\\
&=
\sum_{i \in N} \E_{\theta_i \sim \rho_i} \left[ \E_{\theta' \sim \rho} \left[ \E_{s \sim \sigma} \left[ v_i(\theta_i;s(\theta')) \right] \right] \right] \nonumber\\
&=
\sum_{i \in N} \E_{\theta'_i \sim \rho_i} \left[ \E_{\theta \sim \rho} \left[ \E_{s \sim \sigma} \left[ v_i(\theta_i;s_i(\theta'_i),s_{-i}(\theta_{-i})) \right] \right] \right], \label{eq:sum-payoff-equal}
\end{align}
where in the last equality, $\theta'_{-i}$ is replaced with $\theta_{-i}$ since $\rho$ is a product distribution.
Next, we apply \eqref{eq:ic-comanf} for each $i \in N$.
Fix any $\theta'_i \in \Theta_i$.
If we set $\psi(\theta_i) = \theta'_i$ for any $\theta_i \in \Theta_i$ and $\phi(\theta_i,a_i) = a_i$ for any $\theta_i \in \Theta_i$ and $a_i \in A_i$, then \eqref{eq:ic-comanf} implies
\begin{equation*}
\E_{\theta \sim \rho} \left[ \E_{s \sim \sigma} \left[ v_i(\theta_i; s(\theta)) \right] \right]
\ge
\E_{\theta \sim \rho} \left[ \E_{s \sim \sigma} \left[ v_i(\theta_i; s_i(\theta'_i), s_{-i}(\theta_{-i})) \right] \right].
\end{equation*}
By taking the expectation over $\theta'_i \sim \rho_i$, we obtain
\begin{equation*}
\E_{\theta \sim \rho} \left[ \E_{s \sim \sigma} \left[ v_i(\theta_i; s(\theta)) \right] \right]
\ge
\E_{\theta'_i \sim \rho_i} \left[ \E_{\theta \sim \rho} \left[ \E_{s \sim \sigma} \left[ v_i(\theta_i; s_i(\theta'_i), s_{-i}(\theta_{-i})) \right] \right] \right].
\end{equation*}
By plugging this inequality into \eqref{eq:sum-payoff-equal}, we obtain
\begin{align*}
\E_{\theta \sim \rho} \left[ \E_{\theta' \sim \rho} \left[ \E_{a \sim \pi(\theta')} \left[ \vSW(\theta;a) \right] \right] \right]
\le
\sum_{i \in N} \E_{\theta \sim \rho} \left[ \E_{s \sim \sigma} \left[ v_i(\theta_i; s(\theta)) \right] \right]
=
\E_{\theta \sim \rho} \left[ \E_{a \sim \pi(\theta)} \left[ \vSW(\theta; a) \right] \right].
\end{align*}
From \eqref{eq:sum-payoff-bound}, we obtain
\begin{equation*}
\E_{\theta \sim \rho} \left[ \E_{a \sim \pi(\theta)} \left[ \vSW(\theta; a) \right] \right] \nonumber
\ge
\lambda \E_{\theta \sim \rho} \left[ \max_{a' \in A} \vSW(\theta;a') \right]
- \mu \E_{\theta \sim \rho} \left[ \E_{a \sim \pi(\theta)} \left[ \vSW(\theta; a) \right] \right],
\end{equation*}
and then rearranging terms leads to a desired lower bound on the price of anarchy.
\end{proof}

\begin{remark}
In the proof of \Cref{thm:poa-sum} (and also in the proof of \Cref{lem:poa}), we use \eqref{eq:ic-comanf} only for constant map $\psi$.
We can prove the same PoA bound for the equilibrium concept with $\psi$ restricted to constant maps, which is broader than $\PiCom^0 \cap \PiANF^0$.
However, since each player $i \in N$ knows their type $\theta_i$ by definition, the equilibrium concept in which each player ignores their type is artificial.
\end{remark}

\subsection{Proof of PoA bounds for conditionally smooth mechanisms}\label{sec:poa-mechanism-proof}

\thmpoamechanism*

\begin{proof}
Let $\pi \in \PiCom^0 \cap \PiANF^0$.
By summing \Cref{lem:poa} for each $i \in N$, we obtain
\begin{align*}
\E_{\theta \sim \rho} \left[ \E_{a \sim \pi(\theta)} \left[ \sum_{i \in N} v_i(\theta_i; a) \right] \right]
&\ge
\E_{\theta \sim \rho} \left[ \E_{\theta' \sim \rho} \left[ \E_{a \sim \pi(\theta')} \left[ \sum_{i \in N} v_i(\theta_i; a^*_{i,\theta,a_i}, a_{-i}) \right] \right] \right]\\
&\ge
\E_{\theta \sim \rho} \left[ \E_{\theta' \sim \rho} \left[ \E_{a \sim \pi(\theta')} \left[ \lambda \max_{x \in X} \sum_{i \in N} v_i^+(\theta_i; x_i) - \mu \sum_{i \in N} v_i^-(a) \right] \right] \right] \tag{due to conditional smoothness}\\
&=
\lambda \E_{\theta \sim \rho} \left[ \max_{x \in X} \sum_{i \in N} v_i^+(\theta_i; x_i) \right]
- \mu \E_{\theta \sim \rho} \left[ \E_{a \sim \pi(\theta)} \left[ \sum_{i \in N} v_i^-(a) \right] \right]\\
&=
\lambda \E_{\theta \sim \rho} \left[ \max_{a \in A} \vSW(\theta;a) \right]
- \mu \E_{\theta \sim \rho} \left[ \E_{a \sim \pi(\theta)} \left[ \sum_{i \in N} v_i^-(a) \right] \right].
\end{align*}
Since $v_i(\theta_i;a) = v_i^+(\theta_i;f_i(a)) - v_i^-(a)$ for each $i \in N$, we obtain
\begin{align*}
\E_{\theta \sim \rho} \left[ \E_{a \sim \pi(\theta)} \left[ \vSW (\theta; a) \right] \right]
\ge
\lambda \E_{\theta \sim \rho} \left[ \max_{a \in A} \vSW(\theta;a) \right]
+ (1-\mu) \E_{\theta \sim \rho} \left[ \E_{a \sim \pi(\theta)} \left[ \sum_{i \in N} v_i^-(a) \right] \right].
\end{align*}
If $\mu \le 1$, the second term on the right-hand side is non-negative, and therefore, the PoA is at least $\lambda$.
If $\mu > 1$, since $v_i(\theta_i;a) = v_i^+(\theta_i;f_i(a)) - v_i^-(a) \ge 0$, the second term on the right-hand side is bounded below by $(1-\mu) \bbE_{\theta \sim \rho} \left[ \bbE_{s \sim \sigma} \left[ \vSW(\theta;s(\theta)) \right] \right]$, which implies a $\lambda / \mu$ lower bound on the PoA.
\end{proof}

\section{Bayes correlated equilibria and no-regret dynamics}\label{sec:dynamics}

In this section, we review various classes of Bayes correlated equilibria surveyed by \citet{Forges93} and propose variants of no-regret dynamics converging to them.
\Cref{sec:sr} defines strategy representability and compares it with conditional independence.
\Cref{sec:sf} provides a definition of approximate SFCEs and shows that dynamics minimizing a variant of regret called strategy swap regret converge to them.
\Cref{sec:anf} reviews a definition of approximate ANFCEs and shows that dynamics minimizing a variant of regret called type-wise swap regret converge to them.
\Cref{sec:relation} shows several inclusion relations between the classes of Bayes correlated equilibria and Bayes--Nash equilibria.

\subsection{Strategy representability}\label{sec:sr}

To classify Bayes correlated equilibria in terms of distributions, we need to compare different kinds of distributions.
As we will see, SFCEs and ANFCEs are defined as distributions over strategy profiles, that is, $\sigma \in \Delta(S)$.
This corresponds to scenarios where a mediator generates $s \sim \sigma \sim \Delta(S)$ without observing type profile $\theta \sim \rho$, and then each player $i \in N$ takes action $s_i(\theta_i)$.
On the other hand, communication equilibria and Bayesian solutions are defined as distributions over action profiles for each type profile, which we call \textit{type-wise distributions}.
A type-wise distribution $\pi \in \Delta(A)^\Theta$ determines a distribution $\pi(\theta) \in \Delta(A)$ for each type $\theta \in \Theta$.
This corresponds to scenarios where a mediator generates $a \sim \pi(\theta) \in \Delta(A)$ based on type profile $\theta \in \Theta$.
We write $\pi(\theta; a)$ instead of $\pi(\theta)(a)$ for readability.

To compare these distributions, we adopt the approach of transforming $\sigma \in \Delta(S)$ into $\pi \in \Delta(A)^\Theta$ by considering the marginal distribution.
Suppose that the mediator samples $s \sim \sigma$ independently of $\theta \sim \rho$ and then recommends $s_i(\theta_i)$ to each player $i \in N$.
Then, for each $\sigma \in \Delta(S)$, we can define its corresponding distribution $\pi \in \Delta(A)^\Theta$ by
\begin{equation*}
\pi(\theta;a) = \Pr_{s \sim \sigma} (s(\theta) = a)
\end{equation*}
for each $\theta \in \Theta$ and $a \in A$.
This operation defines a function $\eta \colon \Delta(S) \to \Delta(A)^\Theta$.

We define the strategy representability as the property that $\pi \in \eta(\Delta(S))$.
This property implies that the mediator can realize $\pi$ by sending a randomized recommendation of $s \in S$ without observing the type profile $\theta \in \Theta$.

\begin{definition}[Strategy representability]\label{def:sr}
We call a type-wise distribution $\pi \in \Delta(A)^\Theta$ \textit{strategy-representable}
if there exists $\sigma \in \Delta(S)$ such that $\pi(\theta;a) = \Pr_{s \sim \sigma} (s(\theta) = a)$ for each $\theta \in \Theta$ and $a \in A$.
\end{definition}

As we show below, not all type-wise distributions $\Delta(A)^\Theta$ are strategy-representable.
That is, if we generate $s \sim \sigma$ independently of $\theta \sim \rho$, we cannot realize all distributions in $\Delta(A)^\Theta$.

\begin{figure}
\centering
\begin{tikzpicture}[line width=1pt]
\newcommand{\ptwo}{\scriptsize$1/2$}
\newcommand{\pzero}{\scriptsize$0$}

\node[anchor=east] at (-20pt,55pt){$\theta_1$};
\node[anchor=east] at (-20pt,10pt){$\theta'_1$};

\node[anchor=east] at (-10pt,65pt){$a_1$};
\node[anchor=east] at (-10pt,45pt){$a'_1$};
\node[anchor=east] at (-10pt,20pt){$a_1$};
\node[anchor=east] at (-10pt,00pt){$a'_1$};

\node[anchor=south] at (10pt,85pt){$\theta_2$};
\node[anchor=south] at (55pt,85pt){$\theta'_2$};

\node[anchor=south] at (00pt,75pt){$a_2$};
\node[anchor=south] at (20pt,75pt){$a'_2$};
\node[anchor=south] at (45pt,75pt){$a_2$};
\node[anchor=south] at (65pt,75pt){$a'_2$};

\node[draw, rectangle, minimum width=20pt, minimum height=20pt] at (00pt,45pt){\pzero};
\node[draw, rectangle, minimum width=20pt, minimum height=20pt] at (00pt,65pt){\ptwo};
\node[draw, rectangle, minimum width=20pt, minimum height=20pt] at (20pt,45pt){\ptwo};
\node[draw, rectangle, minimum width=20pt, minimum height=20pt] at (20pt,65pt){\pzero};

\node[draw, rectangle, minimum width=20pt, minimum height=20pt] at (45pt,45pt){\pzero};
\node[draw, rectangle, minimum width=20pt, minimum height=20pt] at (45pt,65pt){\ptwo};
\node[draw, rectangle, minimum width=20pt, minimum height=20pt] at (65pt,45pt){\ptwo};
\node[draw, rectangle, minimum width=20pt, minimum height=20pt] at (65pt,65pt){\pzero};

\node[draw, rectangle, minimum width=20pt, minimum height=20pt] at (00pt,00pt){\pzero};
\node[draw, rectangle, minimum width=20pt, minimum height=20pt] at (00pt,20pt){\ptwo};
\node[draw, rectangle, minimum width=20pt, minimum height=20pt] at (20pt,00pt){\ptwo};
\node[draw, rectangle, minimum width=20pt, minimum height=20pt] at (20pt,20pt){\pzero};

\node[draw, rectangle, minimum width=20pt, minimum height=20pt] at (45pt,00pt){\ptwo};
\node[draw, rectangle, minimum width=20pt, minimum height=20pt] at (45pt,20pt){\pzero};
\node[draw, rectangle, minimum width=20pt, minimum height=20pt] at (65pt,00pt){\pzero};
\node[draw, rectangle, minimum width=20pt, minimum height=20pt] at (65pt,20pt){\ptwo};
\end{tikzpicture}
\caption{An example of a type-wise distribution that is not strategy-representable.}
\label{fig:non-representable}
\end{figure}

\begin{example}\label{ex:sr}
We give an example $\pi \in \Delta(A)^\Theta$ that cannot be represented by any $\sigma \in \Delta(S)$.
Suppose that $A_1 = \{a_1,a'_1\}$, $A_2 = \{a_2,a'_2\}$, $\Theta_1 = \{\theta_1,\theta'_1\}$, and $\Theta_2 = \{\theta_2,\theta'_2\}$.
Let $\rho$ be the uniform distribution on $\Theta_1 \times \Theta_2$.
As illustrated in \Cref{fig:non-representable}, we define $\pi \in \Delta(A)^\Theta$ by
$\pi(\theta_1,\theta_2; a_1,a_2)
= \pi(\theta_1,\theta_2; a'_1,a'_2)
= \pi(\theta_1,\theta'_2; a_1,a_2)
= \pi(\theta_1,\theta'_2; a'_1,a'_2)
= \pi(\theta'_1,\theta_2; a_1,a_2)
= \pi(\theta'_1,\theta_2; a'_1,a'_2)
= \pi(\theta'_1,\theta'_2; a_1,a'_2)
= \pi(\theta'_1,\theta'_2; a'_1,a_2)
= \frac{1}{2}$.
We show that $\pi$ is not strategy-representable.
If there exists $\sigma \in \Delta(S)$ that satisfies $\pi(\theta;a) = \Pr_{s \sim \sigma}(s(\theta) = a)$ for any $\theta \in \Theta$ and $a \in A$,
then $\Pr(s_1(\theta_1) = a_1, ~ s_2(\theta_2) = a_2) = \Pr(s_1(\theta_1) = a'_1, ~ s_2(\theta_2) = a'_2) = \frac{1}{2}$.
This implies that the event $s_1(\theta_1) = a_1$ coincides with the event $s_2(\theta_2) = a_2$, and $s_1(\theta_1) = a'_1$ coincides with $s_2(\theta_2) = a'_2$.
Similarly, from $\Pr(s_1(\theta_1) = a_1, ~ s_2(\theta'_2) = a_2) = \Pr(s_1(\theta_1) = a'_1, ~ s_2(\theta'_2) = a'_2) = \frac{1}{2}$, we can see that $s_1(\theta_1) = a_1$ coincides with $s_2(\theta'_2) = a_2$, and $s_1(\theta_1) = a'_1$ coincides with $s_2(\theta'_2) = a'_2$.
Therefore, $s_2(\theta_2) = s_2(\theta'_2)$ holds with probability $1$.
By symmetry, we can see that $s_1(\theta_1) = s_1(\theta'_1)$ holds with probability $1$.
Hence, the distributions of $(s_1(\theta_1), s_2(\theta_2))$ and $(s_1(\theta'_1), s_2(\theta'_2))$ must be identical, but they are not, which leads to a contradiction.
This implies that $\pi$ is not strategy-representable.
\end{example}

We say $\pi \in \Delta(A)^\Theta$ is a \textit{type-wise product distribution} if there exists some $\pi_{i,\theta_i} \in \Delta(A_i)$ for each $i \in N$ and $\theta_i \in \Theta_i$ such that $\pi(\theta;a) = \prod_{i \in N} \pi_{i,\theta_i}(a_i)$ for every $\theta \in \Theta$ and $a \in A$.
Let $\Pi_{\mathrm{prod}} \subseteq \Delta(A)^\Theta$ be the set of all type-wise product distributions.
As we show below, a finite mixture of type-wise product distributions is always strategy-representable.
We will use this fact to prove that the average of finite rounds of repeated play in the agent normal form is strategy-representable since the distribution in each round is type-wise product.
Furthermore, the inverse also holds.

\begin{proposition}\label{prop:sr}
A distribution $\pi \in \Delta(A)^\Theta$ is strategy-representable
if and only if
$\pi$ is a finite mixture of type-wise product distributions, that is, there exist a positive integer $T$, $\sigma' \in \Delta([T])$, and $\pi^t \in \PiProd$ for each $t \in [T]$ such that $\pi(\theta;a) = \sum_{t \in [T]} \sigma'(t) \pi^t(\theta;a)$ for each $\theta \in \Theta$ and $a \in A$.
\end{proposition}

\begin{proof}
Assume $\pi$ is strategy-representable, that is, there exists $\sigma \in \Delta(S)$ such that $\pi(\theta;a) = \Pr_{s \sim \sigma} (s(\theta) = a)$ for each $\theta \in \Theta$ and $a \in A$.
Let $\sigma' = \sigma$ and define $\pi^{s} \in \Delta(A)^\Theta$ by $\pi^{s}(\theta;a) = \bfone_{\{s(\theta) = a\}}$.
We can check $\pi^{s} \in \PiProd$ by setting $\pi_{i,\theta_i}(a_i) = \bfone_{\{s_i(\theta_i) = a_i\}}$.
Then
\begin{equation*}
\pi(\theta;a)
= \Pr_{s \sim \sigma} (s(\theta) = a)
= \sum_{s \in S} \sigma(s) \bfone_{\{s(\theta) = a\}}
= \sum_{s \in S} \sigma'(s) \pi^{s}(\theta;a)
\end{equation*}
for each $\theta \in \Theta$ and $a \in A$.

Assume there exists a positive integer $T$, $\sigma' \in \Delta(T)$, and $\pi^t \in \PiProd$ for each $t \in [T]$ such that $\pi(\theta;a) = \sum_{t \in [T]} \sigma'(t) \pi^t(\theta;a)$.
Since $\pi^t \in \PiProd$ for each $t \in [T]$, there exists $\pi^t_{i,\theta_i} \in \Delta(A_i)$ for each $i \in N$ and $\theta_i \in \Theta_i$ such that $\pi^t(\theta;a) = \prod_{i \in N} \pi^t_{i,\theta_i}(a_i)$ for each $\theta \in \Theta$ and $a \in A$.
We define $\sigma^t \in \Delta(S)$ by $\sigma^t(s) = \prod_{i \in N} \prod_{\theta_i \in \Theta_i} \pi^t_{i,\theta_i}(s_i(\theta_i))$ for each $s \in S$.
Then for each $\theta \in \Theta$ and $a \in A$, it holds that
\begin{equation*}
\Pr_{s \sim \sigma^t} (s(\theta) = a)
= \sum_{\substack{s \in S \colon \\ s(\theta) = a}} \prod_{i \in N} \prod_{\theta'_i \in \Theta_i} \pi^t_{i,\theta'_i}(s_i(\theta'_i))
= \prod_{i \in N} \pi^t_{i,\theta_i}(a_i)
= \pi^t(\theta;a).
\end{equation*}
Define $\sigma \in \Delta(S)$ by $\sigma(s) = \sum_{t \in [T]} \sigma'(t) \sigma^t(s)$ for every $s \in S$.
Then
\begin{equation*}
\Pr_{s \sim \sigma} (s(\theta) = a)
= \sum_{t \in [T]} \sigma'(t) \Pr_{s \sim \sigma^t} (s(\theta) = a)
= \sum_{t \in [T]} \sigma'(t) \pi^t(\theta;a) 
= \pi(\theta;a)
\end{equation*}
for each $\theta \in \Theta$ and $a \in A$, which implies that $\pi$ is strategy-representable.
\end{proof}

\citet*{Forges93} did not formally define the strategy representability but instead claimed that a property called \textit{conditional independence property} is important to classify Bayes correlated equilibria.
As \citet*{Forges93} claimed, the conditional independence property holds for any strategy-representable distribution.
Although it was not stated explicitly, the conditional independence property is not equivalent to the strategy representability.%
\footnote{A related discussion about \textit{belief invariant Bayesian solutions} can be found in existing studies \citep*{LRS10,Forges06}.}
As claimed in \Cref{prop:com-sr}, the strategy representability exactly characterize the difference between communication equilibria and ANFCEs.

Here, we show the difference between the strategy representability and the conditional independence property.
The conditional independence property was mentioned by \citet*{Forges93} as an important property that holds for SFCEs and ANFCEs.

\begin{definition}[Conditional independence property]
Let $\pi \in \Delta(A)^\Theta$.
We consider a joint distribution in $\Delta(\Theta \times A)$ that first generates $\theta \in \Theta$ according to $\rho$ and then generates $a \in A$ according to $\pi(\theta)$.
We say $\pi$ satisfies the \textit{conditional independence property}
\footnote{It is originally defined as a property of a joint distribution in $\Delta(\Theta \times A)$, but we can consider a joint distribution in $\Delta(\Theta \times A)$ and a type-wise distribution in $\Delta(A)^\Theta$ interchangeably.
For each $\pi \in \Delta(A)^\Theta$, the corresponding joint distribution in $\Delta(\Theta \times A)$ is uniquely determined.
Conversely, for any joint distribution whose marginal for $\theta$ coincides with $\rho$, we can uniquely determine the corresponding distribution in $\Delta(A)^\Theta$.
}
if $a_i$ is conditionally independent of $\theta_{-i}$ given $\theta_i$ for each $i \in N$.
\end{definition}

As \citet*{Forges93} claimed, this property holds for any strategy-representable distribution.
Although it was not claimed explicitly, this property is not equivalent to the strategy representability.

\begin{proposition}
If $\pi \in \Delta(A)^\Theta$ is strategy-representable,
then $\pi$ satisfies the conditional independence property.
On the other hand, there exists $\pi \in \Delta(A)^\Theta$ that satisfies the conditional independence property but not the strategy representability.
\end{proposition}

\begin{proof}
First, we show that the strategy representability implies the conditional independence property.
Assume there exists $\sigma \in \Delta(S)$ such that $\pi(\theta;a) = \Pr_{s \sim \sigma} \left( s(\theta) = a \right)$ for each $\theta \in \Theta$ and $a \in A$.
To prove the conditional independence property, it is sufficient to prove the conditional independence of $a_i$ and $\theta_{-i}$ given $\theta_i$ for any $i \in N$, $\theta \in \Theta$, and $a_i \in A_i$.
Using the strategy representability of $\pi$, we obtain
\begin{align*}
&\Pr_{\substack{\theta' \sim \rho\\a' \sim \pi(\theta')}}(\theta' = \theta, ~ a'_i = a_i)
\Pr_{\theta' \sim \rho}(\theta'_i = \theta_i)\\
&=
\left( \rho(\theta)
\Pr_{a' \sim \pi(\theta)}(a'_i = a_i) \right)
\Pr_{\theta' \sim \rho}(\theta'_i = \theta_i)\\
&=
\rho(\theta)
\Pr_{s' \sim \sigma}(s'_i(\theta_i) = a_i)
\Pr_{\theta' \sim \rho}(\theta'_i = \theta_i)
\tag{from the definition of $\sigma$}
\\
&=
\rho(\theta)
\Pr_{\substack{\theta' \sim \rho \\ s' \sim \sigma}}(s'_i(\theta_i) = a_i , ~ \theta'_i = \theta_i)\\
&=
\rho(\theta)
\Pr_{\substack{\theta' \sim \rho \\ s' \sim \sigma}}(s'_i(\theta'_i) = a_i , ~ \theta'_i = \theta_i)
\tag{since $s'_i(\theta'_i)$ can be replaced by $s'_i(\theta_i)$ when $\theta'_i = \theta_i$}\\
&=
\rho(\theta)
\Pr_{\substack{\theta' \sim \rho \\ a' \sim \pi(\theta')}}(a'_i = a_i, ~ \theta'_i = \theta_i).
\tag{from the definition of $\sigma$}
\end{align*}
Dividing both sides by $(\Pr_{\theta' \sim \rho}(\theta'_i = \theta_i))^2$, we obtain the conditional independence of $a_i$ and $\theta_{-i}$ given $\theta_i$:
\begin{equation*}
\Pr_{\substack{\theta' \sim \rho\\a' \sim \pi(\theta')}}(\theta'_{-i} = \theta_{-i}, ~ a'_i = a_i \mid \theta'_i = \theta_i)
=
\Pr_{\substack{\theta' \sim \rho\\a' \sim \pi(\theta')}}(\theta'_{-i} = \theta_{-i} \mid \theta'_i = \theta_i)
\Pr_{\substack{\theta' \sim \rho\\a' \sim \pi(\theta')}}(a_i = a_i \mid \theta'_i = \theta_i).
\end{equation*}

Next, we show the existence of an example that satisfies the conditional independence property but not the strategy representability.
We consider $\pi \in \Delta(A)^\Theta$ described in \Cref{ex:sr} again (see \Cref{fig:non-representable}).
Since we already proved that this $\pi$ is not strategy-representable, it is sufficient to show that $\pi$ satisfies the conditional independence property.
It can be verified by considering the distribution conditional on each player's type.
If the first player's type is $\theta_1$, the conditional distribution of the first player's action and the second player's type is the uniform distribution on $A_1 \times \Theta_2$, hence conditionally independent.
This is also the case for $\theta'_1$.
Since the definition of $\pi$ is symmetric for two players, we can conclude that $\pi$ satisfies the conditional independence property.

\end{proof}

\subsection{Strategic-form correlated equilibria}\label{sec:sf}

\textit{Strategic-form correlated equilibria} (SFCEs; also known as normal-form correlated equilibria) are defined as correlated equilibria of the strategic form of a Bayesian game.
Recall that in the strategic form, the set of strategies $S_i = {A_i}^{\Theta_i}$ is considered as a decision space for each player $i \in N$.
This class is defined as a set of joint distributions $\sigma \in \Delta(S)$ on strategy profiles in which each player $i \in N$ cannot gain by deviating from recommended strategy $s_i \in S_i$ in expectation.
This deviation is represented by a swap $\phiSF \colon S_i \to S_i$ that maps recommended strategy $s_i \in S_i$ to another strategy $\phiSF(s_i)$.
Here, we define an approximate version in which each player can gain at most $\epsilon$ in expectation by deviating from the recommendation.

\begin{definition}[$\epsilon$-Approximate strategic-form correlated equilibria ($\epsilon$-SFCEs)]
For any $\epsilon \ge 0$, a distribution $\sigma \in \Delta(S)$ is an $\epsilon$-approximate strategic-form correlated equilibrium
if
for any $i \in N$ and any $\phiSF \colon S_i \to S_i$, it holds that
\begin{equation}\label{eq:ic-sf}
\tag{$\mathrm{IC}_{\mathrm{SF}}$}
\E_{\theta \sim \rho} \left[ \E_{s \sim \sigma} \left[ v_i(\theta; s(\theta)) \right] \right]
\ge                              
\E_{\theta \sim \rho} \left[ \E_{s \sim \sigma} \left[ v_i(\theta; (\phiSF(s_i))(\theta_i),s_{-i}(\theta_{-i})) \right] \right] - \epsilon.
\end{equation}
Let $\SigmaSF^\epsilon \subseteq \Delta(S)$ be the set of all $\epsilon$-SFCEs.
\end{definition}

The notion of SFCE is a natural equilibrium concept for the following scenario.
A mediator samples $s \in S$ according to a distribution $\sigma \in \Delta(S)$ and then informs each player $i \in N$ of $s_i$.
Independently, a type profile $\theta$ is sampled from $\rho$, and each player $i \in N$ observes their own type $\theta_i$.
Then, each player $i \in N$ decides their action $s_i(\theta_i) \in A_i$.
If every player cannot gain by deviating from action $s_i(\theta_i)$ in expectation, the distribution $\sigma$ is an SFCE.

The role of this mediator can be interpreted as a correlation device.
A correlation device privately sends (possibly correlated) signals to each player according to some distribution.
An SFCE can be defined as a Nash equilibrium of an extended game equipped with a correlation device.
A correlation device can send any signal in general, but from the revelation principle, we can assume that each signal corresponds to a recommendation of each strategy without loss of generality.

Since SFCEs are correlated equilibria of the strategic form of Bayesian games, no-swap-regret dynamics converging to correlated equilibria can be directly extended to dynamics converging to SFCEs.
In contrast to the problem setting introduced in \Cref{sec:concept}, the decision space for each player $i \in N$ in the strategic form is $S_i$.
In each round $t \in [T]$, each player $i \in N$ decides a randomized strategy $\sigma_i^t \in \Delta(S_i)$ and then obtains the expected payoff $\E_{\theta \sim \rho} \left[ \E_{s \sim \sigma^t} \left[ v_i(\theta;s(\theta)) \right] \right]$, where $\sigma^t$ is the product distribution that independently generates $s_j \sim \sigma_j^t$ for each $j \in N$.
We define the reward vector $u_i^t \in [0,1]^{\Theta_i \times A_i}$ by
\begin{equation}\label{eq:reward}
u_i^t(\theta_i, a_i) = \E_{\theta_{-i} \sim \rho|\theta_i} \left[ \E_{s_{-i} \sim \sigma^t_{-i}} \left[ v_i( \theta; a_i, s_{-i}(\theta_{-i})) \right] \right]
\end{equation}
for each type $\theta_i \in \Theta_i$ and action $a_i \in A_i$, where $\sigma_{-i}^t$ is the product distribution that independently generates $s_j \sim \sigma_j^t$ for each $j \in N \setminus \{i\}$.
Note that the decision space is $S_i$, but the rewards for them can be succinctly represented by $u_i^t \in [0,1]^{\Theta_i \times A_i}$.

Then each player $i \in N$ decides a strategy distribution $\sigma_i^t \in \Delta(S_i)$ in each round $t \in [T]$ and obtains the expected payoff $\E_{\theta_i \sim \rho_i} \left[ \E_{s_i \sim \sigma_i^t} \left[ u_i^t(\theta_i, s_i(\theta_i)) \right] \right]$ according to reward vector $u_i^t \in [0,1]^{\Theta_i \times A_i}$ defined as \eqref{eq:reward}, which is the setting of online learning with stochastic types.
We consider an extension of swap regret to this online learning problem with decision space $S_i$, which we call \textit{strategy swap regret}.

\begin{definition}[{Strategy swap regret}]
For online learning with stochastic types specified by actions $A_i$, types $\Theta_i$, prior distribution $\rho_i$, and reward vector $u_i^t \in [0,1]^{\Theta_i \times A_i}$ for every time round $t \in [T]$, strategy swap regret is defined as
\begin{equation*}
\RSSi^T = \max_{\phiSF \colon S_i \to S_i} \sum_{t=1}^T \E_{\theta_i \sim \rho_i} \left[ \E_{s_i \sim \sigma_i^t} \left[ u_i^t(\theta_i, (\phiSF(s_i))(\theta_i)) - u_i^t(\theta_i, s_i(\theta_i)) \right] \right].
\end{equation*}
\end{definition}

\begin{algorithm}[t]
\caption{Dynamics for strategic-form correlated equilibria}\label{alg:sf-dynamics}
	For each player $i \in N$, let $\calA_{i}$ be a subroutine that minimizes strategy swap regret for online learning with stochastic types.\\
	\For{each round $t = 1,\dots,T$}{
		Each player $i \in N$ decides their randomized strategy $\sigma^t_i \in \Delta(S_i)$ according to $\calA_{i}$ and shares it with the other players.\\
		Each player $i \in N$ computes reward $u_i^t(\theta_i, a_i) = \E_{\theta_{-i} \sim \rho|\theta_i} \left[ \E_{s_{-i} \sim \sigma^t_{-i}} \left[ v_i(\theta; a_i, s_{-i}(\theta_{-i})) \right] \right]$ for every $\theta_i \in \Theta_i$ and $a_i \in A_i$, where $\sigma^t_{-i} \in \Delta(S_{-i})$ is the product distribution that independently generates $s_j \sim \sigma_j^t$ for each $j \in N \setminus \{i\}$.\\
		Feed $u_i^t \in [0,1]^{\Theta_i \times A_i}$ to $\calA_i$ as a reward vector for round $t$.
	}
\end{algorithm}

Theorem 3 of \citet*{BM07} relates each player's swap regret with the convergence to a correlated equilibrium.
By directly applying it to the strategic form of Bayesian games, we can show that if the strategy swap regret of each player $i \in N$ grows sublinear in $T$, we can guarantee the convergence of the dynamics to an SFCE.
This is claimed in the following proposition formally.

\begin{proposition}\label{prop:dynamics-sf}
Let $\sigma_i^t \in \Delta(S_i)$ be the randomized strategy of each player $i \in N$ for each round $t \in [T]$ in \Cref{alg:sf-dynamics}.
Then, the empirical distribution $\sigma \in \Delta(S)$ defined by $\sigma(s) = \frac{1}{T} \sum_{t=1}^T \prod_{i \in N} \sigma_i^t(s_i)$ for each $s \in S$
is a $\frac{\max_{i \in N} \RSSi^T}{T}$-SFCE, where $\RSSi^T$ is the strategy swap regret for each subroutine $\calA_i$.
\end{proposition}

\begin{proof}
From the definition of $\epsilon$-SFCEs, it is sufficient to prove
\begin{equation*}
\E_{\theta \sim \rho} \left[ \E_{s \sim \sigma} \left[ v_i(\theta; (\phiSF(s_i))(\theta_i),s_{-i}(\theta_{-i})) \right] \right]
-
\E_{\theta \sim \rho} \left[ \E_{s \sim \sigma} \left[ v_i(\theta; s(\theta)) \right] \right]
\le \frac{\max_{j \in N} \RSSj^T}{T}
\end{equation*}
for each $i \in N$ and $\phiSF \colon S_i \to S_i$.
Let $\sigma^t \in \Delta(S_{-i})$ be the product distribution that independently generates $s_j \sim \sigma_j^t$ for each $j \in N$ and $\sigma_{-i}^t \in \Delta(S)$ the product distribution that independently generates $s_j \sim \sigma_j^t$ for each $j \in N \setminus \{i\}$.
Then the left-hand side can be bounded as
\begin{align*}
&\E_{\theta \sim \rho} \left[ \E_{s \sim \sigma} \left[ v_i(\theta; (\phiSF(s_i))(\theta_i),s_{-i}(\theta_{-i})) \right] \right]
- 
\E_{\theta \sim \rho} \left[ \E_{s \sim \sigma} \left[ v_i(\theta; s(\theta)) \right] \right]\\
&= 
\E_{t \sim [T]} \left[ \E_{s \sim \sigma^t} \left[ \E_{\theta \sim \rho} \left[ v_i(\theta; (\phiSF(s_i))(\theta_i),s_{-i}(\theta_{-i})) - v_i(\theta; s(\theta)) \right] \right] \right] \tag{from the definition of $\sigma$} \\
&= 
\frac{1}{T} \sum_{t=1}^T \E_{s_i \sim \sigma_i^t} \left[ \E_{s_{-i} \sim \sigma_{-i}^t} \left[ \E_{\theta_i \sim \rho_i} \left[ \E_{\theta_{-i} \sim \rho|\theta_i} \left[ v_i(\theta; (\phiSF(s_i))(\theta_i), s_{-i}(\theta_{-i}))
- 
v_i(\theta; s(\theta))\right] \right] \right]\right] \tag{since $\sigma^t$ is the product distribution}\\
&= 
\frac{1}{T} \sum_{t=1}^T \E_{s_i \sim \sigma_i^t} \left[ \E_{\theta_{i} \sim \rho_{i}} \left[ u_i^t(\theta_i, (\phiSF(s_i))(\theta_i))
- 
u_i^t(\theta_i, s_i(\theta_i))\right] \right] \tag{from the definition of $u_i^t$}\\
&\le \frac{\RSSi^T}{T} \\
&\le \frac{\max_{j \in N} \RSSj^T}{T},
\end{align*}
which completes the proof.
\end{proof}

We provide a detailed description of the dynamics in \Cref{alg:sf-dynamics}.
We should note that this proposition does not imply the existence of a polynomial-time algorithm for computing an approximate SFCE.
This is because the size of the decision space $S_i$ for each player $i \in N$ is exponential.
In \Cref{sec:min-strategy-swap}, we provide an algorithm that has an $O\left(\sqrt{T |A_i|^{|\Theta_i|} \log |A_i|} \right)$ upper bound on strategy swap regret, which is slightly better than a direct application of the swap regret minimization algorithm proposed by \citet*{BM07} but still exponential in $|\Theta_i|$.
Computing an $\epsilon$-SFCE by simulating the dynamics with this algorithm requires time exponential in $|\Theta_i|$.
Whether there exists an algorithm that computes an $\epsilon$-SFCE in time polynomial in $n$, $|A_i|$, $|\Theta_i|$, and $1/\epsilon$ with an oracle for utility functions $(v_i)_{i=1}^n$ is an open problem left for future work.

\begin{remark}
We define approximate SFCEs based on swap regret, but it is also possible to define them based on internal regret.
Note that we have two different definitions of approximate correlated equilibria that are based on internal regret or swap regret.%
\footnote{\citet*{GS18} called them \textit{approximate correlated equilibria} and \textit{approximate rule correlated equilibria}, respectively.}
These definitions coincide when the error is $0$, but different when the error is non-zero.
In this paper, we adopt the extension based on swap regret, but it is also possible to use the extension based on internal regret,
in which $\phiSF \colon S_i \to S_i$ in the definition is restricted to a swap for a single strategy:
\begin{equation*}
\phiSF(s_i) = 
\begin{cases}
s''_i & \text{if $s_i = s'_i$}\\
s_i & \text{otherwise}
\end{cases}
\end{equation*}
for some $s'_i,s''_i \in S_i$.
However, the problem of computing this version of approximate SFCEs is rather easy.
It is because if we set $\sigma \in \Delta(S)$ to the uniform distribution over $S$, then the error is always at most $1/(\min_{i \in N}|S_i|)$.
Hence, for $\epsilon \ge 1/(\min_{i \in N}|S_i|)$, computing it can be done easily.
On the other hand, for $\epsilon \le 1/(\min_{i \in N}|S_i|)$, if we assume $|S_i|$ has the same order for every $i \in N$, then we can straightforwardly run an internal regret minimization algorithm with decision space $S_i$ in time polynomial in $1/\epsilon = \min_{i \in N}|S_i|$.
\end{remark}

To compare SFCEs with other classes defined on $\Delta(A)^\Theta$, we define the set of distributions obtained by mapping SFCEs to $\Delta(A)^\Theta$ by $\eta \colon \Delta(S) \to \Delta(A)^\Theta$ (see \Cref{sec:sr} for the definition of $\eta$).
Formally, it is defined as $\PiSF^\epsilon = \{ \eta(\sigma) \in \Delta(A)^\Theta \mid \sigma \in \SigmaSF^\epsilon \}$.
From the definition, every $\pi \in \PiSF^\epsilon$ is strategy-representable.

As we will see, $\PiSF^\epsilon$ is a subset of approximate communication equilibria.
Informally, computing an $\epsilon$-SFCE is harder than computing an $\epsilon$-approximate communication equilibrium.
If an $\epsilon$-SFCE $\sigma \in \Delta(S)$ can be obtained, then we can compute its corresponding distribution $\eta(\sigma) \in \Delta(A)^\Theta$, which is also an $\epsilon$-approximate communication equilibrium.
Note that this reduction is informal because distributions in $\Delta(S)$ and $\Delta(A)^\Theta$ have an exponentially large support in general, and the computational complexity of computing $\eta(\sigma)$ depends on the succinct representation of $\sigma$.

Whether $\sigma \in \Delta(S)$ is an SFCE cannot be determined by its corresponding $\eta(\sigma) \in \Delta(A)^\Theta$.
As in the following example, even if $\sigma, \sigma' \in \Delta(S)$ correspond to the same $\pi \in \Delta(A)^\Theta$, it is possible that $\sigma$ is an SFCE, but $\sigma'$ is not.

\begin{example}
Suppose $\Theta_1 = \{\theta_1,\theta'_1\}$ and $\Theta_2 = \{\theta_2\}$ with $A_1 = \{a_1,a'_1\}$ and $A_2 = \{a_2,a'_2\}$.
The prior distribution $\rho \in \Delta(\Theta)$ is the uniform distribution over $\Theta$.
Let $v_1(\tilde{\theta}; a_1,a_2) = v_1(\tilde{\theta}; a'_1,a'_2) = 1$ and $v_1(\tilde{\theta}; a'_1,a_2) = v_1(\tilde{\theta}; a_1,a'_2) = 0$ for any $\tilde{\theta} \in \Theta$,
while the second player's payoff is always $0$, i.e., $v_2 \equiv 0$.
We consider two distinct distributions $\sigma, \sigma' \in \Delta(S)$ with the same corresponding $\pi \in \Delta(A)^\Theta$.
The first one $\sigma$ is the uniform distribution over $S$.
The second one $\sigma'$ is the uniform distribution over $s^1,s^2,s^3,s^4$ defined as
\begin{align*}
&s^1_1(\theta_1) = a_1,  ~s^1_1(\theta'_1) = a_1,  ~s^1_2(\theta_2) = a_2, \\
&s^2_1(\theta_1) = a'_1, ~s^2_1(\theta'_1) = a'_1, ~s^2_2(\theta_2) = a_2, \\
&s^3_1(\theta_1) = a_1,  ~s^3_1(\theta'_1) = a'_1, ~s^3_2(\theta_2) = a'_2, \\
&s^4_1(\theta_1) = a'_1, ~s^4_1(\theta'_1) = a_1,  ~s^4_2(\theta_2) = a'_2.
\end{align*}
Both $\pi(\theta_1,\theta_2)$ and $\pi(\theta'_1,\theta_2)$ are the uniform distribution over $A_1 \times A_2$.
The expected payoff for each type of $\theta_1$ and $\theta'_1$ is $\frac{1}{2}$.

First, we show that $\sigma$ is an SFCE.
Since $v_2 \equiv 0$, the second player does not have any incentive to deviate.
Since $\sigma$ is the uniform distribution over $S$, when $s_1 \in S_1$ is recommended, the posterior distribution of $s_2$ is the uniform distribution over $S_2$.
Hence, the expected payoff of each action of $A_1$ is $\frac{1}{2}$, and the first player does not have any incentive to deviate.

Next, we show that $\sigma'$ is not an SFCE.
Since $s^1_1,s^2_1,s^3_1,s^4_1$ are all distinct, the first player can infer the recommendation $s_2(\theta_2)$ for the second player.
The first player can increase the payoff to $1$ by taking the same action as the second player.
Therefore, $\sigma'$ is not an SFCE.
\end{example}

\subsection{Agent-normal-form correlated equilibria}\label{sec:anf}

Agent-normal-form correlated equilibria (ANFCEs) are defined as correlated equilibria of the agent normal form of a Bayesian game.
In an ANFCE, the mediator recommends possibly correlated actions to all hypothetical players $(i,\theta_i) \in N'$.
Since an action profile for the agent normal form assigns some $a_i \in A_i$ to every $i \in N$ and $\theta_i \in \Theta_i$, it can be identified with a strategy profile $s \in S$.
Hence, we can define ANFCEs as the set of distributions over $S$.
The only difference from SFCEs is the information that each player can use for deviations.
While each player $i \in N$ can decide their action according to the recommended actions for all types $\Theta_i$ in SFCEs, each player $i \in N$ can use only the recommended action $s_i(\theta_i)$ for their realized type $\theta_i \in \Theta_i$ in ANFCEs.
Here, we define an approximate version of ANFCEs with additive error $\epsilon$.

\begin{definition}[$\epsilon$-Approximate agent-normal-form correlated equilibria ($\epsilon$-ANFCEs)]
For any $\epsilon \ge 0$, a distribution $\sigma \in \Delta(S)$ is an $\epsilon$-approximate agent-normal-form correlated equilibrium
if
for any $i \in N$ and $\phi \colon \Theta_i \times A_i \to A_i$, it holds that
\begin{equation}\label{eq:ic-anf}
\tag{$\mathrm{IC}_{\mathrm{ANF}}$}
\E_{\theta \sim \rho} \left[ \E_{s \sim \sigma} \left[ v_i(\theta; s(\theta)) \right] \right]
\ge                                        
\E_{\theta \sim \rho} \left[ \E_{s \sim \sigma} \left[ v_i(\theta; \phi(\theta_i, s_i(\theta_i)), s_{-i}(\theta_{-i})) \right] \right] - \epsilon.
\end{equation}
Let $\SigmaANF^\epsilon \subseteq \Delta(S)$ be the set of all $\epsilon$-ANFCEs.
\end{definition}

Here, we consider a natural extension of no-swap-regret dynamics to the agent normal form, which converges to ANFCEs.
Recall that in the agent normal form, the same player with different types is hypothetically considered as distinct players.

In contrast to SFCEs, each player $i \in N$ in ANFCEs is informed of only the action for the realized type $s_i(\theta_i)$, not the full strategy $s_i$.
In an SFCE, each player can use recommendations for unrealized types, which differentiates SFCEs from ANFCEs.
In an ANFCE, each player $i \in N$ with type $\theta_i \in \Theta_i$ can use only $s_i(\theta_i)$ to decide deviation $\phi(\theta_i,s_i(\theta_i))$.
This difference makes $\SigmaANF^\epsilon$ broader than $\SigmaSF^\epsilon$.

\begin{remark}
The incentive constraint in our definition is slightly stronger than the incentive constraint for each hypothetical player $(i,\theta'_i) \in N'$, which can be written as
\begin{equation*}
\E_{\theta \sim \rho} \left[ \bfone_{\{\theta_i=\theta'_i\}} \E_{s \sim \sigma} \left[ v_i(\theta; s(\theta)) \right] \right]
\ge                                        
\E_{\theta \sim \rho} \left[ \bfone_{\{\theta_i=\theta'_i\}} \E_{s \sim \sigma} \left[ v_i(\theta; \phi_{\theta'_i}(s_i(\theta_i)), s_{-i}(\theta_{-i})) \right] \right] - \epsilon
\end{equation*}
for any $\phi_{\theta'_i} \colon A_i \to A_i$.
This can be checked by letting $\phi(\theta'_i,\cdot) = \phi_{\theta'_i}$ and $\phi(\theta''_i,\cdot)$ be the identity map for each $\theta''_i \in \Theta_i \setminus \{\theta'_i\}$ in \eqref{eq:ic-anf}.
Note that this difference is only for the error parameter, and these two incentive constraints are equivalent for $\epsilon = 0$.
\end{remark}

Since the setting of the agent normal form is hypothetical, it is difficult to consider a natural scenario for ANFCEs.
In an ANFCE, the mediator determines a strategy profile $s \in S$ according to $\sigma \in \Delta(S)$ but recommends to each player only $s_i(\theta_i)$.
To realize it, the mediator must know the type profile $\theta \in \Theta$ in advance.
However, if it is the case, then it is possible to realize any distribution $\pi(\theta) \in \Delta(A)$ for each $\theta \in \Theta$, for which the notion of Bayesian solutions introduced later is more appropriate.

By extending no-swap-regret-dynamics to the agent normal form, we can obtain dynamics converging to ANFCEs.
In these dynamics, for every time round $t \in [T]$, each hypothetical player $(i,\theta_i) \in N'$ decides a distribution $\pi_i^t(\theta_i) \in \Delta(A_i)$ over their actions.
Hence, swap regret for each hypothetical player $(i,\theta'_i) \in N'$ considers an action swap $\phi_{\theta'_i} \colon A_i \to A_i$.
As with the incentive constraints discussed above, we consider the sum of swap regrets for all hypothetical players corresponding to $i \in N$, which we call \textit{type-wise swap regret}.

\begin{definition}[{Type-wise swap regret}]
For online learning with stochastic types specified by actions $A_i$, types $\Theta_i$, prior distribution $\rho_i$, and reward vector $u_i^t \in [0,1]^{\Theta_i \times A_i}$ for each round $t \in [T]$, type-wise swap regret is defined as
\begin{equation*}
\RTSi^T = \max_{\phi \colon \Theta_i \times A_i \to A_i} \sum_{t=1}^T \E_{\theta_i \sim \rho_i} \left[ \E_{a_i \sim \pi_i^t(\theta_i)} \left[ u_i^t(\theta_i, \phi(\theta_i, a_i)) - u_i^t(\theta_i, a_i) \right] \right].
\end{equation*}
\end{definition}

\begin{algorithm}[t]
\caption{Dynamics for ANFCEs}\label{alg:anf-dynamics}
	For each $i \in N$, let $\calA_i$ be a subroutine that minimizes type-wise swap regret for online learning with stochastic types.\\
	\For{each round $t = 1,\dots,T$}{
		Each player $i \in N$ decides their randomized strategy $\pi^t_{i} \in \Delta(A_i)^{\Theta_i}$ according to $\calA_{i}$ and shares it with the other players.\\
		Each player $i \in N$ computes reward $u_i^t(\theta_i, a_i) = \E_{\theta_{-i} \sim \rho|\theta_i} \left[ \E_{a_{-i} \sim \pi^t_{-i}(\theta_{-i})} \left[ v_i(\theta; a) \right] \right]$ for every $\theta_i \in \Theta_i$ and $a_i \in A_i$, where $\pi^t_{-i}(\theta_{-i}) \in \Delta(A_{-i})$ is the product distribution that independently generates $a_j \sim \pi_j^t(\theta_j)$ for each $j \in N \setminus \{i\}$.\\
		Feed $u_i^t \in [0,1]^{\Theta_i \times A_i}$ to each $\calA_{i}$ as a reward vector for round $t$.
	}
\end{algorithm}

As with SFCEs, by directly applying Theorem 3 of \citet*{BM07} to the agent normal form of Bayesian games, we obtain the following claim.

\begin{proposition}\label{prop:dynamics-anf}
Let $\pi_i^t \in \Delta(A_i)^{\Theta_i}$ be the type-wise distribution of each player $i \in N$ for each round $t \in [T]$ in \Cref{alg:anf-dynamics}.
Then, the empirical distribution $\sigma \in \Delta(S)$,
which is defined by $\sigma(s) = \frac{1}{T} \sum_{t=1}^T \prod_{i \in N} \prod_{\theta_i \in \Theta_i} \pi_{i}^t(\theta_i; s_i(\theta_i))$ for each $s \in S$,
is a $\frac{\max_{i \in N} \RTSi^T}{T}$-ANFCE,
where $\RTSi^T$ is the type-wise swap regret for each subroutine $\calA_i$.
\end{proposition}

\begin{proof}
From the definition of $\epsilon$-ANFCEs, it is sufficient to prove
\begin{equation*}
\E_{\theta \sim \rho} \left[ \E_{s \sim \sigma} \left[ v_i(\theta; \phi(\theta_i, s_i(\theta_i)), s_{-i}(\theta_{-i})) \right] \right] - \E_{\theta \sim \rho} \left[ \E_{s \sim \sigma} \left[ v_i(\theta; s(\theta)) \right] \right]
\le \frac{\max_{j \in N} \RTSj^T}{T}
\end{equation*}
for each $i \in N$ and $\phi \colon \Theta_i \times A_i \to A_i$.
Let $\sigma^t \in \Delta(S)$ be the product distribution that independently generates $s_j(\theta_j) \sim \pi_j^t(\theta_j)$ for each $j \in N$ and $\theta_j \in \Theta_j$.
For each $\theta_{-i} \in \Theta_{-i}$, let $\pi_{-i}^t(\theta_{-i}) \in \Delta(A_{-i})$ be the product distribution that independently generates $a_j \sim \pi_j^t(\theta_j)$ for each $j \in N \setminus \{i\}$.
Then the left-hand side can be bounded as
\begin{align*}
&\E_{\theta \sim \rho} \left[ \E_{s \sim \sigma} \left[v_i(\theta; \phi(\theta_i, s_i(\theta_i)), s_{-i}(\theta_{-i})) \right] \right] - \E_{\theta \sim \rho} \left[ \E_{s \sim \sigma} \left[ v_i(\theta; s(\theta)) \right] \right]\\
&= 
\E_{t \sim [T]} \left[ \E_{\theta \sim \rho} \left[ \E_{s \sim \sigma^t} \left[ v_i(\theta; \phi(\theta_i, s_i(\theta_i)), s_{-i}(\theta_{-i})) - v_i(\theta; s(\theta)) \right] \right] \right] \tag{from the definition of $\sigma$} \\
&= 
\E_{t \sim [T]} \left[ \E_{\theta \sim \rho} \left[ \E_{a_i \sim \pi_i^t(\theta_i)} \left[ \E_{a_{-i} \sim \pi_{-i}^t(\theta_{-i})} \left[ v_i(\theta; \phi(\theta_i, a_i), a_{-i}) - v_i(\theta; a) \right] \right] \right] \right] \tag{from the definition of $\sigma^t$} \\
&=
\frac{1}{T} \sum_{t=1}^T \E_{\theta_i \sim \rho_i} \left[ \E_{a_i \sim \pi_i^t(\theta_i)} \left[ \E_{\theta_{-i} \sim \rho|\theta_i} \left[ \E_{a_{-i} \sim \pi_{-i}^t(\theta_{-i})} \left[ v_i(\theta; \phi(\theta_i,a_i),a_{-i})
- 
v_i(\theta; a)\right] \right] \right] \right]\\
&=
\frac{1}{T} \sum_{t=1}^T \E_{\theta_i \sim \rho_i} \left[ \E_{a_i \sim \pi_i^t(\theta_i)} \left[ u_i^t(\theta_i, \phi(\theta_i, a_i)) - u_i^t(\theta_i, a_i) \right] \right] \tag{from the definition of $u_i^t$}\\
&\le \frac{\RTSi^T}{T} \\
&\le \frac{\max_{j \in N} \RTSj^T}{T},
\end{align*}
which completes the proof.
\end{proof}

A detailed description of the dynamics can be found in \Cref{alg:anf-dynamics}.
We can obtain an algorithm for minimizing type-wise swap regret by running the swap regret minimization algorithm proposed by \citet*{BM07} separately for each $(i,\theta_i) \in N'$.
Since the rewards for each $(i,\theta_i) \in N'$ are always bounded above by $\rho_i(\theta_i)$, an upper bound of $O(\sqrt{T |A_i| \log |A_i|})$ on swap regret directly applies to type-wise swap regret.

As with $\epsilon$-SFCEs, we can define its counterpart $\PiANF^\epsilon \subseteq \Delta(A)^\Theta$ by $\PiANF^\epsilon = \eta(\SigmaANF^\epsilon)$ (see \Cref{sec:sr} for the definition of $\eta$).
Since the incentive constraints for ANFCEs only rely on the action profile $s(\theta)$ instead of the strategy profile $s$, we can characterize $\PiANF^\epsilon$ as follows.

\begin{proposition}\label{prop:anf-pi}
For any type-wise distribution $\pi \in \Delta(A)^\Theta$, the following are equivalent:
\begin{itemize}
\item[(i)] $\pi \in \PiANF^\epsilon$.
\item[(ii)] $\pi$ is strategy-representable, and
for any $i \in N$ and any $\phi \colon \Theta_i \times A_i \to A_i$, it holds that
\begin{equation}\label{eq:ic-anf-pi}
\tag{$\mathrm{IC}'_{\mathrm{ANF}}$}
\E_{\theta \sim \rho} \left[ \E_{a \sim \pi(\theta)} \left[ v_i(\theta; a) \right] \right]
\ge
\E_{\theta \sim \rho} \left[ \E_{a \sim \pi(\theta)} \left[ v_i(\theta; \phi(\theta_i, a_i),a_{-i}) \right] \right] - \epsilon.
\end{equation}
\end{itemize}
\end{proposition}

\begin{proof}
First, we prove (ii) assuming (i).
Suppose $\pi \in \PiANF^\epsilon$.
Since $\PiANF^\epsilon = \eta(\SigmaANF^\epsilon)$, there exists $\sigma \in \SigmaANF^\epsilon$ such that $\eta(\sigma) = \pi$.
From the definition of $\SigmaANF^\epsilon$, the incentive constraint \eqref{eq:ic-anf} holds for every $i \in N$ and $\phi \colon \Theta_i \times A_i \to A_i$.
Since $\pi = \eta(\sigma)$, when $s$ is generated from $\sigma$, the distribution of $s(\theta)$ equals $\pi(\theta)$.
Then \eqref{eq:ic-anf} implies \eqref{eq:ic-anf-pi}.

Next, we prove (i) assuming (ii).
Suppose $\pi$ is strategy-representable and satisfies the incentive constraint \eqref{eq:ic-anf-pi} for every $i \in N$ and $\phi \colon \Theta_i \times A_i \to A_i$.
From the strategy representability, there exists $\sigma \in \Delta(S)$ such that $\pi = \eta(\sigma)$.
As discussed above, when $s$ is generated from $\sigma$, the distribution of $s(\theta)$ equals $\pi(\theta)$.
Then \eqref{eq:ic-anf-pi} implies \eqref{eq:ic-anf}.
\end{proof}

\begin{table}
\caption{A comparison of the classes of Bayes correlated equilibria (or their counterparts in $\Delta(A)^\Theta$). SR stands for ``strategy representability'' defined in \Cref{sec:sr}.}\label{tab:bce}
\centering
\begin{tabular}{cc}
\toprule
Class & Characterization\\
\cmidrule(lr){1-2}
Strategic-form correlated equilibria & Incentive constraints w.r.t.\ strategies \eqref{eq:ic-sf} \& SR\\
Agent-normal-form correlated equilibria & Incentive constraints w.r.t.\ actions \eqref{eq:ic-anf-pi} \& SR\\
Communication equilibria & Incentive constraints w.r.t.\ types and actions \eqref{eq:ic-com}\\
Bayesian solutions & Incentive constraints w.r.t.\ actions \eqref{eq:ic-anf-pi}\\
\hline
\end{tabular}
\end{table}

\subsection{Relations among classes of Bayes correlated equilibria}\label{sec:relation}

Here, we compare the classes of Bayes correlated equilibria and show their relations, summarized in \Cref{tab:bce}.
First, we define an approximate version of \textit{Bayesian solutions} (also called \textit{partial Bayesian approach}) \citep*{Forges93}.

\begin{definition}[$\epsilon$-approximate Bayesian solutions]
For any $\epsilon \ge 0$, a type-wise distribution $\pi \in \Delta(A)^\Theta$ is an $\epsilon$-approximate Bayesian solution
if
for any $i \in N$ and $\phi \colon \Theta_i \times A_i \to A_i$, it holds that
\begin{equation}\label{eq:ic-bs}
\tag{$\mathrm{IC}_{\mathrm{BS}}$}
\E_{\theta \sim \rho} \left[ \E_{a \sim \pi(\theta)} \left[ v_i(\theta; a) \right] \right]
\ge
\E_{\theta \sim \rho} \left[ \E_{a \sim \pi(\theta)} \left[ v_i(\theta; \phi(\theta_i, a_i), a_{-i}) \right] \right] - \epsilon.
\end{equation}
Let $\PiBS^\epsilon \subseteq \Delta(A)^\Theta$ be the set of all $\epsilon$-approximate Bayesian solutions.
\end{definition}

\citet*{Forges93} showed the relations among the classes in terms of the set of players' payoff values.
Formally, if we denote the expected payoff vector achieved by $\pi \in \Delta(A)^\Theta$ by
\begin{equation*}
\vEP(\pi) = \left(\E_{\theta \sim \rho} \left[ \E_{a \sim \pi(\theta)} \left[ v_i(\theta; a) \right] \right] \right)_{i \in N} \in [0,1]^{N},
\end{equation*}
then \citet*{Forges93} showed $\vEP(\PiSF^0) \subseteq \vEP(\PiANF^0) \cap \vEP(\PiCom^0)$ and $\vEP(\PiANF^0) \cup \vEP(\PiCom^0) \subseteq \vEP(\PiBS^0)$.
Moreover, \citet*{Forges93} showed that these inclusions are strict, and $\vEP(\PiANF^0) \not\subseteq \vEP(\PiCom^0)$ and $\vEP(\PiCom^0) \not\subseteq \vEP(\PiANF^0)$ for some Bayesian games.

Here, we show the same inclusion relations among the approximate versions in terms of type-wise distributions, that is,
\begin{equation*}
\PiSF^\epsilon \subseteq \PiANF^\epsilon \cap \PiCom^\epsilon 
\quad\text{and}\quad
\PiANF^\epsilon \cup \PiCom^\epsilon \subseteq \PiBS^\epsilon
\end{equation*}
for each $\epsilon \ge 0$.
We also show that these inclusion relations are strict.
First, we show the following relation.

\begin{proposition}\label{prop:intersection}
For any Bayesian game, it holds that $\PiSF^\epsilon \subseteq \PiANF^\epsilon \cap \PiCom^\epsilon$ for any $\epsilon \ge 0$.
\end{proposition}

\begin{proof}
Let $\sigma \in \SigmaSF^\epsilon$ be an arbitrary $\epsilon$-SFCE and $\pi = \eta(\sigma) \in \PiSF^\epsilon$ its corresponding type-wise distribution.
From \Cref{prop:com-sr}, it is sufficient to prove the incentive constraint \eqref{eq:ic-comanf} for each $i \in N$, $\psi \colon \Theta_i \to \Theta_i$, and $\phi \colon \Theta_i \times A_i \to A_i$.
Since $\sigma$ is an $\epsilon$-SFCE, the incentive constraint \eqref{eq:ic-sf} holds for any $\phiSF \colon S_i \to S_i$.
If we set $(\phiSF(s_i))(\theta_i) = \phi(\theta_i, s_i(\psi(\theta_i)))$ for each $\theta_i \in \Theta_i$,
then it holds that
\begin{equation*}
\E_{\theta \sim \rho} \left[ \E_{s \sim \sigma} \left[ v_i(\theta; s(\theta)) \right] \right]
\ge
\E_{\theta \sim \rho} \left[ \E_{s \sim \sigma} \left[ v_i(\theta; \phi(\theta_i, s_i(\psi(\theta_i))),s_{-i}(\theta_{-i})) \right] \right] - \epsilon,
\end{equation*}
which completes the proof.
\end{proof}

Next, we provide an example of a Bayesian game in which $\pi \in (\PiANF^0 \cap \PiCom^0) \setminus \PiSF^0$ exists.
While deviations for SFCEs allow players to use recommendations for all other types, the deviations for communication equilibria allow players to use recommendations for at most one other type and ANFCE no other type.
We design a distribution such that a player can gain by using recommendations for two other types.
Then this distribution is a communication equilibrium and ANFCE but not SFCE.

\begin{proposition}
For some Bayesian game, $(\PiANF^0 \cap \PiCom^0) \setminus \PiSF^0 \neq \emptyset$.
\end{proposition}

\begin{proof}
The distribution is defined as follows.
Let $\Theta_1 = \{\theta_1,\theta'_1,\theta''_1\}$ and $\Theta_2 = \{\theta_2\}$ with the uniform prior distribution $\rho$ over $\Theta_1 \times \Theta_2$.
Let $A_1 = \{0,1,2,3,4\}$ and $A_2 = \{1,2,3,4\}$.
We define the first player's payoffs for $\theta_1, \theta'_1, \theta''_1$ by
\begin{align*}
v_1(\theta_1, \theta_2; a_1,a_2) &= \begin{cases}
\frac{1}{2} & \text{if $a_1=0$}\\
1 & \text{if $a_1 \neq 0$ and $a_1=a_2$}\\
0 & \text{if $a_1 \neq 0$ and $a_1 \neq a_2$}.
\end{cases}\\
v_1(\theta'_1, \theta_2; a_1,a_2) &= \begin{cases}
\frac{1}{2} & \text{if $(a_1,a_2) \in \{(4,1),(4,2),(1,3),(1,4)\}$}\\
0 & \text{otherwise}
\end{cases}\\
v_1(\theta''_1, \theta_2; a_1,a_2) &= \begin{cases}
\frac{1}{2} & \text{if $(a_1,a_2) \in \{(2,1),(2,3),(3,2),(3,4)\}$}\\
0 & \text{otherwise}.
\end{cases}
\end{align*}
The second player's payoff $v_2$ is defined as $v_2(\tilde{\theta}_1,\theta_2; a_1,a_2) = \bfone_{\{a_1 \neq a_2\}}$ for any $\tilde{\theta}_1 \in \Theta_1$.
We consider $\pi \in \Delta(A)^\Theta$ such that $\pi(\theta_1,\theta_2)$, $\pi(\theta'_1,\theta_2)$, and $\pi(\theta''_1,\theta_2)$ are the uniform distributions over $\{0\} \times \{1,2,3,4\}$, $\{(4,1),(4,2),(1,3),(1,4)\}$, and $\{(2,1),(2,3),(3,2),(3,4)\}$, respectively.
\begin{itemize}

\item
First, we show that $\pi$ is strategy-representable by providing $\sigma \in \Delta(S)$ such that $\pi = \eta(\sigma)$.
From the definition of $\pi(\theta_1,\theta_2)$, it must hold that $s_1(\theta_1) = 0$ with probability $1$.
From the definition of $\pi(\theta'_1,\theta_2)$, it holds that 
\begin{align*}
&\Pr_{s \sim \sigma}(s_1(\theta'_1) = 4 \mid s_2(\theta_2) = 1) = 1
&\qquad
&\Pr_{s \sim \sigma}(s_1(\theta'_1) = 4 \mid s_2(\theta_2) = 2) = 1\\
&\Pr_{s \sim \sigma}(s_1(\theta'_1) = 1 \mid s_2(\theta_2) = 3) = 1
&\qquad
&\Pr_{s \sim \sigma}(s_1(\theta'_1) = 1 \mid s_2(\theta_2) = 4) = 1.
\end{align*}
Similarly, from the definition of $\pi(\theta''_1,\theta_2)$, it holds that
\begin{align*}
&\Pr_{s \sim \sigma}(s_1(\theta''_1) = 2 \mid s_2(\theta_2) = 1) = 1
&\qquad
&\Pr_{s \sim \sigma}(s_1(\theta''_1) = 3 \mid s_2(\theta_2) = 2) = 1\\
&\Pr_{s \sim \sigma}(s_1(\theta''_1) = 2 \mid s_2(\theta_2) = 3) = 1
&\qquad
&\Pr_{s \sim \sigma}(s_1(\theta''_1) = 3 \mid s_2(\theta_2) = 4) = 1.
\end{align*}
Hence, $\sigma$ can be defined as
\begin{equation}\label{eq:corresponding-sigma}
\begin{aligned}
&\Pr_{s \sim \sigma}(s_1(\theta_1) = 0, ~s_1(\theta'_1) = 4, ~s_1(\theta''_1) = 2, ~s_2(\theta_2) = 1) = \frac{1}{4}\\
&\Pr_{s \sim \sigma}(s_1(\theta_1) = 0, ~s_1(\theta'_1) = 4, ~s_1(\theta''_1) = 3, ~s_2(\theta_2) = 2) = \frac{1}{4}\\
&\Pr_{s \sim \sigma}(s_1(\theta_1) = 0, ~s_1(\theta'_1) = 1, ~s_1(\theta''_1) = 2, ~s_2(\theta_2) = 3) = \frac{1}{4}\\
&\Pr_{s \sim \sigma}(s_1(\theta_1) = 0, ~s_1(\theta'_1) = 1, ~s_1(\theta''_1) = 3, ~s_2(\theta_2) = 4) = \frac{1}{4}.
\end{aligned}
\end{equation}
Therefore, $\pi$ is strategy-representable.

\item
Next, we show that for any $\psi \colon \Theta_i \to \Theta_i$ and $\phi \colon \Theta_i \times A_i \to A_i$, distribution $\pi$ satisfies the incentive constraint for communication equilibria \eqref{eq:ic-com}.
Since the second player's payoff is at most $1$ and the expected payoff is $1$ under $\pi$, this player does not have any incentive to deviate from the recommendation.
Under distribution $\pi$, the first player always obtains payoff $\frac{1}{2}$, and then the expected payoff is $\frac{1}{2}$.
Then we show that the first player cannot gain by using any combination of $\psi$ and $\phi$.
When the type is $\theta'_1$ or $\theta''_1$, the payoff is at most $\frac{1}{2}$ and the expected payoff is $\frac{1}{2}$ under $\pi$, the first player does not have any incentive to deviate from the recommendation.
When the type is $\theta_1$, we consider the case of $\psi(\theta_1) = \theta_1$ and the others separately.
If $\psi(\theta_1) = \theta_1$, then the first player is always recommended taking the action $0$, and the posterior distribution on the second player's action is still the uniform distribution.
If $\psi(\theta_1) = \theta'_1$ (or $\psi(\theta_1) = \theta''_1$), then the first player can learn whether $s_2(\theta_2) \in \{1,2\}$ or $s_2(\theta_2) \in \{3,4\}$ ($s_2(\theta_2) \in \{1,3\}$ or $s_2(\theta_2) \in \{2,4\}$, respectively).
By using this information, the first player can take the same action $s_2(\theta_2)$ with probability $1/2$, but the expected payoff is $\frac{1}{2} \cdot 1 = \frac{1}{2}$, which is not better than the recommendation. 
Therefore, the right-hand side is at most $\frac{1}{2}$ for any $\psi$ and $\phi$.

\end{itemize}
These two facts imply $\pi \in \PiANF^0 \cap \PiCom^0$ from \Cref{prop:com-sr}.
Finally, we show $\pi \not\in \PiSF^0$.
As discussed above, it is sufficient to consider $\sigma$ that satisfies the conditions \eqref{eq:corresponding-sigma}.
When the first player's type is $\theta_1$, the first player can learn $s_2(\theta_2)$ from $s_1(\theta'_1)$ and $s_1(\theta''_1)$.
Then $\phiSF$ can be defined such that $(\phiSF(s_1))(\theta_1) = s_2(\theta_2)$ always holds.
This deviation provides the expected payoff $\frac{1}{3} \cdot 1 + \frac{1}{3} \cdot \frac{1}{2} + \frac{1}{3} \cdot \frac{1}{2} = \frac{2}{3}$ to the first player, which is better than $\frac{1}{2}$ obtained by the recommendation $\sigma$.
Therefore, $\sigma$ is not an SFCE and $\pi \not\in \PiSF^0$.
\end{proof}

\begin{remark}
The example used in the proof also proves that $\vEP(\PiANF^0 \cap \PiCom^0) \setminus \vEP(\PiSF^0) \neq \emptyset$ for some Bayesian game.
This is stronger than the claim that $(\vEP(\PiANF^0) \cap \vEP(\PiCom^0)) \setminus \vEP(\PiSF^0) \neq \emptyset$ for some Bayesian game, which was proved by \citet*{Forges93} with distinct $\pi \in \PiANF^0$ and $\pi' \in \PiCom^0$ such that $\vEP(\pi) = \vEP(\pi')$.

We can check $\vEP(\pi) = (\frac{1}{2},1) \not\in \vEP(\PiSF^0)$ as follows.
If the second player's expected payoff is $1$, then $a_1 \neq a_2$ holds with probability $1$.
Hence, the first player with type $\theta_1$ never obtains the payoff $1$.
To achieve the expected payoff $\frac{1}{2}$, the first player must achieve the payoff $\frac{1}{2}$ with probability $1$, but then $\sigma$ must satisfy the conditions \eqref{eq:corresponding-sigma}.
As proved above, $\sigma$ is not an SFCE.
\end{remark}

\paragraph{Bayesian solutions}
Next, we show that $\epsilon$-approximate Bayesian solutions contain the union of $\epsilon$-ANFCEs and $\epsilon$-approximate communication equilibria.

\begin{proposition}\label{prop:union}
For any Bayesian game, it holds that $\PiANF^\epsilon \cup \PiCom^\epsilon \subseteq \PiBS^\epsilon$ for any $\epsilon \ge 0$.
\end{proposition}

\begin{proof}
Assume $\pi \in \PiANF^\epsilon$ or $\pi \in \PiCom^\epsilon$.
If $\pi \in \PiANF^\epsilon$, then $\pi$ satisfies the incentive constraint \eqref{eq:ic-anf-pi} for any $\phi \colon \Theta_i \times A_i \to A_i$.
Since \eqref{eq:ic-bs} is identical to \eqref{eq:ic-anf-pi}, this implies $\pi \in \PiBS^\epsilon$.
If $\pi \in \PiCom^\epsilon$, then $\pi$ satisfies the incentive constraint \eqref{eq:ic-com} for any $\psi \colon \Theta_i \to \Theta_i$ and $\phi \colon \Theta_i \times A_i \to A_i$.
By considering $\psi$ as the identity map, we can obtain \eqref{eq:ic-bs}, which implies $\pi \in \PiBS^\epsilon$.
\end{proof}

This inclusion for $\epsilon = 0$ is strict for some Bayesian game.
\citet*{Forges93} provided an example of a Bayesian game for which $\vEP(\PiBS^0) \setminus (\vEP(\PiANF^0) \cup \vEP(\PiCom^0)) \neq \emptyset$ holds.
Since $\vEP(\PiANF^0) \cup \vEP(\PiCom^0) = \vEP(\PiANF^0 \cup \PiCom^0)$ holds in general, $\vEP(\PiBS^0) \setminus \vEP(\PiANF^0 \cup \PiCom^0) \neq \emptyset$ in this game, which implies $\PiBS^0 \setminus (\PiANF^0 \cup \PiCom^0) \neq \emptyset$.

\paragraph{Bayes--Nash equilibria}

Finally, we consider the relations between Bayes--Nash equilibria and the classes of Bayes correlated equilibria.
Recall that there are two interpretations of a Bayesian game: the strategic form and the agent normal form.
For each of these two different normal-form games corresponding to the same Bayesian game, we can consider the concept of Nash equilibria.
A well-known interesting fact is that these two concepts are equivalent.
This equilibrium concept is called \textit{Bayes--Nash equilibrium}.
Here, we define approximate versions of Bayes--Nash equilibria in the strategic form and the agent normal form, respectively.

First, we define Bayes--Nash equilibria as Nash equilibria of the strategic form.
In the strategic form, each player $i \in N$ independently decides a strategy $s_i \in S_i$ according to some distribution $\sigma_i \in \Delta(S_i)$.
Here, we denote this Nash equilibrium by the product distribution of $\sigma_1, \dots, \sigma_n$.
Let $\SigmaProd \subseteq \Delta(S)$ be the set of all product distributions, that is, for each $\sigma \in \SigmaProd$, there exists some $\sigma_i \in \Delta(S_i)$ for each $i \in N$ such that $\sigma(s) = \prod_{i=1}^n \sigma_i(s_i)$ holds for every $s \in S$.
We then define an approximate version of Bayes--Nash equilibria in the strategic form as follows.

\begin{definition}[{$\epsilon$-Approximate Bayes--Nash equilibria ($\epsilon$-BNEs) in the strategic form}]
For any $\epsilon \ge 0$, a distribution $\sigma \in \SigmaProd \subseteq \Delta(S)$ is an $\epsilon$-approximate Bayes--Nash equilibrium in the strategic form
if
for any $i \in N$ and any $\phiSF \colon S_i \to S_i$, it holds that
\begin{equation}\label{eq:ic-bne}
\tag{$\mathrm{IC}_{\mathrm{BNE}}$}
\E_{\theta \sim \rho} \left[ \E_{s \sim \sigma} \left[ v_i(\theta; s(\theta)) \right] \right]
\ge                          
\E_{\theta \sim \rho} \left[ \E_{s \sim \sigma} \left[ v_i(\theta; (\phiSF(s_i))(\theta_i),s_{-i}(\theta_{-i})) \right] \right] - \epsilon.
\end{equation}
Let $\SigmaBNE^\epsilon \subseteq \Delta(S)$ be the set of all $\epsilon$-approximate Bayes--Nash equilibria in the strategic form.
\end{definition}

Next, we define Bayes--Nash equilibria as Nash equilibria in the agent normal form.
In the agent normal form, each hypothetical player $(i,\theta_i) \in N'$ decides a distribution $\pi_{i,\theta_i} \in \Delta(A_i)$.
A type profile $\theta \in \Theta$ sampled from $\rho \in \Delta(\Theta)$ determines the active players,
whose actions $a_i \in A_i$ are independently sampled from $\pi_{i,\theta_i}$ for each $i \in N$.
Therefore, the action profile of the active players follows a type-wise product distribution $\pi \in \PiProd$.
Recall that we call $\pi \in \Delta(A)^\Theta$ a type-wise product distribution if there exists some $\pi_{i,\theta_i} \in \Delta(A_i)$ for each $i \in N$ and $\theta_i \in \Theta_i$ such that $\pi(\theta;a) = \prod_{i \in N} \pi_{i,\theta_i}(a_i)$ for every $\theta \in \Theta$ and $a \in A$.
We define an approximate version of Bayes--Nash equilibria in the agent normal form as follows.

\begin{definition}[{$\epsilon$-Approximate Bayes--Nash equilibria ($\epsilon$-BNEs) in the agent normal form}]
For any $\epsilon \ge 0$, a type-wise product distribution $\pi \in \PiProd \subseteq \Delta(A)^\Theta$ is an $\epsilon$-approximate Bayes--Nash equilibria in the agent normal form
if
for any $i \in N$ and $\phi \colon \Theta_i \times A_i \to A_i$, it holds that
\begin{equation}\label{eq:ic-bne-pi}
\tag{$\mathrm{IC}'_{\mathrm{BNE}}$}
\E_{\theta \sim \rho} \left[ \E_{a \sim \pi(\theta)} \left[ v_i(\theta; a) \right] \right]
\ge
\E_{\theta \sim \rho} \left[ \E_{a \sim \pi(\theta)} \left[ v_i(\theta; \phi(\theta_i, a_i), a_{-i}) \right] \right] - \epsilon.
\end{equation}
Let $\PiBNE^\epsilon \subseteq \Delta(A)^\Theta$ be the set of all $\epsilon$-approximate Bayes--Nash equilibria in the agent normal form.
\end{definition}

Here, we show the equivalence of $\SigmaBNE^\epsilon$ and $\PiBNE^\epsilon$.
This is an approximate version of the well-known fact that Nash equilibria in the strategic form and the agent normal form are equivalent \citep*{Harsanyi67}.

\begin{proposition}\label{prop:bne}
For any distribution $\sigma \in \SigmaBNE^\epsilon$, its corresponding type-wise distribution satisfies $\eta(\sigma) \in \PiBNE^\epsilon$.
Conversely,
for any type-wise distribution $\pi \in \PiBNE^\epsilon$, there exists a distribution $\sigma \in \SigmaBNE^\epsilon$ such that $\eta(\sigma) = \pi$.
\end{proposition}

Finally, we relate the classes of Bayes correlated equilibria to Bayes--Nash equilibria.
The following proposition shows that if a type-wise product distribution is contained in each class of Bayes correlated equilibria, it is a Bayes--Nash equilibrium.

\begin{proposition}
$\PiSF^\epsilon \cap \PiProd = \PiANF^\epsilon \cap \PiProd = \PiCom^\epsilon \cap \PiProd = \PiBS^\epsilon \cap \PiProd = \PiBNE^\epsilon$.
\end{proposition}

\begin{proof}
From \Cref{prop:intersection,prop:union}, $\PiSF^\epsilon$ is a subset of $\PiANF^\epsilon$, $\PiCom^\epsilon$, and $\PiBS^\epsilon$.
Moreover, $\PiBS^\epsilon$ is a superset of $\PiSF^\epsilon$, $\PiANF^\epsilon$, and $\PiCom^\epsilon$.
Therefore, it is sufficient to show that $\PiBNE^\epsilon \subseteq \PiSF^\epsilon \cap \PiProd$ and $\PiBS^\epsilon \cap \PiProd \subseteq \PiBNE^\epsilon$.

First, we prove $\PiBNE^\epsilon \subseteq \PiSF^\epsilon \cap \PiProd$.
Suppose $\pi \in \PiBNE^\epsilon$.
From \Cref{prop:bne}, there exists $\sigma \in \SigmaBNE^\epsilon$ such that $\eta(\sigma) = \pi$.
Since this $\sigma$ satisfies \eqref{eq:ic-bne} for each $i \in N$ and $\phiSF \colon S_i \to S_i$, which is equivalent to \eqref{eq:ic-sf}.
Hence, $\sigma \in \SigmaSF^\epsilon$, and then $\pi = \eta(\sigma) \in \PiSF^\epsilon$.
Moreover, from the definition of $\epsilon$-BNEs in the agent normal form, $\pi$ is a type-wise product distribution, which implies $\PiBNE^\epsilon \subseteq \PiSF^\epsilon \cap \PiProd$.

Next, we prove $\PiBS^\epsilon \cap \PiProd \subseteq \PiBNE^\epsilon$.
Suppose $\pi \in \PiBS^\epsilon \cap \PiProd$.
Since $\pi \in \PiBS^\epsilon$, the incentive constraint \eqref{eq:ic-bs} holds for each $i \in N$ and $\phi \colon \Theta_i \times A_i$, which is equivalent to \eqref{eq:ic-bne-pi}.
Moreover, $\pi$ is a type-wise product distribution since $\pi \in \PiProd$.
From the definition of $\epsilon$-BNEs in the agent normal form, we obtain $\pi \in \PiBNE^\epsilon$.
\end{proof}

\section{Algorithm for minimizing strategy swap regret}\label{sec:min-strategy-swap}

Here we propose an algorithm for minimizing strategy swap regret $\RSSi$ for an online learning problem with stochastic types.

Since strategy swap regret $\RSSi$ is swap regret when $S_i$ is regarded as the decision space,
we can apply the swap regret minimization algorithm proposed by \citet*{BM07}.
The direct application leads to an upper bound of $O(\sqrt{T |S_i| \log |S_i|})$ on the strategy swap regret.

Their algorithm reduces swap regret minimization to $|S_i|$ external regret minimization problems with decision space $S_i$.
In the case of strategy swap regret, we can further reduce them to $|S_i||\Theta_i|$ external regret minimization problems with decision space $A_i$.
As a result, we obtain an upper bound of $O(\sqrt{T |S_i| \log |A_i|})$ on the strategy swap regret, which is slightly better than the above bound.
Since $|S_i| = |A_i|^{|\Theta_i|}$, this bound is still exponentially large in the number of types $|\Theta_i|$.
Moreover, the algorithm uses exponentially many subroutines, which require exponential time computation for each round.
It is open whether there is an efficient algorithm with strategy swap regret sublinear in $T$ and polynomial in $n$, $|A_i|$, and $|\Theta_i|$.
If it exists, we can compute an $\epsilon$-SFCE in polynomial time by simulating the dynamics.

\begin{algorithm}[t]
\caption{Algorithm for minimizing strategy swap regret}\label{alg:strategy}
	\KwIn{The set of types $\Theta_i$ and the set of actions $A_i$ are specified in advance. The reward vector $u_i^t \in [0,1]^{\Theta_i \times A_i}$ is given at the end of each round $t \in [T]$.}
	Let $\calE_{s_i,\theta_i}$ be AdaHedge algorithm (see, e.g., \citep*[Section 7.6]{Orabona19}) with decision space $S_i$ for each $\theta_i \in \Theta_i$ and $s_i \in S_i$\;
	\For{each round $t = 1,\dots,T$}{
		Let $z^t_{s_i,\theta_i} \in \Delta(A_i)$ be the output of $\calE_{s_i,\theta_i}$ in round $t$ for each $s_i \in S_i$ and $\theta_i \in \Theta_i$\;
		Define $P^t \in [0,1]^{S_i \times S_i}$ by $P^t(s_i,s'_i) = \prod_{\theta_i \in \Theta_i} z^t_{s'_i,\theta_i}(s_i(\theta_i))$ for each $s_i,s'_i \in S_i$\;
		Compute $\sigma_i^t \in \Delta(S_i)$ such that $P^t \sigma_i^t = \sigma_i^t$ by eigenvector computation\;
		Let $\sigma_i^t$ be the output for round $t$\;
		Observe reward vector $u_i^t \in [0,1]^{\Theta_i \times A_i}$ and feed the reward vector $\tilde{u}^t_{s_i,\theta_i} \in [0,1]^{A_i}$ defined by $\tilde{u}^t_{s_i,\theta_i}(a_i) = \sigma_i^t(s_i) u_i^t(\theta_i, a_i) $ for each $a_i \in A_i$ to subroutine $\calE_{s_i,\theta_i}$ for each $s_i \in S_i$, $\theta_i \in \Theta_i$\;
	}
\end{algorithm}

The reduction proceeds as follows.
Let $\calE_{s_i,\theta_i}$ be AdaHedge (see, e.g., \citep*[Section 7.6]{Orabona19}) with decision space $A_i$ for each $\theta_i \in \Theta_i$ and $s_i \in S_i$.
In each round $t \in [T]$, each subroutine $\calE_{s_i,\theta_i}$ outputs $z^t_{s_i,\theta_i} \in \Delta(A_i)$.
From these outputs, we define the stochastic matrix $P^t \in [0,1]^{S_i \times S_i}$ by $P^t(s_i,s'_i) = \prod_{\theta_i \in \Theta_i} z^t_{s'_i,\theta_i}(s_i(\theta_i))$ for each $s_i,s'_i \in S_i$.
Since each $P^t \in \calP$ is a stochastic matrix, we can compute its stationary distribution $\sigma_i^t \in \Delta(S_i)$ that satisfies $P^t \sigma_i^t = \sigma_i^t$, and let $\sigma_i^t$ be the decision for round $t$.
Then feed the reward vector $\tilde{u}^t_{s_i,\theta_i} \in [0,1]^{A_i}$ defined by $\tilde{u}^t_{s_i,\theta_i}(a_i) = \sigma_i^t(s_i) u_i^t(\theta_i, a_i) $ for each $a_i \in A_i$ to subroutine $\calE_{s_i,\theta_i}$ for each $s_i \in S_i$, $\theta_i \in \Theta_i$.
Let
\begin{equation*}
R^T_{s_i,\theta_i} = \max_{a_i \in A_i} \sum_{t=1}^T \tilde{u}^t_{s_i,\theta_i}(a_i) - \sum_{t=1}^T \sum_{a_i \in A_i} z_{s_i,\theta_i}^t(a_i) \tilde{u}^t_{s_i,\theta_i}(a_i)
\end{equation*}
be the external regret of $\calE_{s_i,\theta_i}$ for each $s_i \in S_i$ and $\theta_i \in \Theta_i$.
Then strategy swap regret $\RSSi^T$ equals the sum of the external regrets as follows.

\begin{lemma}
\begin{equation*}
\RSSi^T = \sum_{\theta_i \in \Theta_i} \rho_i(\theta_i) \sum_{s_i \in S_i} R^T_{s_i,\theta_i}.
\end{equation*}
\end{lemma}

\begin{proof}
Recall that the strategy swap regret is defined as
\begin{equation*}
\RSSi^T = \max_{\phiSF \colon S_i \to S_i} \sum_{t=1}^T \E_{\theta_i \sim \rho_i} \left[ \E_{s_i \sim \sigma_i^t} \left[ u_i^t(\theta_i, (\phiSF(s_i))(\theta_i)) \right] \right] - \sum_{t=1}^T \E_{\theta_i \sim \rho_i} \left[ \E_{s_i \sim \sigma_i^t} \left[ u_i^t(\theta_i, s_i(\theta_i)) \right] \right].
\end{equation*}
We analyze the first term and the second term separately.
The first term is
\begin{align*}
&\max_{\phiSF \colon S_i \to S_i} \sum_{t=1}^T \E_{\theta_i \sim \rho_i} \left[ \E_{s_i \sim \sigma_i^t} \left[ u_i^t(\theta_i, (\phiSF(s_i))(\theta_i)) \right] \right]\\
&= \max_{\phiSF \colon S_i \to S_i} \sum_{t=1}^T \sum_{\theta_i \in \Theta_i} \rho_i(\theta_i) \sum_{s_i \in S_i} \sigma_i^t(s_i) u_i^t(\theta_i,(\phiSF(s_i))(\theta_i))\\
&= \max_{\phiSF \colon S_i \to S_i} \sum_{t=1}^T \sum_{\theta_i \in \Theta_i} \rho_i(\theta_i) \sum_{s_i \in S_i} \tilde{u}_{s_i,\theta_i}^t((\phiSF(s_i))(\theta_i)) \tag{from the definition of $\tilde{u}^t_{s_i,\theta_i}$}\\
&= \sum_{\theta_i \in \Theta_i} \rho_i(\theta_i) \sum_{s_i \in S_i} \max_{a_i \in A_i} \sum_{t=1}^T \tilde{u}_{s_i,\theta_i}^t(a_i),
\end{align*}
where the last equality holds because $(\phiSF(s_i))(\theta_i)$ can be optimized separately for each $s_i \in S_i$ and $\theta_i \in \Theta_i$.

The second term is
\begin{align*}
&\sum_{t=1}^T \E_{\theta_i \sim \rho_i} \left[ \E_{s_i \sim \sigma_i^t} \left[ u_i^t(\theta_i, s_i(\theta_i)) \right] \right]\\
&= \sum_{t=1}^T \sum_{\theta_i \in \Theta_i} \rho_i(\theta_i) \sum_{s_i \in S_i} \sigma_i^t(s_i) u_i^t(\theta_i,s_i(\theta_i))\\
&= \sum_{t=1}^T \sum_{\theta_i \in \Theta_i} \rho_i(\theta_i) \sum_{s_i,s'_i \in S_i} P^t(s_i,s'_i) \sigma_i^t(s'_i) u_i^t(\theta_i,s_i(\theta_i)) \tag{since $P^t \sigma_i^t = \sigma_i^t$}\\
&= \sum_{t=1}^T \sum_{\theta_i \in \Theta_i} \rho_i(\theta_i) \sum_{s_i,s'_i \in S_i} \sigma_i^t(s'_i) \prod_{\theta'_i \in \Theta_i} z^t_{s'_i,\theta'_i}(s_i(\theta'_i)) u_i^t(\theta_i,s_i(\theta_i)) \tag{from the definition of $P^t$}\\
&= \sum_{t=1}^T \sum_{\theta_i \in \Theta_i} \rho_i(\theta_i) \sum_{s_i,s'_i \in S_i} \prod_{\theta'_i \in \Theta_i} z^t_{s'_i,\theta'_i}(s_i(\theta'_i)) \tilde{u}_{s'_i,\theta_i}^t(s_i(\theta_i)) \tag{from the definition of $\tilde{u}^t_{s_i,\theta_i}$}\\
&= \sum_{t=1}^T \sum_{\theta_i \in \Theta_i} \rho_i(\theta_i) \sum_{s'_i \in S_i} \sum_{a_i \in A_i} \left( \sum_{\substack{s_i \in S_i\colon \\ s_i(\theta_i) = a_i}} \prod_{\theta'_i \in \Theta_i} z^t_{s'_i,\theta'_i}(s_i(\theta'_i)) \right) \tilde{u}_{s'_i,\theta_i}^t(a_i)\\
&= \sum_{t=1}^T \sum_{\theta_i \in \Theta_i} \rho_i(\theta_i) \sum_{s'_i \in S_i} \sum_{a_i \in A_i} z^t_{s'_i,\theta_i}(a_i) \tilde{u}_{s'_i,\theta_i}^t(a_i).
\end{align*}

Finally, we obtain
\begin{align*}
\RSSi^T
&= \sum_{\theta_i \in \Theta_i} \rho_i(\theta_i) \sum_{s_i \in S_i} \max_{a_i \in A_i} \sum_{t=1}^T \tilde{u}_{s_i,\theta_i}^t(a_i) - \sum_{t=1}^T \sum_{\theta_i \in \Theta_i} \rho_i(\theta_i) \sum_{s_i \in S_i} \sum_{a_i \in A_i} z^t_{s_i,\theta_i}(a_i) \tilde{u}_{s_i,\theta_i}^t(a_i) \\
&= \sum_{\theta_i \in \Theta_i} \rho_i(\theta_i) \sum_{s_i \in S_i} \left\{ \max_{a_i \in A_i} \sum_{t=1}^T \tilde{u}_{s_i,\theta_i}^t(a_i) - \sum_{t=1}^T \sum_{a_i \in A_i} z^t_{s_i,\theta_i}(a_i) \tilde{u}_{s_i,\theta_i}^t(a_i) \right\} \\
&= \sum_{\theta_i \in \Theta_i} \rho_i(\theta_i) \sum_{s_i \in S_i} R^T_{s_i,\theta_i}. \qedhere
\end{align*}
\end{proof}

By applying \Cref{thm:adahedge} (an upper bound for AdaHedge), we obtain an upper bound on strategy swap regret.

\begin{theorem}
\begin{equation*}
\RSSi \le 6 \sqrt{T |S_i| \log |A_i|}.
\end{equation*}
\end{theorem}

\begin{proof}
Since the sum of the maximum reward for $\calE_{s_i,\theta_i}$ is bounded above as $\sum_{t=1}^T \max_{a_i \in A_i} \tilde{u}^t_{s_i,\theta_i}(a_i) \le \sum_{t=1}^T \max_{a_i \in A_i} \sigma_i^t(s_i) u_i^t(\theta_i,a_i) \le \sum_{t=1}^T \sigma_i^t(s_i)$, we can obtain an upper bound on the external regret
\begin{equation*}
R^T_{s_i,\theta_i}
\le 6 \sqrt{\sum_{t=1}^T \sigma_i^t(s_i) \log |A_i|}.
\end{equation*}
By summing this upper bound for all $\theta_i \in \Theta_i$ and $s_i \in S_i$, we obtain
\begin{align*}
\RSSi^T
&= \sum_{\theta_i \in \Theta_i} \rho_i(\theta_i) \sum_{s_i \in S_i} R^T_{s_i,\theta_i}\\
&\le \sum_{\theta_i \in \Theta_i} \rho_i(\theta_i) \sum_{s_i \in S_i} \left\{ 6 \sqrt{\sum_{t=1}^T \sigma_i^t(s_i) \log |A_i|} \right\}\\
&\le \sum_{\theta_i \in \Theta_i} \rho_i(\theta_i) \left\{ 6 \sqrt{|S_i| \sum_{s_i \in S_i} \sum_{t=1}^T \sigma_i^t(s_i) \log |A_i|} \right\} \tag{by the Cauchy--Schwarz inequality}\\
&= \sum_{\theta_i \in \Theta_i} \rho_i(\theta_i) \left\{ 6 \sqrt{ T |S_i| \log |A_i|} \right\} \tag{since $\sum_{s_i \in S_i} \sigma_i^t(s_i) = 1$ for each $t \in [T]$}\\
&= 6 \sqrt{ T |S_i| \log |A_i|}. \qedhere
\end{align*}
\end{proof}

\section{On linear swap regret minimization}\label{sec:linear-swap}

\citet*{MMSS22} defined linear swap regret for online linear optimization with a polytope constraint as follows.

\begin{definition}[Linear swap regret {\citep*{MMSS22}}]
Let $\calP \subseteq \bbR^d$ be a polytope and $\calM(\calP)$ be the set of all valid linear transformations, where a linear transformation $M \colon \bbR^d \to \bbR^d$ is defined to be valid when $M x \in \calP$ holds for all $x \in \calP$.
For online linear optimization with reward vector $\bar{u}^t \in [0,1]^d$ for each $t \in [T]$ and feasible region $\calP$, linear swap regret is defined as
\begin{equation*}
\RLS^T = \max_{M \in \calM(\calP)} \sum_{t=1}^T \left( \langle M x^t, \bar{u}^t \rangle - \langle x^t, \bar{u}^t \rangle \right),
\end{equation*}
where $x^t \in \calP$ is the algorithm's output in each round $t \in [T]$.
\end{definition}

\citet*{MMSS22} defined linear swap regret for a general polytope but focused on a special case of $\calP = \calX$ with applications to two-player Bayesian games.
Recall that $\calX = \{ x \in [0,1]^{\Theta_i \times A_i} \mid \sum_{a_i \in A_i} x(\theta_i,a_i) = 1 ~ (\forall \theta_i \in \Theta_i) \}$ is the set of vectors that represent each $\pi_i \in \Delta(A_i)^{\Theta_i}$.

Here, we show that linear swap regret minimization with $\calP = \calX$ can be reduced to untruthful swap regret minimization.
Note that they propose an algorithm for minimizing linear swap regret in this special case (Algorithm 2 in their paper), but focused on guarantees on the Stackelberg value and did not provide any rigorous upper bound on linear swap regret.

The following proposition claims that all valid linear transformations can be expressed by some $Q \in \calQ$, which represents a linear transformation for untruthful swap regret.
Recall that $\calQ$ is defined by \eqref{eq:calQ}.

\begin{proposition}
For any $M \in \calM(\calX)$, there exists some $Q \in \calQ$ such that $Mx = Qx$ holds for any $x \in \calX$.
\end{proposition}

\begin{proof}
Fix any $M \in \calM(\calX)$.
We construct $Q \in \calQ$ such that $Mx=Qx$ holds for any $x \in \calX$.

Since $M$ is a valid linear transformation, $Mx \in \calX$ holds for any $x \in \calX$.
We use this fact for $x^1,x^2 \in \calX$ defined as follows.
Fix any $\theta'_i \in \Theta_i$ and $a^1_i, a^2_i \in A_i$.
Define $x^1(\theta'_i,a^1_i) = 1$ and $x(\theta'_i,a_i) = 0$ for any other $a_i \in A_i \setminus \{a^1_i\}$.
Similarly, let $x^2(\theta'_i,a^2_i) = 1$ and $x^2(\theta'_i,a_i) = 0$ for any other $a_i \in A_i \setminus \{a^2_i\}$.
The other entries of $x^1$ can be defined arbitrarily, and the other entries of $x^2$ are set to be equal to $x^1$.
That is, for any $\theta_i \in \Theta_i \setminus \{\theta'_i\}$, we must have $x^1(\theta_i,a_i) = x^2(\theta_i,a_i)$ for any $a_i \in A_i$.

Now we use $Mx^1 \in \calX$ and $Mx^2 \in \calX$.
For each $\theta_i \in \Theta_i$, we have
\begin{equation*}
\sum_{a_i \in A_i} (Mx^1)(\theta_i,a_i) = 1
\quad
\text{and}
\quad
\sum_{a_i \in A_i} (Mx^2)(\theta_i,a_i) = 1,
\end{equation*}
which implies
\begin{equation*}
\sum_{a_i \in A_i} (M(x^1-x^2))(\theta_i,a_i) = 0.
\end{equation*}
Since $(x^1-x^2)(\theta'_i,a^1_i) = 1$, $(x^1-x^2)(\theta'_i,a^2_i) = -1$, and the other entries of $x^1-x^2$ are $0$, we obtain
\begin{equation*}
\sum_{a_i \in A_i} M((\theta_i,a_i), (\theta'_i,a^1_i)) = \sum_{a_i \in A_i} M((\theta_i,a_i), (\theta'_i,a^2_i)).
\end{equation*}
Since this holds for any pair of $a^1_i,a^2_i \in A_i$, there exists some value $W(\theta_i,\theta'_i) \in \bbR$ for each $\theta_i,\theta'_i \in \Theta_i$ such that 
\begin{equation*}
\sum_{a_i \in A_i} M((\theta_i,a_i), (\theta'_i,a'_i)) = W(\theta_i,\theta'_i)
\end{equation*}
for any $a_i \in A_i$.
Moreover, by using $Mx \in \calX$ for each $x \in \calX$, it holds that for any $\theta_i \in \Theta_i$,
\begin{align*}
1
&=
\sum_{a_i \in A_i} (Mx)(\theta_i,a_i)\\
&=
\sum_{a_i \in A_i} \sum_{\theta'_i \in \Theta_i} \sum_{a'_i \in A_i} x(\theta'_i,a'_i) M((\theta_i,a_i),(\theta'_i,a'_i))\\
&=
\sum_{\theta'_i \in \Theta_i} \sum_{a'_i \in A_i} x(\theta'_i,a'_i) \left\{ \sum_{a_i \in A_i} M((\theta_i,a_i),(\theta'_i,a'_i)) \right\} \\
&=
\sum_{\theta'_i \in \Theta_i} \left\{ \sum_{a'_i \in A_i} x(\theta'_i,a'_i) \right\} W(\theta_i,\theta'_i) \tag{from the definition of $W(\theta_i,\theta'_i)$} \\
&=
\sum_{\theta'_i \in \Theta_i} W(\theta_i, \theta'_i). \tag{since $x \in \calX$}
\end{align*}

If every entry of $M$ is included in $[0,1]$, then the proof is finished, but this does not necessarily holds.
We show that we can achieve $M \in \calQ$ by shifting each entry of $M$ without changing $Mx$ for every $x \in \calX$.
Fix any $\theta_i \in \Theta_i$ and $a_i \in A_i$, and we focus on the row $M((\theta_i,a_i),(\cdot,\cdot))$.
For any $\theta'_i \in \Theta_i$, we define $a^1_i \in \argmax_{a'_i \in A_i} M((\theta_i,a_i),(\theta'_i,a'_i))$ and $a^2_i \in \argmin_{a'_i \in A_i} M((\theta_i,a_i),(\theta'_i,a'_i))$.
We reuse the above definition of $x^1$ and $x^2$ with these new $a^1_i$ and $a^2_i$.
Since $(Mx^1)(\theta_i,a_i)$ and $(Mx^2)(\theta_i,a_i)$ are included in $[0,1]$, we have
\begin{equation}\label{eq:sum-01}
\max_{a'_i \in A_i} M((\theta_i,a_i),(\theta'_i,a'_i)) - \min_{a'_i \in A_i}M((\theta_i,a_i),(\theta'_i,a'_i)) \in [0,1].
\end{equation}
If we set $x(\theta'_i,a_i) = 1$ for some $a_i \in \argmax_{a'_i \in A_i} M((\theta_i,a_i),(\theta'_i,a'_i))$ for each $\theta_i$ and set $0$ otherwise, $(Mx)(\theta_i,a_i) \in [0,1]$ implies
\begin{equation}\label{eq:maxsum-ub}
\sum_{\theta'_i \in \Theta_i} \max_{a'_i \in A_i}M((\theta_i,a_i),(\theta'_i,a'_i)) \in [0,1].
\end{equation}
Similarly, we can obtain
\begin{equation}\label{eq:minsum-lb}
\sum_{\theta'_i \in \Theta_i} \min_{a'_i \in A_i}M((\theta_i,a_i),(\theta'_i,a'_i)) \in [0,1].
\end{equation}

For any $\theta'_i,\theta''_i \in \Theta_i$, even if we add an arbitrary value $C \in \bbR$ to all entries $M((\theta_i,a_i),(\theta'_i,\cdot))$ and subtract $C$ from all entries $M((\theta_i,a_i),(\theta''_i,\cdot))$, the value of $(Mx)(\theta_i,a_i)$ does not change.
If we denote this shifted matrix by $M'$, this fact can be checked as
\begin{equation*}
(M'x)(\theta_i,a_i)
=
(Mx)(\theta_i,a_i)
+
\sum_{a'_i \in A_i} C x(\theta'_i,a'_i)
-
\sum_{a'_i \in A_i} C x(\theta''_i,a'_i)
=
(Mx)(\theta_i,a_i).
\end{equation*}

Now we consider shifting the entries of $M$ and obtain $M'$.
Fix any $\theta^*_i \in \Theta_i$.
For $\theta'_i \in \Theta_i \setminus \{\theta^*_i\}$, we define all entries $M'((\theta_i,a_i),(\theta'_i,\cdot))$ by
\begin{equation*}
M'((\theta_i,a_i),(\theta'_i,a''_i))
= M((\theta_i,a_i),(\theta'_i,a''_i)) - \min_{a'_i \in A_i} M((\theta_i,a_i),(\theta'_i,a'_i))
\end{equation*}
for each $a''_i \in A_i$.
For $\theta^*_i$, we define
\begin{equation*}
M'((\theta_i,a_i),(\theta^*_i,a''_i))
= M((\theta_i,a_i),(\theta^*_i,a''_i)) + \sum_{\theta'_i \in \Theta_i \setminus \{\theta^*_i\}} \min_{a'_i \in A_i} M((\theta_i,a_i),(\theta'_i,a'_i))
\end{equation*}
for each $a''_i \in A_i$.
Note that the sum of shifts is $0$, and therefore, $Mx = M'x$ for every $x \in \calX$.
We check that all the entries of $M'$ on this row are included in $[0,1]$.
For $\theta'_i \in \Theta_i \setminus \{\theta^*_i\}$, all entries $M'((\theta_i,a_i),(\theta'_i,\cdot))$ are included in $[0,1]$ due to \eqref{eq:sum-01}.
For $\theta^*_i$, the largest entry is
\begin{align*}
\max_{a''_i \in A_i} M'((\theta_i,a_i),(\theta^*_i,a''_i))
&=
\max_{a'_i \in A_i} M((\theta_i,a_i),(\theta^*_i,a'_i)) + \sum_{\theta'_i \in \Theta_i \setminus \{\theta^*_i\}} \min_{a'_i \in A_i} M((\theta_i,a_i),(\theta'_i,a'_i))\\
&\le
\sum_{\theta'_i \in \Theta_i \setminus \{\theta^*_i\}} \max_{a'_i \in A_i} M((\theta_i,a_i),(\theta'_i,a'_i)) \in [0,1]
\end{align*}
due to \eqref{eq:maxsum-ub}.
The smallest entry is
\begin{equation*}
\min_{a''_i \in A_i} M'((\theta_i,a_i),(\theta^*_i,a''_i))
= \sum_{\theta'_i \in \Theta_i} \min_{a'_i \in A_i} M((\theta_i,a_i),(\theta'_i,a'_i)) \in [0,1]
\end{equation*}
due to \eqref{eq:minsum-lb}.
Therefore, all entries of $M'$ on this row are included in $[0,1]$. 

We can apply the same shifting operation to all rows of $M$ and denote the obtained matrix by $Q$.
All the entries of $Q$ are included in $[0,1]$ with $Qx = Mx$ for every $x \in \calX$.
From its definition, $W(\theta_i,\theta'_i)$ is the sum of non-negative values, hence non-negative for each $\theta_i,\theta'_i \in \Theta_i$.
Since $\sum_{\theta'_i \in \Theta_i} W(\theta_i,\theta'_i) = 1$, each $W(\theta_i,\theta'_i)$ is at most $1$.
Therefore, $Q$ satisfies all the constraints for $\calQ$, and we have $Q \in \calQ$.
\end{proof}

By using this proposition, we can reduce linear swap regret minimization to untruthful swap regret minimization as
\begin{equation*}
\RLS^T
= \max_{M \in \calM(\calP)} \sum_{t=1}^T \left( \langle M x^t, \bar{u}^t \rangle - \langle x^t, \bar{u}^t \rangle \right)
= \max_{Q \in \calQ} \sum_{t=1}^T \left( \langle Q x^t, \bar{u}^t \rangle - \langle x^t, \bar{u}^t \rangle \right),
\end{equation*}
which equals $\RUSi^T$ from \Cref{lem:RUSi-Phi}.

\section{On Bayes coarse correlated equilibria}\label{sec:bcce}

Here, we briefly present definitions of Bayes coarse correlated equilibria for comparison with Bayes correlated equilibria.
As with SFCEs, we can define coarse correlated equilibria of the strategic form.
In an equilibrium of this class, each player $i \in N$ does not have incentive to ignore recommendation $s_i \in S_i$ and stick to any strategy $s'_i \in S_i$.

\begin{definition}[Strategic-form coarse correlated equilibria (SFCCEs)]
A distribution $\sigma \in \Delta(S)$ is a strategic-form coarse correlated equilibrium
if
for any $i \in N$ and any $s'_i \in S_i$, it holds that
\begin{equation*}
\E_{\theta \sim \rho} \left[ \E_{s \sim \sigma} \left[ v_i(\theta; s(\theta)) \right] \right]
\ge                              
\E_{\theta \sim \rho} \left[ \E_{s \sim \sigma} \left[ v_i(\theta; s'_i(\theta_i),s_{-i}(\theta_{-i})) \right] \right].
\end{equation*}
\end{definition}

Similarly, we can define coarse correlated equilibria of the agent normal form.
In an equilibrium of this class, each hypothetical player $(i,\theta_i) \in N'$ does not have incentive to stick to any action $a'_i \in A_i$.
\citet*{HST15} proposed dynamics converging to this class of equilibria.

\begin{definition}[Agent-normal-form coarse correlated equilibria (ANFCCEs)]
A distribution $\sigma \in \Delta(S)$ is an agent-normal-form coarse correlated equilibrium
if
for any $i \in N$, $\theta'_i \in \Theta_i$, and $a'_i \in A_i$, it holds that
\begin{equation*}
\E_{\theta \sim \rho} \left[ \bfone_{\{\theta_i = \theta'_i\}} \E_{s \sim \sigma} \left[ v_i(\theta; s(\theta)) \right] \right]
\ge                              
\E_{\theta \sim \rho} \left[ \bfone_{\{\theta_i = \theta'_i\}} \E_{s \sim \sigma} \left[ v_i(\theta; a'_i,s_{-i}(\theta_{-i})) \right] \right].
\end{equation*}
\end{definition}

The incentive constraints for ANFCCEs are imposed separately for each type $\theta'_i$.
For (non-coarse) ANFCEs, the separate constraints are equivalent to the total constraints \eqref{eq:ic-anf} (for the non-approximate version, i.e., $\epsilon = 0$), and the difference between SFCEs and ANFCEs comes from the difference between deviation $\phiSF \colon S_i \to S_i$ and $\phi \colon \Theta_i \times A_i \to A_i$.
For SFCCEs and ANFCCEs, there is no difference in deviations, but the total constraints for SFCEs and the separate constraints for ANFCCEs are not equivalent.
If we take the summation of the incentive constraint for ANFCCEs over type $\theta'_i \in \Theta_i$, we obtain those for SFCCEs by setting $s'_i(\theta'_i) = a'_i$ for each $\theta'_i \in \Theta_i$.
On the other hand, the incentive constraints for SFCCEs do not imply those for ANFCCEs.
Thus, the set of ANFCCEs is a subset of the set of SFCCEs, which makes a striking contrast to the relation of SFCEs and ANFCEs.
A generalization of this relation for extensive-form games with imperfect information was proved by \citet*{FarinaBS20}.
Here, we show this relation is strict even for Bayesian games by presenting a simple example.

\begin{example}
We consider a Bayesian game with two players.
There are two possible types for each player, and their types are completely correlated.
More precisely, let $\Theta_1 = \{\theta_1,\theta'_1\}$ and $\Theta_2 = \{\theta_2,\theta'_2\}$, and $(\theta_1,\theta_2)$ realizes with probability $0.5$ and $(\theta'_1,\theta'_2)$ with probability $0.5$.
Each player has two possible actions $A_1 = \{a_1, a'_1\}$ and $A_2 = \{a_2, a'_2\}$, respectively.
The payoffs for player $1$ with each type are presented in \Cref{tab:cce-example}.
The payoffs for player $2$ are always defined to be $0$.

\begin{table}
\caption{A Bayesian game for which the set of SFCCEs and the set of ANFCCEs are different. The left table is the payoffs for player $1$ with type $\theta_1$, and the right table is the payoffs for player 1 with type $\theta'_1$. The payoff for player $2$ is always $0$.}\label{tab:cce-example}
\centering
\begin{tabular}{cc}
\begin{minipage}{0.3\textwidth}
\centering
\begin{tabular}{c|cc}
	& $a_2$ & $a'_2$\\
	\hline
	$a_1$ & $0$ & $0$ \\
	$a'_1$ & $0.5$ & $0$
\end{tabular}	
\end{minipage}
\medskip
\begin{minipage}{0.3\textwidth}
\centering
\begin{tabular}{c|cc}
	& $a_2$ & $a'_2$\\
	\hline
	$a_1$ & $0$ & $1$ \\
	$a'_1$ & $1$ & $0$
\end{tabular}	
\end{minipage}
\end{tabular}
\end{table}

We consider a distribution $\sigma \in \Delta(S)$ that recommends two possible strategies $s,s'$ each with probability $0.5$.
For $(\theta_1,\theta_2)$, these two strategies are the same: both of $s_1$ and $s_2$ recommend $a_1$ to player $1$ and $a_2$ to player $2$.
For $(\theta'_1,\theta'_2)$, these two strategies are different: $s_1(\theta'_1) = a'_1$ and $s_2(\theta'_2) = a_2$, while $s'_1(\theta'_1) = a_1$ and $s'_2(\theta'_2) = a'_2$.
The expected payoff for player $1$ is $0$ for $(\theta_1,\theta_2)$ and $1$ for $(\theta'_1,\theta'_2)$.
The total expected payoff is $0.5$.

First, we show that $\sigma$ is an SFCCE by considering an optimal deviation for player $1$.
For $\theta_1$, selecting $a'_1$ increases payoff to $0.5$.
On the other hand, for $\theta'_1$, selecting either of $a_1$ or $a'_1$ decreases payoff to $0.5$.
In total, player $1$ cannot increase the expected payoff.
Since the payoff for player $2$ is always $0$, this player never has incentive to deviate.
Therefore, $\sigma$ is an SFCCE.
On the other hand, $\sigma$ is not an ANFCCE since player $1$ with $\theta_1$ can gain by choosing $a'_1$.
\end{example}

\citet*{CKK15} and \citet*{JL23} used the following definition of Bayes coarse correlated equilibria.
This is similar to that of ANFCCEs, but defined on $\Delta(A)^\Theta$, not on $\Delta(S)$.
To our knowledge, this class has not yet been given any specific name.
A possible candidate is ``coarse Bayesian solutions'' since this notion is a coarse variant of Bayesian solutions.

\begin{definition}[coarse Bayesian solutions]
A type-wise distribution $\pi \in \Delta(A)^\Theta$ is a coarse Bayesian solution
if
for any $i \in N$, $\theta'_i \in \Theta_i$, and $a'_i \in A_i$, it holds that
\begin{equation*}
\E_{\theta \sim \rho} \left[ \bfone_{\{\theta_i = \theta'_i\}} \E_{a \sim \pi(\theta)} \left[ v_i(\theta; a) \right] \right]
\ge                              
\E_{\theta \sim \rho} \left[ \bfone_{\{\theta_i = \theta'_i\}} \E_{a \sim \pi(\theta)} \left[ v_i(\theta; a'_i,a_{-i}) \right] \right].
\end{equation*}
\end{definition}

\end{document}